\documentclass[reqno]{amsart}
\usepackage{amsmath}
\usepackage{amsfonts}
\usepackage{mathrsfs}
\usepackage {amssymb,euscript} 
\usepackage {amsmath}
\usepackage {amsthm}
\usepackage {amscd}
\usepackage{geometry}
\usepackage{hyperref}
\usepackage{color}
\usepackage{epsfig}
\usepackage{caption}

\usepackage{slashed}

\usepackage{tikz}

\usepackage{graphics}

\title[Emergence of Apparent Horizon]{Emergence of Apparent Horizon in Gravitational Collapse} 
\date{\today}

\author{Xinliang An}
\address{Department of Mathematics, National University of Singapore, Singapore 119076}
\email{matax@nus.edu.sg}

\theoremstyle{definition}
\newtheorem{lemma}{Lemma}[section]

\newtheorem{proposition}[lemma]{Proposition}
\newtheorem{theorem}[lemma]{Theorem}

\newtheorem{remark}{Remark}

\numberwithin{equation}{section}

\begin{document}

\newcommand{\ub}{\underline{u}}
\newcommand{\Cb}{\underline{C}}
\newcommand{\Lb}{\underline{L}}
\newcommand{\Lh}{\hat{L}}
\newcommand{\Lbh}{\hat{\Lb}}
\newcommand{\phib}{\underline{\phi}}
\newcommand{\Phib}{\underline{\Phi}}
\newcommand{\Db}{\underline{D}}
\newcommand{\Dh}{\hat{D}}
\newcommand{\Dbh}{\hat{\Db}}
\newcommand{\omb}{\underline{\omega}}
\newcommand{\omh}{\hat{\omega}}
\newcommand{\ombh}{\hat{\omb}}
\newcommand{\Pb}{\underline{P}}
\newcommand{\chib}{\underline{\chi}}
\newcommand{\chih}{\hat{\chi}}
\newcommand{\chibh}{\hat{\chib}}

\newcommand{\alb}{\underline{\alpha}}
\newcommand{\zeb}{\underline{\zeta}}
\newcommand{\beb}{\underline{\beta}}
\newcommand{\etb}{\underline{\eta}}
\newcommand{\Mb}{\underline{M}}
\newcommand{\oth}{\hat{\otimes}}

\def\a {\alpha}
\def\b {\beta}
\def\ab {\alphab}
\def\bb {\betab}
\def\nab {\nabla}

\def\ub {\underline{u}}
\def\th {\theta}
\def\Lb {\underline{L}}
\def\Hb {\underline{H}}
\def\chib {\underline{\chi}}
\def\chih {\hat{\chi}}
\def\chibh {\hat{\underline{\chi}}}
\def\omegab {\underline{\omega}}
\def\etab {\underline{\eta}}
\def\betab {\underline{\beta}}
\def\alphab {\underline{\alpha}}
\def\Psib {\underline{\Psi}}
\def\hot{\widehat{\otimes}}
\def\Phib {\underline{\Phi}}
\def\thb {\underline{\theta}}
\def\t {\tilde}
\def\st {\tilde{s}}

%%%%%%%%%
\def\rhoc{\check{\rho}}
\def\sigmac{\check{\sigma}}
\def\Psic{\check{\Psi}}
\def\kappab{\underline{\kappa}}
\def\betabc {\check{\underline{\beta}}}

%%%%%%%%%%%%
%%%%%%%%%%%%
\def\d {\delta}
\def\f {\frac}
\def\i {\infty}
\def\l {\bigg(}
\def\r {\bigg)}
\def\S {S_{u,\underline{u}}}
\def\o{\omega}
\def\O{\Omega}\
\def\be{\begin{equation}\begin{split}}
\def\en{\end{split}\end{equation}}
\def\at{a^{\frac{1}{2}}}
\def\af{a^{\frac{1}{4}}}
\def\od{\omega^{\dagger}}
\def\ombd{\underline{\omega}^{\dagger}}
\def\K{K-\frac{1}{|u|^2}}
\def\ut{\frac{1}{|u|^2}}
\def\s{\frac{\delta a^{\frac{1}{2}}}{|u|}}
\def\Kb{K-\frac{1}{(u+\underline{u})^2}}
\def\bf{b^{\frac{1}{4}}}
\def\bt{b^{\frac{1}{2}}}
%%%%%
\def\de{\delta}
\def\ls{\lesssim}
\def\om{\omega}
\def\Om{\Omega}
\def\lo{\lambda_1}
\def\lt{\lambda_2}
\def\phib{\bar{\phi}}
\def\bR{\bar{R}}
\def\tR{\tilde{R}}

%%%%%%%%%%%%%
%%%%%%%%%%%%%

\newcommand{\e}{\epsilon}
\newcommand{\et} {\frac{\epsilon}{2}}
\newcommand{\ef} {\frac{\epsilon}{4}}
\newcommand{\LH} {L^2(H_u)}
\newcommand{\LHb} {L^2(\underline{H}_{\underline{u}})}
\newcommand{\M} {\mathcal}
\newcommand{\TM} {\tilde{\mathcal}}
\newcommand{\p}{\psi\hspace{1pt}}
\newcommand{\q}{\underline{\psi}\hspace{1pt}}
\newcommand{\Li}{_{L^{\infty}(S_{u,\underline{u}})}}
\newcommand{\Lt}{_{L^{2}(S)}}
\newcommand{\da}{\delta^{-\frac{\epsilon}{2}}}
\newcommand{\db}{\delta^{1-\frac{\epsilon}{2}}}
\newcommand{\D}{\Delta}

%%%%%%%%%%%%
%%%%%%%%%%%%%

\renewcommand{\div}{\mbox{div }}
\newcommand{\curl}{\mbox{curl }}
\newcommand{\trchb}{\mbox{tr} \chib}
\def\trch{\mbox{tr}\chi}
\newcommand{\tr}{\mbox{tr}}

\newcommand{\Ls}{{\mathcal L} \mkern-10mu /\,}
\newcommand{\eps}{{\epsilon} \mkern-8mu /\,}
%%%%%%%%%

\newcommand{\xib}{\underline{\xi}}
\newcommand{\psib}{\underline{\psi}}
\newcommand{\rhob}{\underline{\rho}}
\newcommand{\thetab}{\underline{\theta}}
\newcommand{\gammab}{\underline{\gamma}}
\newcommand{\nub}{\underline{\nu}}
\newcommand{\lb}{\underline{l}}
\newcommand{\mub}{\underline{\mu}}
\newcommand{\Xib}{\underline{\Xi}}
\newcommand{\Thetab}{\underline{\Theta}}
\newcommand{\Lambdab}{\underline{\Lambda}}
\newcommand{\vphb}{\underline{\varphi}}

\newcommand{\ih}{\hat{i}}

\newcommand{\tcL}{\widetilde{\mathscr{L}}}

\newcommand{\sRic}{Ric\mkern-19mu /\,\,\,\,}
\newcommand{\sL}{{\cal L}\mkern-10mu /}
\newcommand{\sLh}{\hat{\sL}}
\newcommand{\sg}{g\mkern-9mu /}
\newcommand{\seps}{\epsilon\mkern-8mu /}
\newcommand{\sd}{d\mkern-10mu /}
\newcommand{\sR}{R\mkern-10mu /}
\newcommand{\snab}{\nabla\mkern-13mu /}
\newcommand{\sdiv}{\mbox{div}\mkern-19mu /\,\,\,\,}
\newcommand{\scurl}{\mbox{curl}\mkern-19mu /\,\,\,\,}
\newcommand{\slap}{\mbox{$\triangle  \mkern-13mu / \,$}}
\newcommand{\sGamma}{\Gamma\mkern-10mu /}
\newcommand{\somega}{\omega\mkern-10mu /}
\newcommand{\somb}{\omb\mkern-10mu /}
\newcommand{\spi}{\pi\mkern-10mu /}
\newcommand{\sJ}{J\mkern-10mu /}
\renewcommand{\sp}{p\mkern-9mu /}
\newcommand{\su}{u\mkern-8mu /}

\dedicatory{Dedicated to my father, Yuqing An}

\begin{abstract}
We solve Einstein vacuum equations in a spacetime region up to the ``center'' of gravitational collapse. Within this region, we construct a sequence of marginally outer trapped surfaces (MOTS) with areas going to zero. These MOTS form a marginally outer trapped tube (apparent horizon). It emerges from a point and is smooth (except at that point). In the proof we employ a scale critical trapped surface formation criterion established by An and Luk and a new type of quasilinear elliptic equation is studied. {\color{black}One of the main conclusions} in this paper proves a conjecture of Ashtekar on black hole thermodynamics. And the spacetimes constructed here could also be viewed as (non-spherically symmetric) generalizations of the well-known Vaidya spacetime. 
\end{abstract}  

\maketitle

\tableofcontents

\section{Introduction}

In this paper, we study gravitational collapse for the $3+1$ dimensional Einstein vacuum equations (EVEs):
\begin{equation}\label{Ei}
\mbox{Ric}_{\mu\nu}=0
\end{equation}
in a strong field regime. 

In modern physics, black hole mechanics plays an important role, which connects the laws of thermodynamics with the dynamical evolution of  the black-hole boundary. Back to 1970s, Bekenstein and Hawking found that the area of black-hole event horizon could be viewed as the entropy of a black hole. People further established four laws of black hole mechanics along the black-hole event horizons. 

In general relativity, there are two types of boundaries for a black hole region: event horizon and apparent horizon.  An event horizon is a boundary in spacetime, beyond which events cannot affect an outside observer (at future null infinity). For defining an event horizon, people need to know the global informations of the whole spacetime. An apparent horizon on the other side is defined locally. {\color{black}It is the boundary, where outward directed light rays sent in the interior region cannot escape to exterior.} When studying gravitational collapse via either pure math or numerical approaches, people are using an apparent horizon to mean the boundary of a black hole region.  
 
A natural research direction is to establish black hole mechanics along an apparent horizon. This was open until 2003. In \cite{AK03, AK},  Ashtekar and Krishnan showed that if a smooth and spacelike apparent horizon exits, then the black hole mechanics along an apparent horizon could be established. However, there is an important step missing: Does a spacelike apparent horizon exist? Can it form dynamically in gravitational collapse? Due to physical intuitions, Ashtekar further conjectured that \textit{the apparent horizon of a ``black hole" region may emerge from a spacetime point and it is spacelike in the beginning}. In current paper, we prove this conjecture of Ashtekar. 

We construct a series of marginally outer trapped surfaces (MOTS) up to the center. These MOTS constitutes the apparent horizon. For an open set of initial data, we show that an apparent horizon is forming dynamically in  evolution. It is smooth and spacelike (except at the center). 

In this paper, we combine techniques from both hyperbolic PDE and quasilinear elliptic equations. With hyperbolic part, we take a limit of a scale-critical result by An-Luk \cite{AL}. This gives the existence of EVEs in a spacetime region up to the center. Then along each incoming null hypersurface, we find the MOTS. Mathematically, these MOTS are solutions to a new type of quasilinear elliptic equations. 

The a priori estimates for these new equations are obtained for the first time. These equations are solved for the first time. And this result is also the first demonstrating formation of apparent horizon in dynamics.

\subsection{{\color{black}Penrose's Incompleteness Theorem}}
One of the central questions in general relativity is formation of singularities for (\ref{Ei}). In 1965, Penrose \cite{Penrose} proved his celebrated incompleteness theorem:
\textit{A spacetime with a non-compact Cauchy hypersurface satisfying EVEs and containing a compact trapped surface is future causally geodesically incomplete. }
Here a trapped surface is a two dimensional sphere, both future null expansions of which are negative. 

Applying Penrose's incompleteness theorem, singularity formation for EVEs is then reduced to trapped surface formation. However, the proof of Penrose's theorem does not guarantee that a trapped surfaces can arise in evolution. This latter problem requires an detailed understanding of the dynamics of EVEs in a large data regime. \\

\subsection{Formation of Trapped Surfaces}
For vacuum spacetimes, this was open for a long time until a recent breakthrough result in 2008. In a 589-page monumental work \cite{Chr:book}, Christodoulou proved that a trapped surface can indeed form dynamically from regular initial data free of trapped surfaces. 

In \cite{Chr:book} a characteristic initial value problem (see the figure below) for EVEs is studied. Initial data are prescribed along a truncated incoming cone $\Hb_0$ and a truncated outgoing cone $H_0$. 
Here 2-sphere $S_{0,0}$ is the intersection of these two cones.

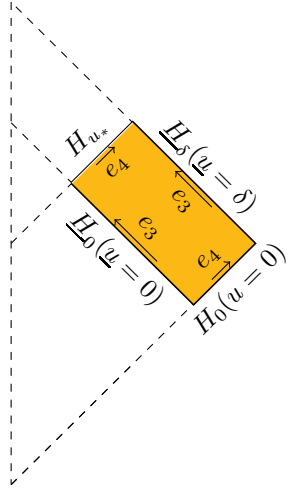
\begin{center}
\begin{tikzpicture}[scale=0.8]

\draw [white](3,-1)-- node[midway, sloped, below,black]{$H_0(u=0)$}(4,0);

\draw [white](2,2)--node [midway,sloped,above,black] {$\Hb_{\delta}(\ub=\delta)$}(4,0);
\draw [white](1,1)--node [midway,sloped, below,black] {$\Hb_{0}(\ub=0)$}(3,-1);
\draw [dashed] (0, 4)--(0, -4);
\draw [dashed] (0, -4)--(4,0)--(0,4);
\draw [dashed] (0,0)--(2,2);

\draw [dashed] (0,-4)--(2,-2);
\draw [dashed] (0,2)--(3,-1);
\draw [very thick] (1,1)--(3,-1)--(4,0)--(2,2)--(1,1);
\fill[yellow!70!red] (1,1)--(3,-1)--(4,0)--(2,2)--(1,1);
\draw [white](1,1)-- node[midway,sloped,above,black]{$H_{u_*}$}(2,2);
\draw [->] (3.3,-0.6)-- node[midway, sloped, above,black]{$e_4$}(3.6,-0.3);
\draw [->] (1.4,1.3)-- node[midway, sloped, below,black]{$e_4$}(1.7,1.6);
\draw [->] (3.3,0.6)-- node[midway, sloped, below,black]{$e_3$}(2.7,1.2);
\draw [->] (2.4,-0.3)-- node[midway, sloped, above,black]{$e_3$}(1.7,0.4);
\end{tikzpicture}
\end{center}

We prescribe Minkowkian initial data along $\Hb_0$ such that the sphere $S_{0,0}$ is the standard $2$-sphere with radius $1$. The initial data on $H_0$ are prescribed in a region $0\leq \ub\leq \delta$. Denote $\chih$ to be the traceless part of the outgoing null second fundamental form $\chi$. And we require $\chih$ to be large in terms of $\delta$. 

This form of initial data was called a ``short pulse data'' in \cite{Chr:book}. According to these short pulse data, Christodoulou constructed a hierarchy of large and small quantities (parametrized by different weights in $\delta$). This hierarchy is later shown to be preserved by the nonlinear evolution. Hence, despite being a large data problem, a long time existence theorem can be established. We have the existence of a solution to EVEs and could control the geometry from initial data to a region where a ``black hole'' is about to form.  As long as the incoming radiation per unit solid angle is bounded uniformly below independent of $\de$, a trapped surface is guaranteed to form in the causal future of the initial data. We summarize
the main theorem\footnote{Christodoulou's original result allows the initial data to be posed at past null infinity. Here, we only mention a version in a finite region.} as follows:

\begin{theorem}[Christodoulou \cite{Chr:book}]\label{Chr.thm}
Consider the characteristic initial value problem for EVEs such that $\Hb_0$ coincides with a backwards light cone\footnote{Here, and in the remainder of this paper, we normalize the $u$ coordinate on the backwards light cone as follows. Let $C=\{(t,x_1,x_2,x_3):t\leq 0,\,t^2=x_1^2+x_2^2+x_3^2\}$ be the backward light cone in Minkowski space emanating from the origin. Define $r=\sqrt{x_1^2+x_2^2+x_3^2}$ and $u=\f12(t-r+2)$. Notice in particular that $u=0$ ($t=-1, r=1$) on a standard sphere of radius $1$ and $u=1$ ($t=0, r=0$) on the vertex.} in Minkowski space for $0\leq u\leq 1$. For every $B>0$ and $u_*\leq 1$, there exists $\de=\de(B,u_*)>0$ sufficiently small such that if the initial $\chih_0$, prescribed on $H_0$ for $0\leq \ub\leq \de$, satisfies
\begin{equation}\label{Chr.upper.bound}
\sum_{i\leq 5,\,j\leq 3}\de^{\frac 12+j}\|\nab_{e_4}^j \nab^i\chih_0\|_{L^\infty_{\ub}L^2(S_{0,\ub})}\leq B, 
\end{equation}
where $e_4$ and $\nab$ are respectively outgoing null vector and angular derivative on a 2-sphere $S_{u,\ub}$, then the solution to EVEs remains regular in $0\leq u\leq u_*$, $0\leq \ub\leq \de$. Moreover, if the initial data also verify the lower bound
\begin{equation}\label{Chr.lower.bound}
\inf_{\o\in \mathbb{S}^2} \int_0^{\de} |\chih_0|^2(\ub',\o)\,d\ub' \geq M_* > 2(1-u_*),
\end{equation}
then, after choosing $\de$ to be small (depending on $B$, $u_*$ and $M_*$), the sphere $S_{u_*,\de}$ is a trapped surface.\footnote{The initial data constructed in \cite{Chr:book} satisfy both \eqref{Chr.upper.bound} and \eqref{Chr.lower.bound} at the same time. Moreover, the initial data can be chosen to obey 
$\inf_{\o\in \mathbb{S}^2} \int_0^{\de} |\chih_0|^2(\ub',\o)\,d\ub'<2.$
Thus for $\de$ sufficiently small, it can be proved that the initial hypersurface $H_1$ does not contain any trapped surfaces.}
\end{theorem}
A trapped surface could be considered as a ``localized black hole''. After \cite{Chr:book}, various extensions of Theorem \ref{Chr.thm} have been achieved.  Interested readers are referred to \cite{An, Dafermos, KR:Scarred, KR:Trapped, LY, L-R:Interaction, R-T, Yu1, Yu2}.  In all these works, trapped surfaces formed are of area $1$ and similar upper and lower bounds as in \cite{Chr:book} are employed.\\

A recent work \cite{KLR} relaxed the lower bound assumption in \cite{Chr:book}:
\begin{theorem}[Klainerman-Luk-Rodnianski \cite{KLR}]
Assume that the initial data for EVEs satisfy the condition \eqref{Chr.upper.bound} in Theorem \ref{Chr.thm}. If the initial data also verify the lower bound
\begin{equation*}
\sup_{\o\in \mathbb{S}^2} \int_0^{\de} |\chih_0|^2(\ub',\o)\,d\ub' \geq M_* > 0,
\end{equation*}
then, after choosing $\de$ small, a compact trapped surface can be guaranteed to formed to the future of the initial data.
\end{theorem}
To prove this result, in \cite{KLR} they deform a ``double null foliation'' and solve a quasilinear elliptic inequality. See also a geometric approach by Le in \cite{Le}. \\

Another recent work \cite{AL} relaxed the upper bound assumption in \cite{Chr:book}: 

\begin{theorem}[An-Luk \cite{AL}] \label{thm1.3}
Consider the following characteristic initial value problem for EVEs: The initial incoming null hypersurface $\Hb_0$ is required to coincide with a backwards light cone in Minkowski space with $0\leq u \leq 1$. Given $\delta$, for every $B$, there exist $a_0=a_0 (B)$ and $b_0=b_0(B)$ sufficiently large. Pick any $a$ and $b$ satisfying $a_0 \leq a \leq \delta^{-1}$ and $b_0\leq b \leq a^{\frac12} \leq \delta^{-\frac12}$.  Along $H_0$, for $0\leq \ub \leq \d$, if the initial $\chih_0$ verifies
\begin{equation}\label{AL.upper.bound}
\sum_{i\leq 5, j\leq 3}\delta^{j} a^{-\frac12}\|\nab^{j}_{e_4}\nab^{i}\chih_{0}\|_{L^{\infty}_{\ub}L^2(S_{0,\ub})}\leq B,\footnote{{\color{black}In the proof of \cite{AL} we don't use the $\nab_{e_4}$ derivative of $\nab^i\chih$. For rougher initial data we could proceed as in \cite{L-R:Propagation} and replace \eqref{AL.upper.bound} with $\sum_{i\leq 5} a^{-\frac12}\|\nab^{i}\chih_{0}\|_{L^{\infty}_{\ub}L^2(S_{0,\ub})}\leq B$.}}
\end{equation}
then the solution to EVEs remains regular in $0\leq u \leq 1-b\delta \at$, $0\leq \ub \leq \d$. Moreover, if the initial data also satisfy a lower bound
\begin{equation*}
\inf_{\o\in \mathbb{S}^2}\int_0^{\delta}|\chih_{0}|^2(\ub', \o)d\ub'\geq 4\delta a,
\end{equation*}
then the sphere $S_{1-\d a, \d}$ is a trapped surface.
\end{theorem}

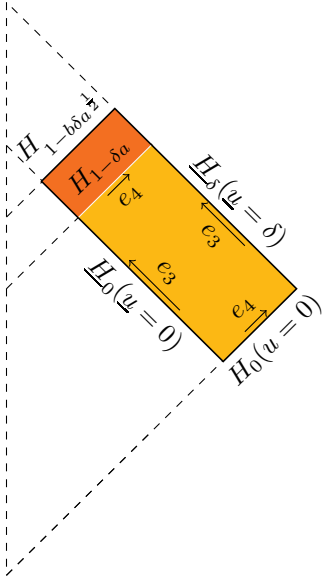
\begin{minipage}[!t]{0.4\textwidth}
\begin{tikzpicture}[scale=0.95]
\draw [white](3,-1)-- node[midway, sloped, below,black]{$H_0(u=0)$}(4,0);

\draw [white](0.5,1.5)-- node[midway,sloped,above,black]{$H_{1-b \delta \at}$}(1.5,2.5);
\draw [white](2,2)--node [midway,sloped,above,black] {$\Hb_{\delta}(\ub=\delta)$}(4,0);
\draw [white](1,1)--node [midway,sloped, below,black] {$\Hb_{0}(\ub=0)$}(3,-1);
\draw [dashed] (0, 4)--(0, -4);
\draw [dashed] (0, -4)--(4,0)--(0,4);
\draw [dashed] (0,0)--(2,2);
\draw [dashed] (0,1)--(1.5,2.5);
\draw [dashed] (0,-4)--(2,-2);
\draw [dashed] (0,2)--(3,-1);
\draw [very thick] (1,1)--(3,-1)--(4,0)--(2,2)--(1,1);
\draw [very thick] (1,1)--(0.5,1.5)--(1.5,2.5)--(2,2)--(1,1);
\fill[yellow!70!red] (1,1)--(3,-1)--(4,0)--(2,2)--(1,1);
\fill[yellow!30!red](1,1)--(0.5,1.5)--(1.5,2.5)--(2,2)--(1,1);
\draw [white](1,1)-- node[midway,sloped,above,black]{$H_{1-\delta a}$}(2,2);
\draw [->] (3.3,-0.6)-- node[midway, sloped, above,black]{$e_4$}(3.6,-0.3);
\draw [->] (1.4,1.3)-- node[midway, sloped, below,black]{$e_4$}(1.7,1.6);
\draw [->] (3.3,0.6)-- node[midway, sloped, below,black]{$e_3$}(2.7,1.2);
\draw [->] (2.4,-0.3)-- node[midway, sloped, above,black]{$e_3$}(1.7,0.4);
\end{tikzpicture}
\end{minipage}
\begin{minipage}[!t]{0.58\textwidth}
Note that the radius of $\S$ is of size $1-u$.\\

In \cite{AL} when setting $a=\d^{-1}$, we are back to the main theorem in \cite{Chr:book}.\\

As explained in \cite{AL}, the characteristic initial data in this theorem are of the following sizes:
$$H^{\f32}\, \mbox{norm} \sim a^{\f12},  \quad H^1\, \mbox{norm} \sim \d^{\f12} a^{\f12}.$$
For Christodoulou's initial data in \cite{Chr:book} ($a=\d^{-1}$):
$$H^{\f32}\, \mbox{norm} \sim{\d}^{-\f12}, \quad H^1\, \mbox{norm} \sim1.$$

Moreover, here $a$ could be chosen as a large universal constant (independent of $\d$).\\

The smaller the initial data, the harder to form a trapped surface. And $H^{\f32}$ norm is the critical norm for EVEs in $\mathbb{R}^{3+1}$.  We hence call our improved trapped surfaces formation criterion as \textit{a scale critical trapped surfaces formation criterion}.\\
\end{minipage}
\hspace{0.05\textwidth}

Theorem \ref{thm1.3} also answers a natural question in general relativity: \textit{to form a trapped surface, what is the least size of initial data?} Let $a$ be a large universal constant. In \cite{AL} we have constructed smooth characteristic initial data, for metric the $H^{\f32}$ norm of which is of size $\at$. \footnote{By definition $\chih$ is essentially $\partial_{\ub}g$ and it is of size $\at/|u|$. By dimensional analysis $\partial_{\ub}$ is of size $\d^{-1}$ and $\partial^{\f12}_{\ub}\thicksim\d^{-\f12}$, we have
$$\int_{H^{(0,\d)}_{u=0}}|\partial_{\ub} g|^2 =\int_0^{\d}\int_{S_{0,\ub'}}|\partial_{\ub} g(u,\ub')|^2 d\theta^1 d\theta^2 d \ub'\approx \d \at \at \approx \d a,$$
$$\int_{H^{(0,\d)}_{u=0}}|\partial^{\f32}_{\ub} g|^2 =\int_0^{\d}\int_{S_{0,\ub'}}|\partial^{\f12}_{\ub}\partial_{\ub} g(u,\ub')|^2 d\theta^1 d\theta^2 d \ub'\approx \d \d^{-\f12}\at \d^{-\f12}\at \approx a.$$ 
Since trivial initial data are prescribed along $\Hb_0$, the characteristic initial data are then of the following sizes:
$$H^{\f32}\, \mbox{norm} \sim a^{\f12},  \quad H^1\, \mbox{norm} \sim \d^{\f12} a^{\f12}, \quad H^{s}\, \mbox{norm} \sim \d^{\f{3-2s}{2}} a^{\f12}.$$} And these initial data lead to trapped surface formation in future.  \\

{\color{black} With the same $\d$ (length of short-pulse regions in \cite{Chr:book} and \cite{AL}), the energy input in \cite{AL} is much smaller and we expect to form a much smaller trapped surface.} The trapped surface formed in \cite{Chr:book} is of radius 1, while in \cite{AL} it is of radius $\delta a$. Hence, an additional difficulty arising for \cite{AL} is to control the behavior of solutions to EVEs in the region close to the vertex: $u=1$ and $\ub=0$.  

Instead of decaying, all the geometric quantities grow. In \cite{AL}, we introduce weighted estimates (as to study the decay rates) to carefully track and derive these growth rates.  \\

\subsection{Natural Questions}

In all results above, the boundary of the trapped region (apparent horizon) is not studied. The next natural questions are to find the apparent horizon and to show that it is emerging from a spacetime point. In modern physics, these are long-existing hypotheses. Given such a boundary, Ashtekar and Krishnan (\cite{AK03}, \cite{AK}) established the laws of black hole mechanics along it. 

In this paper, we address these two questions and give positive answers:
\begin{enumerate}
\item A new type of quasilinear elliptic equations is solved. Its solutions are corresponding to marginally outer trapped surface (MOTS) along an incoming null hypersurface. 
\item This equation is different from Jang's equation \footnote{Along a spacelike hypersurface, MOTS locates at the blow-up points for solutions to Jang's equation. In this paper,  along an incoming null hypersurface a new quasilinear elliptic equation for MOTS is solved. Its solutions remain regular.}. Its a priori estimates are established for the first time.
\item This is the first result proving that an apparent horizon is formed in dynamics. With a new idea to take the limit of a scale-critical result,  this apparent horizon is shown to emerge from a spacetime point and spacelike.
\end{enumerate} 
This paper thus proves the long-existing hypotheses about spacelike apparent horizon (dynamical horizon) in modern physics. The second law of black mechanics is verified in Section 11. These new spacetimes could be used as the test bed for exploring more physical thoughts.

\subsection{Marginally Outer Trapped Surface (MOTS)}

Fix $\ub\in [0, \d]$. Along each incoming null hypersurface $\Hb_{\ub}$, the boundary of the trapped region is called a marginally outer trapped surface (MOTS): let $\tr\chi'$ and $\tr\chib'$ denote respectively the outgoing null expansion and incoming null expansion of a 2-sphere $S$ with coordinates $(u,\ub,\o)=(1-R(\ub, \o), \ub, \o)$; on $S$ if conditions $\tr\chi'=0$ and $\tr\chib'<0$ hold pointwisely, then this $S$ is called a marginally outer trapped surface (MOTS).\\

The first theorem of this paper shows the existence of MOTS along each $\Hb_{\ub}$: 
\begin{theorem}\label{thm1.6}
For EVEs, prescribe characteristic initial data as in Theorem \ref{thm1.3} satisfying the condition \eqref{AL.upper.bound}. If we also have
\begin{equation*}
 \int_0^{\ub}|\chih_0|^2 (0, \ub', \o)d\ub'=f(\ub, \o)\ub a \quad  \mbox{for each}  \quad b\d a^{-\f12} \leq \ub\leq \d,
\end{equation*}
where $f(\ub, \o)$ is a smooth function with properties $20/21\leq f(\ub, \o)\leq 22/21$ and $|\partial^i_{\o} f(\ub, \o)|\lesssim 1$ for all $i\in \mathbb{N}$ and $\o\in \mathbb{S}^2$, then, along every $\Hb_{\ub}$ ($b\d a^{-\f12} \leq \ub\leq \d$), there exists a unique MOTS $M_{\ub}$. Moreover, if requiring
\begin{equation*}
\sum_{0\leq i< \infty, 0 \leq j<\infty}\delta^{j} a^{-\frac12}\|\nab^{j}_{e_4}\nab^{i}\chih_{0}\|_{L^{\infty}_{\ub}L^2(S_{0,\ub})}\leq B,
\end{equation*}
then for different $\ub$, $\{M_{\ub}\}$ form a 3-dimensional smooth hypersurface (dashed line in the colored region). We name this 3-dimensional hypersurface a marginally outer trapped tube.  
\end{theorem}

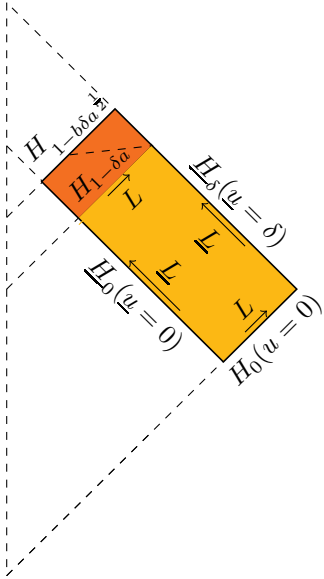
\begin{minipage}[!t]{0.4\textwidth}
\begin{tikzpicture}[scale=0.95]
\draw [white](3,-1)-- node[midway, sloped, below,black]{$H_0(u=0)$}(4,0);

\draw [white](0.6,1.5)-- node[midway,sloped,above,black]{$H_{1-b \delta \at}$}(1.6,2.5);
\draw [white](2,2)--node [midway,sloped,above,black] {$\Hb_{\delta}(\ub=\delta)$}(4,0);
\draw [white](1,1)--node [midway,sloped, below,black] {$\Hb_{0}(\ub=0)$}(3,-1);
\draw [dashed] (0, 4)--(0, -4);
\draw [dashed] (0, -4)--(4,0)--(0,4);
\draw [dashed] (0,0)--(2,2);
\draw [dashed] (0,1)--(1.5,2.5);
\draw [dashed] (0,-4)--(2,-2);
\draw [dashed] (0,2)--(3,-1);
\draw [very thick] (1,1)--(3,-1)--(4,0)--(2,2)--(1,1);
\draw [very thick] (1,1)--(0.5,1.5)--(1.5,2.5)--(2,2)--(1,1);
\fill[yellow!70!red] (1,1)--(3,-1)--(4,0)--(2,2)--(1,1);
\fill[yellow!30!red](1,1)--(0.5,1.5)--(1.5,2.5)--(2,2)--(1,1);
\draw [yellow!70!red](1,0.9)-- node[midway,sloped,above,black]{$H_{1-\delta a}$}(2,1.9);
\draw [->] (3.3,-0.6)-- node[midway, sloped, above,black]{$L$}(3.6,-0.3);
\draw [->] (1.4,1.3)-- node[midway, sloped, below,black]{$L$}(1.7,1.6);
\draw [->] (3.3,0.6)-- node[midway, sloped, below,black]{$\Lb$}(2.7,1.2);
\draw [->] (2.4,-0.3)-- node[midway, sloped, above,black]{$\Lb$}(1.7,0.4);
\draw [dashed] (0.875,1.875)--(2,2);
\end{tikzpicture}
\end{minipage}
\begin{minipage}[!t]{0.58\textwidth}
In this paper, we also call the marginally outer trapped tube defined above an apparent horizon. \\

In \cite{AL}, setting $1\ll b \leq a^{\f12} \leq \d^{-\f12}$, together with the upper bound  
$\sum_{i\leq 7}\|\nab^i \chih_0\|_{L^{\infty}_{\ub}L^2(S_{0,\ub})}\leq \at,$
we have existence to EVEs up to {\color{black} $u=1-b\d \at$. Here $S_{1-\d a, \d}$ is already trapped.}\\ 

If choosing $1\ll b\leq \at=\d^{-\f12}$, we have an apparent horizon (the dashed line) connecting two MOTS of radius $b\d^{\f12}$ and $1$. \\

Here we have a unique MOTS along each incoming null hypersurface $\Hb_{\ub}$, as a consequence the MOTS we construct is outermost and stable.    
\end{minipage}
\hspace{0.05\textwidth}

\subsection{Emergence of Apparent Horizon}
Since the results in \cite{AL} is scaling-critical (the result is independent of $\d$),  we could keep $a$ as a universal large constant and let $\ub\rightarrow 0$, together with Theorem \ref{thm1.6}, we have

\begin{theorem}[main theorem]\label{thm1.8}
Consider the following characteristic initial value problem for EVEs:  The initial incoming null hypersurface $\Hb_0$ is required to coincide with a backwards light cone in Minkowski space with $0\leq u \leq 1$. Given $\delta$, for every $B$, there exist $a_0=a_0 (B)$ and $b_0=b_0(B)$ sufficiently large. Pick any $a$ and $b$ satisfying $a_0 \leq a \leq \delta^{-1}$ and $b_0\leq b \leq a^{\frac12} \leq \delta^{-\frac12}$. And along $H_0$, for $0\leq \ub \leq \delta$ we prescribe $\chih_0$ such that \begin{equation*}
\sum_{i\leq 5,  j\leq 3} a^{-\frac12}\|{{\color{black}\ub}}^{j}\nab^{j}_{e_4}\nab^{i}\chih_{0}\|_{L^{\infty}_{\ub}L^2(S_{0,\ub})}\leq B,
\end{equation*}
Then EVEs (\ref{Ei}) admits a unique solution in the region: $0\leq \ub \leq \d$ and $0\leq u \leq 1-b\ub \at$. Moreover, if requiring
\begin{equation*}
 \int_0^{\ub}|\chih_0|^2 (\ub', \o)d\ub'=f(\ub, \o)\ub a \quad  \mbox{for each}  \quad 0<\ub\leq \d,
\end{equation*}
with smooth function $f(\ub,\o)$ satisfying $20/21\leq f(\ub, \o)\leq 22/21$ and $|\partial^i_{\o} f(\ub, \o)|\lesssim 1$ for all $i\in \mathbb{N}$ and $\o\in\mathbb{S}^2$, then along each $\Hb_{\ub}$ $(0< \ub \leq \d)$, there exists a unique MOTS $M_{\ub}$. Furthermore, if requiring
\begin{equation*}
\sum_{0\leq i< \infty, 0 \leq j<\infty} a^{-\frac12}\|{{\color{black}\ub}}^{j}\nab^{j}_{e_4}\nab^{i}\chih_{0}\|_{L^{\infty}_{\ub}L^2(S_{1,\ub})}\leq B,
\end{equation*}
then except at the vertex where $u=1$ and $\ub=0$, the collection of MOTS $\{M_{\ub}\}$ forms {\color{black} a smooth marginally outer trapped tube (apparent horizon).  In Section \ref{section DH} under more restrictive initial-data condition,  we prove that the apparent horizon is spacelike.} \\
\end{theorem}

{\color{black}Under the initial condition in Section \ref{section DH}}, for the area of each $M_{\ub}$: $\mbox{Area}(M_{\ub})$, in Section \ref{entropy} we further prove
\begin{theorem}\label{thm entropy}
$$\lim_{\ub\rightarrow 0}\mbox{Area}(M_{\ub})=0 \quad \quad \mbox{and} \quad \quad \mbox{Area}(M_{\ub'})>\mbox{Area}(M_{\ub}) \quad \mbox{for} \quad \ub'>\ub.$$
\end{theorem} 

In \cite{AK03, AK} Ashtekar and Krishnan studied black hole mechanics along an spacelike apparent horizon (dynamical horizon). Inspired by their works, we call Theorem \ref{thm entropy} \textit{the area law} along the dynamical horizon.

\begin{remark} In order to go further towards the singularity, here we introduce \textit{a limiting argument}:

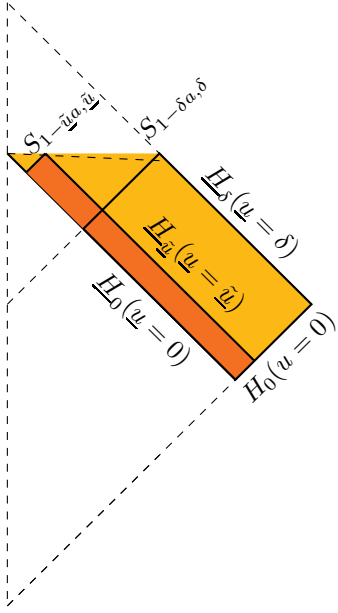
\begin{minipage}[!t]{0.4\textwidth}
\begin{tikzpicture}[scale=1]
\draw [white](0.6,1.9)--node [midway,sloped, above,black] {$S_{1-\tilde{\ub} a, \tilde{\ub}}$}(1.2, 2.5);
\fill[yellow!70!red](0,2)--(1,1)--(2,2)--(0,2);
%\draw (1.5,-0.5) node[very near end, sloped,below]{$H_{u_{\infty}}(u=u_{\infty})$}--(2,0); 
\draw [white](3,-1)-- node[midway, sloped, below,black]{$H_0(u=0)$}(4,0);
\draw [white](2,2)--node [midway,sloped, above,black] {$S_{1-\delta a, \delta}$}(2.8,2.8);
\draw [white](2,2)--node [midway,sloped,above,black] {$\Hb_{\delta}(\ub=\delta)$}(4,0);
\draw [white](1,1)--node [midway,sloped, below,black] {$\Hb_{0}(\ub=0)$}(3,-1);
\draw [dashed] (0, 4)--(0, -4);
\draw [dashed] (0, -4)--(4,0)--(0,4);
\draw [dashed] (0,0)--(2,2);
%\draw [dashed] (0,1)--(1.5,2.5);
\draw [dashed] (0,-4)--(2,-2);
\draw [thick] (0,2)--(3,-1);
%\draw [thick] (1,1)--(0.5,1.5)--(1.5,2.5)--(2,2)--(1,1);
\fill[yellow!70!red] (1,1)--(3,-1)--(4,0)--(2,2)--(1,1);
%\fill[yellow!30!red] (1, 2)--(3.5, -0.5)--(3,-1)--(0.5,1.5)--(1,2); %Orange
\fill[yellow!30!red] (0.5, 2)--(3.25,-0.75)--(3,-1)--(0.25, 1.75)--(0.5, 2);
%\fill[yellow!30!red](1,1)--(0.5,1.5)--(1.5,2.5)--(2,2)--(1,1); %Orange
%\fill[red!40!](1, 2)--(0,2)--(0.5, 1.5)--(1, 2); %Pink
\draw [white] (1.2, 1.3)--node [midway,sloped,above,black] {$\Hb_{\tilde{\ub}}(\ub=\tilde{\ub})$}(3.25,-0.75); 

\draw [thick] (1,1)--(3,-1)--(4,0)--(2,2)--(1,1);

%\draw [thick] (1, 2)--(3.5, -0.5);%middle upper long
%\draw [thick] (0.5,1.5)--(1,2); %middle upper short 
\draw [thick] (0.5, 2)--(3.25,-0.75); %thin upper hong
\draw [thick] (0.25, 1.75)--(0.5, 2); %thin upper short 
\draw[dashed] (0, 2)--(2.1, 1.9);
\end{tikzpicture}
\end{minipage}
\begin{minipage}[!t]{0.5\textwidth}
The proof in \cite{AL} depends only on the largeness of $a, b$ and is independent of $\d$. With characteristic initial data as in Theorem \ref{thm1.8}, for each $\tilde{\ub}\in [0,\d]$ we have $b\leq \at \leq \d^{-\f12}\leq \tilde{\ub}^{-\f12}$. Viewing $\tilde{\ub}$ as the new $\delta$ in Theorem \ref{thm1.3}, we then obtain the existence of EVEs in the region $0\leq u \leq 1-b\tilde{\ub} \at$ and $0\leq \ub\leq \tilde{\ub}$. Let $\tilde{\ub}\rightarrow 0$. We have the solution to EVEs in the whole colored region. \\

The dashed line contained in the colored region is an apparent horizon up to the vertex: $u=1$ and $\ub=0$. 
\end{minipage}
\hspace{0.05\textwidth}  

\end{remark}

\begin{remark}
For the characteristic initial data in Theorem \ref{thm1.8}, there is a jump discontinuity for $\chih_0$: for $\ub\leq 0$, $\chih_0 (\ub, \o)=0$; while for $\ub>0$, $\chih_0\approx \at$. To be more precise, here we prove the existence of EVEs in the whole colored region by combining the local well-posedness result established by Luk and Rodnianski in \cite{L-R:Propagation} and the a priori estimates derived by An and Luk in \cite{AL}. And the spacetimes constructed here could be viewed as generalizations of the well known Vaidya spacetime, which is a spherically symmetric solution (with jump discontinuity) to Einstein-dust system.  
\end{remark}

\begin{remark}
In this paper, an apparent horizon is defined to be the collection of MOTSs $M_{\ub}$, and each $M_{\ub}$ is required to lay on $\Hb_{\ub}$. According to Theorem \ref{thm1.8}, on each $\Hb_{\ub}$ there exists a unique MOTS. In this sense, we also have a uniqueness result for the apparent horizon (defined in this way). 
\end{remark}
 
\begin{remark}
By a similar argument, we also have the following result:
\begin{theorem}\label{thm1.9}
Consider the following characteristic initial value problem for EVEs:  The initial incoming null hypersurface $\Hb_0$ is required to coincide with a backwards light cone in Minkowski space with $0\leq u \leq 1$. Given $\delta$, for every $B$, there exist $a_0=a_0 (B)$ and $b_0=b_0(B)$ sufficiently large. Pick any $a$ and $b$ satisfying $a_0 \leq a \leq \delta^{-1}$ and $b_0\leq b \leq a^{\frac12} \leq \delta^{-\frac12}$. And along $H_0$, for $0\leq \ub \leq 2\delta$ we prescribe $\chih_0$ such that \begin{equation*}
\sum_{i\leq 5,  j\leq 3} a^{-\frac12}\|{{\color{black}\ub}}^{j}\nab^{j}_{e_4}\nab^{i}\chih_{0}\|_{L^{\infty}_{\ub}L^2(S_{0,\ub})}\leq B,
\end{equation*}
Then EVEs (\ref{Ei}) admits a unique solution in the region: $0\leq \ub \leq 2\d$ and $0\leq u \leq 1-b\ub \at$. Moreover, if requiring
\begin{equation*}
 \begin{cases}
 \int_0^{\ub}|\chih_0|^2 (\ub', \o)d\ub'=f(\ub, \o)\ub a \quad  \mbox{for each}  \quad 0<\ub\leq \d,\\
 \int_{\d}^{\ub}|\chih_0|^2 (\ub', \o)d\ub'=0 \quad \quad \quad \quad \, \, \, \mbox{for each} \quad \d\leq \ub <2\d,
\end{cases}
\end{equation*}
with smooth function $f(\ub,\o)$ satisfying $20/21\leq f(\ub, \o)\leq 22/21$ and $|\partial^i_{\o} f(\ub, \o)|\lesssim 1$ for all $i\in \mathbb{N}$ and $\o\in\mathbb{S}^2$, then along each $\Hb_{\ub}$ $(0< \ub \leq 2\d)$, there exists a unique MOTS $M_{\ub}$. And all the MOTSs $\{M_{\ub}\}_{0\leq \ub \leq 2\d}$ form a 3-dimensional piecewise smooth apparent horizon. \\

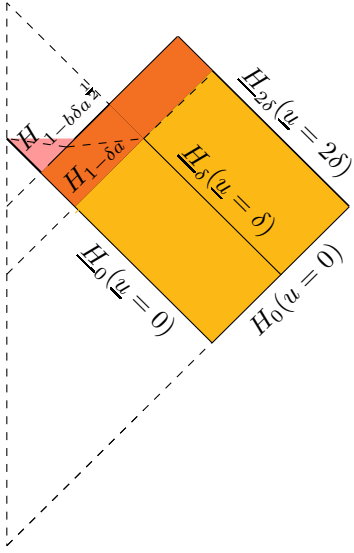
\begin{minipage}[!t]{0.4\textwidth}
\begin{tikzpicture}[scale=0.9]
%\draw (1.5,-0.5) node[very near end, sloped,below]{$H_{u_{\infty}}(u=u_{\infty})$}--(2,0);
\draw [white](3,-1)-- node[midway, sloped, below,black]{$H_0(u=0)$}(5,1);
\draw [white](2,2)--node [midway,sloped,above,black] {$\Hb_{\delta}(\ub=\delta)$}(4,0);
\draw [white](1,1)--node [midway,sloped, below,black] {$\Hb_{0}(\ub=0)$}(3,-1);
\draw [dashed] (0, 4)--(0, -4);
\draw [dashed] (0, -4)--(4,0)--(0,4);
\draw [dashed] (0,0)--(2,2);
\draw [dashed] (0,1)--(1.5,2.5);
\draw [dashed] (0,-4)--(2,-2);
\draw [dashed] (0,2)--(3,-1);
\draw [thick] (1,1)--(3,-1)--(4,0)--(2,2)--(1,1);
\draw [thick] (1,1)--(0.5,1.5)--(1.5,2.5)--(2,2)--(1,1);
\fill[yellow!70!red] (1,1)--(3,-1)--(4,0)--(2,2)--(1,1);
\fill[yellow!30!red](1,1)--(0.5,1.5)--(1.5,2.5)--(2,2)--(1,1);
\fill[red!40!](1, 2)--(0,2)--(0.5, 1.5)--(1, 2);
\draw [yellow!70!red](1,0.9)-- node[midway,sloped,above,black]{$H_{1-\delta a}$}(2,1.9);
%\draw [->] (3.3,-0.6)-- node[midway, sloped, above,black]{$L$}(3.6,-0.3);
%\draw [->] (1.4,1.3)-- node[midway, sloped, below,black]{$L$}(1.7,1.6);
%\draw [->] (3.3,0.6)-- node[midway, sloped, below,black]{$\Lb$}(2.7,1.2);
%\draw [->] (2.4,-0.3)-- node[midway, sloped, above,black]{$\Lb$}(1.7,0.4);
\draw [dashed] (0.875,1.875)--(2,2);
\draw[thick] (0,2)--(0.5, 1.5);
\draw [dashed] (0.875, 1.875)--(0,2);
\draw[thick] (1.5, 2.5)--(2.5, 3.5);
\draw[thick] (2.5, 3.5)--(3,3);
\draw[thick] (2,2)--(3,3);
\draw[thick] (3,3)--(5,1);
\draw[thick] (5,1)--(4,0);
\fill[yellow!30!red](1.5,2.5)--(2.5,3.5)--(3,3)--(2,2)--(1.5,2.5);
\fill[yellow!70!red] (2,2)--(3,3)--(5,1)--(4,0)--(2,2);
\draw[dashed] (2,2)--(3.02, 2.98);
\draw [thin](2,2)--node [midway,sloped,above,black] {$\Hb_{\delta}(\ub=\delta)$}(4,0);
\draw [thin](3,3)--node [midway,sloped,above,black] {$\Hb_{2\delta}(\ub=2\delta)$}(5,1);
\draw [thin](0.6,1.6)-- node[midway,sloped,above,black]{$H_{1-b \delta \at}$}(1.5,2.5);
%\draw [thick](0.5, 1.5)--(1, 2);
\end{tikzpicture}
\end{minipage}
\begin{minipage}[!t]{0.5\textwidth}
{\color{black} Note: the dashed line in the shadowed region is the 3-dimensional apparent horizon. \\

For $0 \leq \ub \leq \d$, it is emerging from the ``center" and expands as ``short-pulse'' coming in. \\

For $\d\leq \ub \leq 2\d$, after switching off the ``short pulse'', intuitively this new piece of apparent horizon will be tilted toward outgoing null direction.}
\end{minipage}

\end{theorem}

\end{remark}

\subsection{Strategy of the Proof}

\subsubsection{Deformation of Foliations}

In the original double null foliation, on $S_{1-\ub a, \ub}$ the outgoing null expansion ($\tr\chi|_{S_{1-\ub a, \ub}}$) is negative pointwisely. Hence $S_{1-\ub a, \ub}$ is a trapped surface.  

To find a 2-sphere $M_{\ub}$, where the  
outgoing null expansion vanishes for all points (thus $M_{\ub}$ is a marginally outer trapped surface). We deform the foliation on $\underline{H}_{\ub}$: $\{(u, \ub, \o): u=1-R(\o)\}$. 
Adapted to this set, we use the following null frames \footnote{ \color{black}{Here $\{e_1, e_2, e_3, e_4\}$ and $\{e'_1, e'_2, e'_3, e'_4\}$ are two sets of null tetrads (moving frames). $\{e_1, e_2\}$ and $\{e'_1, e'_2\}$ are tangent to some spacelike $2$-spheres. And $\{e_3, e_4\},\,\,\{e'_3, e'_4\}$ are corresponding null pairs. See details in Section  \ref{sec.setting}-\ref{deformation formula}. }}
\begin{equation}\label{nfr}
e'_3=e_3, \quad e'_a=e_a-\O e_a(R)e_3, \quad e'_4=e_4-2\O e^a(R) e_a+\O^2 |\nab R|^2 e_3. 
\end{equation}
By definition $e_3(u)=\O^{-1}$. We thus have $e'_a(u+R-1)=e_a(R)-e_a(R)\O e_3(u)=0$. Here $e_3$ is orthogonal to any vector tangent to $\Hb_{\ub}$, thus we can easily check
\begin{equation*}
g(e'_a, e'_b)=g(e_a, e_b)=\d_{ab}, \quad g(e'_4, e'_a)=g(e'_4, e'_4)=0, \quad g(e'_3, e'_4)=-2.
\end{equation*}
We then compute the null expansion $\tr\chi'$ for these new frames and get (see Proposition \ref{1st deformation} and also \cite{KLR} ) 
\begin{equation}\label{deformation equation}
\begin{split}
\tr\chi'=&\tr\chi-2\O\D R(\o)-4\O\eta_a\nab^a R(\o)-4\O^2\chibh_{bc}\nab^b R(\o) \nab^c R(\o)\\
&-\O^2 \tr\chib |\nab R(\o)|^2-8\O^2\omb|\nab R(\o)|^2.
\end{split}
\end{equation}
For a fixed $\o\in\mathbb{S}^2$, $\D$ and $\nab$ are Laplacian operator and angular derivative to the original double null foliation on $S_{1-R(\o), \ub}$ at the point $(1-R(\o), \ub, \o)$. And 
$$\eta_a:=-\f12 g(D_{e_3}e_a, e_4), \quad \chib_{bc}:=g(D_{e_b}e_3, e_c), \quad \omb:=-\f14 g(D_{e_3}e_4, e_3).$$
Here $D$ is the covariant derivative of 4 dimensional spacetimes; $\chibh$ and $\tr\chib$ are respectively the traceless and trace parts of $\chib$.  

\begin{remark}
In \cite{KLR}, Klainerman, Luk and Rodnianski derived (\ref{deformation equation}) and first solved the quasilinear elliptic inequality
\begin{equation*}
\begin{split}
\tr\chi'=&\tr\chi-2\O\D R(\o)-4\O\eta_a\nab^a R(\o)-4\O^2\chibh_{bc}\nab^b R(\o) \nab^c R(\o)\\
&-\O^2 \tr\chib |\nab R(\o)|^2-8\O^2\omb|\nab R(\o)|^2\\
<&0.
\end{split}
\end{equation*}
They thus gave \textit{a fully anisotropic mechanism for formation of trapped surfaces in vacuum}. 
\end{remark}

\subsubsection{A Quasilinear Elliptic Equation} 

For the purpose of a priori estimates, we rewrite $\D R$ with $\D'_{M_{\ub}} R$. $\D'_{M_{\ub}}$ is the Laplacian operator on the unknown sphere $M_{\ub}$.
Proposition \ref{change laplacian} will give
\begin{equation*}
\begin{split}
\D R=&e^a e_a (R)\\
=&\D'_{M_{\ub}} R-2\O \chibh_{ab}\nab^a R \nab^b R.
\end{split}
\end{equation*}
Then $\tr\chi'=0$ is equivalent to a quasilinear elliptic equation: 
\begin{equation}\label{1.12}
\begin{split}
&\D'_{M_{\ub}} R(\o)+2\eta_b\nab^b R(\o)+\f12 \O \tr\chib |\nab R(\o)|^2\\
&+4\O \omb |\nab R(\o)|^2-\f{\O^{-1}}{2} \tr\chi\\
=0.
\end{split}
\end{equation}
\begin{remark}
Note that $\nab'_{M_{\ub}} R=e_a'(R)=e_a(R)=\nab R.$
For simplicity, we also use $\nab R$ to stand for $\nab'_{M_{\ub}} R$. 
\end{remark}

\begin{remark}
In (\ref{1.12}), $\eta_b, \O, \tr\chib, \omb, \tr\chi$ are geometric quantities respect to the original double null foliation and take values at point $(1-R(\o), \ub, \o)$. 
\end{remark}

\begin{remark}
We notice that the induced metric $g'_{M_{\ub}}$ on $M_{\ub}$ depends on $R(\o)$, $\o$ but not on $\nab R(\o)$. We come to this conclusion because the induced metric on 2-sphere $M_{\ub}$, i.e., $u=1-R(\theta_1, \theta_2, \ub)$ along $\underline{H}_{\ub}$ is  
$$g'_{\theta_i \theta_j}=g_{\theta_i \theta_j}+\f{\partial (1-R)}{\partial \theta_i} \f{\partial (1-R)}{\partial \theta_j}\cdot g(\f{\partial}{\partial u}, \f{\partial}{\partial u})=g_{\theta_i \theta_j}.$$
Here we employ the property of double null foliation: $\f{\partial}{\partial u}$ is null and $g(\f{\partial}{\partial u}, \f{\partial}{\partial u})=0$. Hence,
$$g'_{M_{\ub}}(1-R(\o), \ub, \o)=g(1-R(\o), \ub, \o).$$
And equation (\ref{1.12}) is a quasilinear elliptic equation. 
\end{remark}

\subsubsection{A priori Estimates}
Fix $\underline{H}_{\ub}$. Denote $M_0(\o):=\int_0^{\ub} |\chih_0|^2 (0, \ub', \o)d\ub'$. 
By the assumptions of Theorem \ref{thm1.8} $$M_0(\o)=\ub a f(\ub, \o), \quad \mbox{with} \quad \f{20}{21}\leq f(\ub, \o) \leq \f{22}{21}.$$
From the estimates in \cite{AL}, we have
$$\tr\chi=\f{2}{R(\o)}-\f{M_0(\o)}{R(\o)^2}+l.o.t., \quad \tr\chib=-\f{2}{R(\o)}+l.o.t., \quad \O=1+l.o.t..$$
Also treat $\eta_b$ and $\omb$ as lower order terms, equation (\ref{1.12}) is thus transfered to
$$\D'_M R(\o)-\f{1}{R(\o)}|\nab R(\o)|^2-\f{1}{R(\o)}+\f{M_0(\o)}{2R(\o)^2}+l.o.t.=0.$$
We first derive $C^0$ estimates for $R(\o)$ by using maximal principle and obtain
$$\f{10}{21}[1+o(1)]\ub a \leq R(\o)\leq \f{11}{21}[1+o(1)]\ub a.$$
We then derive $C^1$ (gradient) estimates. {\color{black}A first try is to use Bochner's formula:
$$\D'_M |\nab R|^2=2|\nab^2 R|^2+2\mbox{Ric}(\nab R, \nab R)+2\nab'^{a} R \nab'_{a}(\D'_M R).$$
If we rewrite $2\nab'^{a} R \nab'_{a}(\D'_M R)$ with the help of equation (\ref{1.12}) and the definition $e'_a=e_a-\O e_a(R)e_3$, we notice that the highest order nonlinear term $|\nab R|^4$ would pop up. However, the coefficient of $|\nab R|^4$ in above equation is negative, which prevents us from using maximal principle to derive uniform bound of $|\nab R|$. To cure this problem,  we need to explore the structures of equation (\ref{1.12}) further. 

Instead of calculating $\D'_M |\nab R|^2$, we consider $\D'_M\l h(R)|\nab R|^2\r$, with a function $h(R)$ to be determined. (Later we will see that $h(R)=1+\f{8}{\ub^2 a^2}(R-\f{\ub a}{2})^2$ is a good choice.)
With Bochner's formula and (\ref{1.12}), for this time we have}
\begin{equation}\label{weightedBochner}
\begin{split}
&\D'_M\l h(R)|\nab R|^2\r\\
=&h'(R)\D'_M R\cdot |\nab R|^2+h''(R)|\nab R|^4+\f{2h'(R)\nab'^a R}{h(R)}\cdot \nab'_{a}\l h(R)|\nab R|^2\r-\f{2h'(R)h'(R)}{h(R)}|\nab R|^4\\
&+h(R)\l 2|\nab^2 R|^2+2\mbox{Ric}(\nab R, \nab R)+2\nab'^{a} R \nab'_{a}(\D'_M R)\r\\
=&h'(R)|\nab R|^2\cdot(-\f12\O \tr\chib |\nab R|^2+\f{\O^{-1}}{2}\tr\chi-2\eta^b \nab_b R-4\O \omb |\nab R|^2)\\
&+h''(R)|\nab R|^4+\f{2h'(R)\nab'^a R}{h(R)}\cdot \nab'_{a}\l h(R)|\nab R|^2\r-\f{2h'(R)h'(R)}{h(R)}|\nab R|^4\\
&+h(R)\l 2|\nab^2 R|^2+2\mbox{Ric}(\nab R, \nab R)\r\\
&+h(R)\cdot 2\nab'^a R \nab'_a (-\f12\O \tr\chib |\nab R|^2+\f{\O^{-1}}{2}\tr\chi-2\eta^b \nab_b R-4\O \omb |\nab R|^2)\\
\end{split}
\end{equation}

With the calculations in Subsection \ref{C1 Estimate}, for $|c, c_a|\ll 1$ we further have
\begin{equation*}
\begin{split}
&\D'_M\l h(R)|\nab R|^2\r\\
\geq&h(R)\l 2|\nab^2 R|^2+2\mbox{Ric}(\nab R, \nab R)\r\\
&+\l \f{h'(R)}{R}+h''(R)-\f{2h'(R)h'(R)}{h(R)}-\f{2h(R)}{R^2}-\f{2h'(R)}{R}\r \cdot (1+c) \cdot |\nab R|^4\\
&+\l h'(R)\cdot \f{\O^{-1}}{2}\tr\chi+h(R)(\f{1}{2R^2}+\f{c}{R^2})\r |\nab R|^2 \\
&+\f{2h'(R)\nab'^a R}{h(R)}\cdot \nab'_{a}\l h(R)|\nab R|^2\r-({\O\tr\chib+8\O\omb)\nab'^a R}\cdot\nab'_a(h(R)|\nab R|^2)\\
&-h(R)|\nab R|^2\nab^a R \f{c_a}{R^2}+h'(R)|\nab R|^2\nab^a R \f{c_a}{R}+h(R)\nab^a R \f{c_a}{R^2}-4h(R)\nab^a R \nab'_a\nab'_b R\,\eta^b.
\end{split}
\end{equation*}

{\color{black}\noindent We hope $\f{h'(R)}{R}+h''(R)-\f{2h'(R)h'(R)}{h(R)}-\f{2h(R)}{R^2}-\f{2h'(R)}{R}$ (the coefficient in front of $|\nab R|^4$) to be positive. For further application, we also hope that the coefficient of $|\nab R|^2$ is positive.}

Construct
$$h(R)=1+\f{8}{\ub^2 a^2}(R-\f{\ub a}{2})^2.$$ 
With $C^0$ estimates $|R-\f{\ub a}{2}|\leq \f{\ub a}{20}$ and the estimates in \cite{AL}, it is straightforward to check
\begin{equation*}
\begin{split}
& h'(R)=\f{16}{\ub^2 a^2} (R-\f{\ub a}{2}), \quad \quad h''(R)=\f{16}{\ub^2 a^2}, \quad \quad |h(R)-1|\leq \f{1}{50},\\
&|h'(R)|\leq \f{4}{5\ub a}, \quad \quad |\O^{-1}\tr\chi|\leq \f{3\ub a}{20 R^2},\\
&\l \f{h'(R)}{R}+h''(R)-\f{2h'(R)h'(R)}{h(R)}-\f{2h(R)}{R^2}-\f{2h'(R)}{R}\r \cdot (1+c) \geq \f{2}{\ub^2 a^2},\\
&\l h'(R)\cdot \f{\O^{-1}}{2}\tr\chi+h(R)(\f{1}{2R^2}+\f{c}{R^2})\r\geq \f{1}{6R^2}.
\end{split}
\end{equation*}

For $2 \mbox{Ric}(\nab R, \nab R)$ term, relying on an estimate (\ref{RicciBound}) in appendix we have
\begin{equation*}
\begin{split}
2\mbox{Ric}(\nab R, \nab R)\geq& -\f{\ub \at}{R}\cdot\f{|\nab R|^4}{R^2}-\f{\ub \at}{R^2} \cdot|\nab^2 R|\cdot |\nab R|^2-\f{\ub a}{R^3}\cdot\f{\ub \at}{R}\cdot|\nab R|^3.
\end{split}
\end{equation*}

We also use $2|\nab^2 R|^2$ to control all the terms involving $\nab^2 R$. Thus, we arrive at 
\begin{equation*}
\begin{split}
\D'_M \l h(R)|\nab R|^2 \r\geq& |\nab R|^4 \cdot\f{1}{\ub^2 a^2}+|\nab R|^2\cdot \f{1}{8R^2}-\f{1}{R^2}\cdot o(1).
\end{split}
\end{equation*}
{\color{black}Note that both the coefficient of $|\nab R|^4$ and the coefficient of $|\nab R|^2$ are positive!}

\noindent Denote $[h(R)|\nab R|^2](\tilde{\o}_{max}):=\max_{\o\in S^2} [h(R(\o))|\nab R(\o)|^2]$.  We hence derive 
$$0\geq \D'_{M}[h(R)|\nab R|^2](\tilde{\o}_{max})\geq |\nab R|^4(\tilde{\o}_{max})\f{1}{\ub^2 a^2}+|\nab R|^2(\tilde{\o}_{max})\f{1}{8R^2}-\f{1}{R^2}\cdot o(1).$$
This implies
$$|\nab R|^2(\tilde{\o}_{max})\ll1.$$
Let $|\nab R|(\o_{max}):=\max_{\o\in S^2} |\nab R|$.  We derive
\begin{equation*}
\begin{split}
|\nab R|^2 (\o_{max})=&\f{1}{[h(R)](\o_{max})}[h(R)|\nab R|^2](\o_{max})\\
\leq& \f{1}{[h(R)](\o_{max})}[h(R)|\nab R|^2](\tilde{\o}_{max})\\
\leq& \f{50}{49}[h(R)|\nab R|^2](\tilde{\o}_{max})\ll 1.
\end{split}
\end{equation*}

Therefore, we conclude
\begin{equation}
|\nab R| (\o)\ll1 \quad \mbox{for all} \quad \o\in S^2. 
\end{equation}

Furthermore, together with elliptic estimates, we obtain
$$\|\nab^2 R(\o)\|_{L^p(M)}\lesssim (\ub a)^{-1+\f{2}{p}}, \quad \quad \|R(\o)\|_{C^{1,q}(M)}\lesssim (\ub a)^{-q}, \quad \quad \|R(\o)\|_{C^{2,q}(M)}\lesssim (\ub a)^{-1-q}.$$

\subsubsection{Existence of MOTS}
With a priori estimates, to solve (\ref{1.12}) we employ the method of continuity.  Recall
$$M_0(\o)=\ub a f(\ub, \o), \quad \mbox{with} \quad \f{20}{21} \leq f(\ub, \o) \leq \f{22}{21}.$$
We first solve the following equation:
\begin{equation}\label{medium step}
\D'_M R(\o)+\f12\O\tr\chib|\nab R(\o)|^2-\f{1}{R(\o)}+\f{\ub a f(\ub, \o)}{2 R(\o)^2}=0.
\end{equation}
We introduce $\lambda$ and let 
$$F(R(\o),\lambda):=\D'_M R(\o)+\f12\O\tr\chib|\nab R(\o)|^2-\f{1}{R(\o)}+\f{\ub a}{2 R(\o)^2}[1+(f(\ub, \o)-1)\lambda].$$
When $\lambda=0$, $R(\o)={\ub a}/{2}$ is a solution to 
$$F(R(\o),0)=\D'_M R(\o)+\f12\O\tr\chib|\nab R(\o)|^2-\f{1}{R(\o)}+\f{\ub a}{2 R(\o)^2}=0.$$
For $0\leq \tilde{\lambda} \leq 1$, assume that $\tilde{R}(\o)$ is a solution to 
\begin{equation}\label{tildeR}
\D'_M \tR(\o)+\f12\O\tr\chib|\nab \tR(\o)|^2-\f{1}{\tR(\o)}+\f{\ub a}{2 \tR(\o)^2}[1+(f(\ub, \o)-1)\tilde{\lambda}]=0.
\end{equation}
{\color{black}Since $0\leq\tilde{\lambda}\leq 1$ we have ${20}/{21}\leq 1+(f(\ub, \o)-1)\tilde{\lambda} \leq {22}/{21}.$ Thus by a priori estimates outlined above, we have $$\f{10}{21}[1+o(1)]\ub a \leq \tR(\o) \leq \f{11}{21}[1+o(1)]\ub a  \quad \mbox{and} \quad |\nab \tR|\ll 1.$$} 
\noindent By direct calculations {\color{black} in Section \ref{Continuity1}}, it follows 
\begin{equation}\label{FR linearization}
\begin{split}
&F_{R}(\tR(\o),\lambda)[W]\\
=&\lim_{\epsilon\rightarrow 0}\f{1}{\epsilon}\l F(\tR+\epsilon W, \lambda)-F(\tR,\lambda) \r\\
=&\D'_{1-\tR,\ub}W+\f{1}{\tR^3}\l -\tR{\color{black}+}\tR|\nab\tR|^2-\ub a (f(\ub, \o)-1)(\tilde{\lambda}-\lambda)\r\cdot [1+o(1)]\cdot W.
\end{split}
\end{equation}

\noindent {\color{black}Note that in this step we always choose $\lambda$ close to $\tilde{\lambda}$. {\color{black}Together with $|\nab R|\ll 1$}, hence the coefficient in front of $W$ is negative. Thus operator $F_R(\tR(\o),\lambda)[W]$ is invertible for $W$. By the method of continuity, repeating this process, we extend $\lambda$ to the whole region $[0,1]$. Hence for $0\leq \lambda \leq 1$ there exists $R_{\lambda}(\o)$ such that $F(R_{\lambda}(\o),\lambda)=0$. In particular $R_1(\o)$ is a solution to (\ref{medium step}), which is $F(R_1(\o),\lambda)=0$. }

This is the starting point for another continuity argument: we further define operator $G$ through 
\begin{equation*}
\begin{split}
G(R(\o),\lambda):=&\D'_M R(\o)+\f12 \O\tr\chib|\nab R(\o)|^2-\f{1}{R(\o)}+\f{\ub a f(\ub, \o)}{2R(\o)^2}\\
&+\lambda \l 2\eta_b\nab^b R(\o)+4\O \omb |\nab R(\o)|^2-\f{\O^{-1}}{2} \tr\chi+\f{1}{R(\o)}-\f{\ub a f(\ub, \o)}{2R(\o)^2}\r.
\end{split}
\end{equation*}
When $\lambda=0$, we have $R_1(\o)$ is a solution to $F(R_1(\o),1)=0$, that is $G(R_1(\o),0)=0$. Therefore, it follows that 
$\tr\chi'=0$ is equivalent to $G(R(\o),1)=0$:
\begin{equation*}
\begin{split}
&\D'_M R(\o)+2\eta_b \nab^b R(\o)+\f12\O\tr\chib|\nab R(\o)|^2\\
&+4\O\omb|\nab R(\o)|^2-\f{\O^{-1}}{2}\tr\chi\\
=&0.
\end{split}
\end{equation*}
For any $\tilde{\lambda}\in (0,1)$, let $\tR(\o)$ solve
\begin{equation*}
\begin{split}
G(\tR(\o),\tilde{\lambda})=&\D'_M \tR(\o)+\f12 \O\tr\chib|\nab \tR(\o)|^2-\f{1}{\tR(\o)}+\f{\ub a f(\ub, \o)}{2\tR(\o)^2}\\
&+\tilde{\lambda} \l 2\eta_b\nab^b \tR(\o)+4\O \omb |\nab \tR(\o)|^2-\f{\O^{-1}}{2} \tr\chi+\f{1}{\tR(\o)}-\f{\ub a f(\ub, \o)}{2\tR(\o)^2}\r\\=&0.
\end{split}
\end{equation*}
By direct calculations {\color{black} in Section \ref{Continuity2}, for $\lambda$ close to $\tilde{\lambda}$ we have} 
\begin{equation*}
\begin{split}
&G_{R}(\tR(\o),\lambda)[W]\\
=&\lim_{\epsilon\rightarrow 0}\f{1}{\epsilon}\l G(\tR+\epsilon W, \lambda)-G(\tR,\lambda) \r\\
=&\D'_{1-\tR,\ub}W-\f{1}{\tR^2}\cdot [1+o(1)]\cdot W,
\end{split}
\end{equation*}
which is invertible for $W$. {\color{black}By the method of continuity, we then extend $\lambda$ to the whole region $[0,1]$.} Hence, there exists a solution $R(\o)$ for $G(R(\o),1)=0$, which satisfies
$$\tr\chi'=0.$$

\subsubsection{Uniqueness of MOTS}

Let's assume we have two solutions $\tR(\o)$ and $R(\o)$ satisfying $\tr\chi'=0$.  We then derive an elliptic equation for $\tR(\o)-R(\o)$. Together with a priori estimates and bounds derived in \cite{AL}, from Section \ref{uniqueness} we have
\begin{equation*}
\begin{split}
&\D_{1-R(\o),\ub}\l \tR(\o)-R(\o)\r-\nu(\o)\l \tR(\o)-R(\o)\r\\
&+\f{1}{\ub^2 a^2}\cdot \l \tR(\o)-R(\o) \r \cdot o(1)\\
&+\f{1}{\ub^2 a^2}\cdot \f{\partial}{\partial \theta_i} \l \tR(\o)-R(\o)\r\cdot o(1)\\
=0,
\end{split}
\end{equation*}
with
$$\nu(\o)\geq \f{64}{81\ub^2 a^2}.$$
By maximal principle, we conclude that 
$$\tR(\o)=R(\o) \quad \mbox{for} \quad \o\in \mathbb{S}^2.$$

\subsubsection{Smoothness of Apparent Horizon}
Along each incoming null hypersurface $\underline{H}_{\ub}$, there exists a unique 2-dimensional MOTS $M_{\ub}$, which is corresponding to the unique $C^2$ solution $R(\ub, \o)$ to $\tr\chi'=0$. Varying $\ub'$, $u=1-R(\ub', \o)$ is thus a three dimensional hypersurface. 

{\color{black} For different $\ub'$ and $\ub$, with the help of a priori estimates and bounds derived in \cite{AL}, in Section \ref{regularity of horizon} we then derive an elliptic equation for $\f{R(\ub', \o)-R(\ub, \o)}{\ub'-\ub}:$}  
\begin{equation*}
\begin{split}
&\Delta_{1-R(\ub, \o), \ub} \l \f{R(\ub',\o)-R(\ub,\o)}{\ub'-\ub}-h(\ub, \o; \ub')\r-\f{\nu(\ub, \o; \ub')}{\ub^2 a^2} \l \f{R(\ub',\o)-R(\ub,\o)}{\ub'-\ub}-h(\ub, \o; \ub')\r \\
&+\f{1}{\ub^2 a^2} \cdot \l\f{R(\ub', \o)-R(\ub, \o)}{\ub'-\ub}-h(\ub, \o; \ub')\r \cdot o(1)\\
&+\f{1}{\ub^2 a^2} \cdot \l\f{\partial}{\partial \theta_i}\f{R(\ub', \o)-R(\ub, \o)}{\ub'-\ub}-h(\ub, \o; \ub')\r \cdot o(1)\\
=&\f{a}{\ub^2 a^2}\cdot o(1).
\end{split}
\end{equation*}
{\color{black}With this equation, in Section \ref{regularity of horizon} we derive uniform bound for $\f{R(\ub', \o)-R(\ub, \o)}{\ub'-\ub}$. Via standard argument for difference quotient, we prove that $\f{\partial R}{\partial \ub}(\ub, \o)=\lim_{\ub'\rightarrow\ub}\f{R(\ub', \o)-R(\ub, \o)}{\ub'-\ub}$ exists and lies in $L^2(M_{\ub})$. Then combining repeated arguments for difference quotient and elliptic estimates, we improve the regularity of $R(\ub, \o)$ and obtain that (with sufficiently smooth initial data in \cite{AL})$$\f{\partial^k R}{\partial \ub^k}\in C^{\infty}(M) \quad \mbox{for any} \quad k\in \mathbb{Z}^{+}.$$ 

 And we will show that $h(\ub, \o;\ub')$ is close to $\f{R(\ub', \o)-R(\ub, \o)}{\ub'-\ub}$ in $L^{\infty}(M)$ norm. By estimating $h(\ub,\o;\ub')$, we calculate the size of $\f{\partial R}{\partial \ub}$. More precisely, }we have both $h(\ub, \o; \ub')$ and $\nu(\ub, \o; \ub')$ being smooth functions respect to $\ub$ and $\o$. And there exist smooth functions $h(\ub, \o)$ and $\nu(\ub, \o)$ such that 
$$h(\ub, \o)=\lim_{\ub'\rightarrow\ub}h(\ub, \o; \ub'), \quad \quad \nu(\ub,\o)=\lim_{\ub'\rightarrow\ub}\nu(\ub, \o; \ub').$$
When $\ub'$ is close to $\ub$, in Section \ref{regularity of horizon} we derive
$$\f{64}{81} \leq \nu(\ub, \o; \ub')\leq \f{1088}{375}, \quad \quad \quad \f{64}{81} \leq \nu(\ub, \o)\leq \f{1088}{375}.$$
{\color{black}And under the conditions in Section \ref{section DH} we obtain
$$h(\ub, \o; \ub')=[\f12+o(1)]a,  \quad \quad \quad h(\ub, \o)=[\f12+o(1)]a.$$
Through elliptic estimates and standard argument for difference quotient, we have 
$$|\f{\partial R(\ub,\o)}{\partial \ub}-h(\ub, \o)|\leq a\cdot o(1).$$}
Recall the apparent horizon we have constructed is a three dimensional hypersurface:
\begin{equation*}
u=1-R(\theta_1, \theta_2, \ub).
\end{equation*}
The components of the induced metric read
\begin{equation*}
g'_{\theta_i \theta_j}=g_{\theta_i \theta_j}+\f{\partial (1-R)}{\partial \theta_i}\cdot \f{\partial(1-R)}{\partial \theta_j} \cdot g(\f{\partial}{\partial u}, \f{\partial}{\partial u})=g_{\theta_i \theta_j},
\end{equation*}
\begin{equation*}
g'_{\ub\, \ub}=g_{\ub\, \ub}+2\f{\partial (1-R)}{\partial \ub}\cdot \f{\partial \ub}{\partial \ub} \cdot g(\f{\partial}{\partial u}, \f{\partial}{\partial \ub})=4\f{\partial R}{\partial \ub}=4h(\ub, \o)\cdot[1+o(1)],
\end{equation*}
\begin{equation*}
g'_{\theta_i \ub}=g_{\theta_i \ub}+\f{\partial (1-R)}{\partial \theta_i}\cdot \f{\partial \ub}{\partial \ub}\cdot g(\f{\partial}{\partial u}, \f{\partial}{\partial \ub})=2\f{\partial R}{\partial \theta_i}.
\end{equation*}
Our parameter $a$ is a large fixed positive constant. Hence the tangent vectors $\f{\partial}{\partial \theta_1}, \f{\partial}{\partial \theta_2}, \f{\partial}{\partial \ub}$ are all spacelike. 
Let $\lambda_1, \lambda_2, \lambda_3$ be any real numbers. Since $h(\ub, \o)=[1/2+o(1)]a$, 
\begin{equation*}
\begin{split}
&g'(\lambda_1 \f{\partial}{\partial \theta_1}+\lambda_2 \f{\partial}{\partial \theta_2}+\lambda_3 \f{\partial}{\partial \ub},\lambda_1 \f{\partial}{\partial \theta_1}+\lambda_2 \f{\partial}{\partial \theta_2}+\lambda_3 \f{\partial}{\partial \ub})\\
=&\lambda_1^2 \cdot g_{\theta_1 \theta_1}+\lambda_2^2 \cdot g_{\theta_2 \theta_2}+4\lambda_1 \lambda_3 \f{\partial R}{\partial \theta_1}+4\lambda_2 \lambda_3 \f{\partial R}{\partial \theta_2}+\lambda_3^2 \cdot h(\o)\cdot [1+o(1)]\\
\geq& 0.
\end{split}
\end{equation*}
Therefore, the apparent horizon formed is spacelike hence a dynamical horizon defined in \cite{AG, AK03, AK}.  

Using the property that the apparent horizon is spacelike, in Section {\ref{entropy}} we prove 
$$\mbox{Area}(M_{\ub'})>\mbox{Area}(M_{\ub}) \quad \mbox{for} \quad \ub'>\ub.$$
This is corresponding to the second law of black hole mechanics. 

\subsection{Related Results on Spacelike Hypersurfaces and Comparison} 
In this paper, the MOTS is defined along an incoming null hypersurface. A MOTS can also be defined along a spacelike hypersurface and it plays an important role in proving positive mass theorem and Penrose's inequality. We refer the interested readers to \cite{Al, AEM, AMS, AM, AG, AK03, AK, E1, E2, M, SY, Yau} and references therein.\\

As to the existence of MOTS, on a spacelike hypersurface Andersson and Metzger \cite{AM} and Eichmair \cite{E1} proved

\textit{Theorem 1.1 in \cite{AM} and Theorem 3.1 in \cite{E1} \footnote{For more detailed statement of these two theorems, interested readers are referred to Theorem 3.3 in \cite{AEM}. }: Set $(N, g, k)$ to be a Cauchy data set. Assume that $N$ is compact and has two boundary components, an inner and an outer boundary. Assuming that the inner boundary is outer trapped and the outer boundary is outer untrapped, then there exists a stable MOTS in $N$.  }

In both \cite{AM} and \cite{E1}, Jang's equation is employed. MOTSs are not critical points for a variational principle, hence the familiar minimization arguments for the existence of minimal surfaces do not generalize to MOTSs. In a talk given at the Miami Waves conference in 2004, Schoen \cite{Sc} pointed out that blowup surfaces for Jang's equation are marginal surfaces and actually provides a result replacing the above mentioned barrier arguments.  In \cite{AM} Andersson and Metzger gave a closer analysis of the blow-up surface. Together with curvature estimates for MOTS and a novel bending procedure to convert the one-sided trapping assumption into the two-sided boundary
curvature conditions, they provided a detailed proof of Schoen's approach. In \cite{E1} Eichmair proved the same results by employing Perron method and tools from geometric measure theory to force and control a blow-up of Jang's equation. With Perron method, in \cite{E1} Eichmair also settled down the Plateau problem. In both \cite{AM} and \cite{E1}, their theorems are beyond the perturbative regime. \\

In this paper, the MOTS is constructed along an incoming null hypersurface. Let $M_{\ub}(\o): ( u=1-R(\o), \ub, \o )$ be the MOTS along $\Hb_{\ub}$. The elliptic equation for the unknown $R(\o)$ is quasilinear in nature. But  as stated earlier, the induced metric $g'_{M_{\ub}}$  on $M_{\ub}$ depends on unknown function $R(\o)$ and independent variable $\o$ but not on $\nab R(\o)$. While, for Jang's equation the induced metric (on a graph $\hat{M}=(x, u(x))$ in the Riemannian product space $(\hat{M}\times\mathbb{R}, g+dt^2)$) depends also on the gradient of the unknown function. Therefore, the form of equation for MOTS in this paper is much nicer than Jang's equation. Using a priori estimates, we conclude $R(\o)$ is $C^{\infty}(\o)$.  

Another advantage of this paper is that we solve hyperbolic EVEs and quasilinear elliptic equation for MOTS at the same time. We first construct spacetimes with trapped surface formation. And within these spacetimes, we solve the elliptic equation for MOTS. Thus, we do not need any assumptions on the (unsolved) spacetimes as in  \cite{AM} and \cite{E1}. We prescribe assumptions directly on the characteristic Cauchy initial data.  

Moreover, this approach allows us to prove that the marginally outer trapped tube (apparent horizon) is formed from a point $O$. By applying the scale critical trapped surface formation criterion \cite{AL}, we have the existence of EVEs in a spacetime region up to $O$ along incoming null hypersurfaces. And we solve the equation for MOTS along these null hypersurfaces. This gives a sequence of MOTS connecting $O$. Note that the methods as in \cite{AM, E1} do not apply here. Because we do not have (existence) informations about the spacetime to the future of $O$, and we hence do not have a proper spacelike initial data set $(N, g, k)$. 

Finally, it has to be stated that the characteristic initial data prescribed for Theorem \ref{thm1.8} has a jump discontinuity for $\chih_0$. Therefore, the spacetimes constructed in this paper could be considered as generalizations of the well-known Vaidya spacetimes \footnote{Vaidya spacetime describes a spherically symmetric solution to Einstein-dust system, which is either emitting or absorbing null dusts.} in physics literature. 

\subsection{Outline of the Paper}
In Section 2 and Section 3, we give the basic settings. In Section 4, we derive a priori estimates. In Section 5 and Section 6, we apply the method of continuity. In Section 7, we derive a lower bound for injectivity radius for MOTS. In Section 8, we study uniqueness of apparent horizon.  In Section 9, we prove regularity of apparent horizon. In Section 10, we demonstrate that the apparent horizon constructed is spacelike. In Section 11, we verify the second law of black hole mechanics.

\subsection{Acknowledgements}
We are grateful for enlightening discussions with Jonathan Luk and for his helpful suggestions on an earlier version of the manuscript. We thank Po-Ning Chen, Demetrios Christodoulou, Zheng-Chao Han, Sergiu Klainerman, Yakov Shlapentokh-Rothman, Shadi Tahvildar-Zadeh and Willie Wong for valuable conversations.

\section{Basic Setup} \label{sec.setting}

In this section, we review the setup of double null foliation and the coordinate system. We also define the geometric quantities associated with them.

\subsection{Double Null Foliation and Coordinate System}\label{secdnf}
To explore the geometric structures hidden in EVEs, we define the following double null foliation:

\begin{minipage}[!t]{0.5\textwidth}
\begin{tikzpicture}[scale=0.95]

\draw [white](3,-1)-- node[midway, sloped, below,black]{$H_0(u=0)$}(4,0);

\draw [white](0.5,1.5)-- node[midway,sloped,above,black]{$H_{1-b \delta \at}$}(1.5,2.5);
\draw [white](2,2)--node [midway,sloped,above,black] {$\Hb_{\delta}(\ub=\delta)$}(4,0);
\draw [white](1,1)--node [midway,sloped, below,black] {$\Hb_{0}(\ub=0)$}(3,-1);
\draw [dashed] (0, 4)--(0, -4);
\draw [dashed] (0, -4)--(4,0)--(0,4);
\draw [dashed] (0,0)--(2,2);
\draw [dashed] (0,1)--(1.5,2.5);
\draw [dashed] (0,-4)--(2,-2);
\draw [dashed] (0,2)--(3,-1);
\draw [very thick] (1,1)--(3,-1)--(4,0)--(2,2)--(1,1);
\draw [very thick] (1,1)--(0.5,1.5)--(1.5,2.5)--(2,2)--(1,1);
\fill[yellow!70!red] (1,1)--(3,-1)--(4,0)--(2,2)--(1,1);
\fill[yellow!30!red](1,1)--(0.5,1.5)--(1.5,2.5)--(2,2)--(1,1);
\draw [white](1,1)-- node[midway,sloped,above,black]{$H_{1-\delta a}$}(2,2);
\draw [->] (3.3,-0.6)-- node[midway, sloped, above,black]{$e_4$}(3.6,-0.3);
\draw [->] (1.4,1.3)-- node[midway, sloped, below,black]{$e_4$}(1.7,1.6);
\draw [->] (3.3,0.6)-- node[midway, sloped, below,black]{$e_3$}(2.7,1.2);
\draw [->] (2.4,-0.3)-- node[midway, sloped, above,black]{$e_3$}(1.7,0.4);
\end{tikzpicture}
\end{minipage}
\begin{minipage}[!t]{0.5\textwidth}
Let $u$ and $\ub$ be solutions to the eikonal equations
$$g^{\mu\nu}\partial_\mu u\partial_\nu u=0,$$
$$g^{\mu\nu}\partial_\mu\ub\partial_\nu \ub=0.$$
We choose $u$ and $\ub$ satisfying $u=0$ on $H_0$ and $\ub=0$ on $\Hb_0$. \\

Here $\ub$ is increasing towards the future
while $u$ is decreasing towards the future.
\end{minipage}
\hspace{0.05\textwidth} 
Let
$$L'^\mu=-2g^{\mu\nu}\partial_\nu u,\quad  \quad \Lb'^\mu={\color{black}-}2g^{\mu\nu}\partial_\nu \ub.$$ 
Both $L'$ and $\Lb'$ are future directed, null geodesic vector fields. Define
$$2\Omega^{-2}=-g(L',\Lb').$$
Denote
$$e_3=\Omega\Lb', \quad \quad e_4=\Omega L'$$
to be the normalized null pair such that 
$$g(e_3,e_4)=-2.$$
 Define also
$$\Lb=\Omega^2\Lb', \quad \quad L=\Omega^2 L'.$$

By prescribing 
$$\Omega=1\quad\mbox{on $H_0$ and $\Hb_0$},$$
we fix the gauge on the initial characteristic hypersurfaces.We denote the level sets of $u$ as $H_u$ and the level sets of $\ub$ as $\Hb_{\ub}$. By the eikonal equations, $H_u$ and $\Hb_{\ub}$ are null hypersurface. The  intersections of the hypersurfaces $H_u$ and $\Hb_{\ub}$ are topologically 2-spheres. We denote them by $S_{u,\ub}$. 

For the spacetime in a neighborhood of $S_{0,0}$, we define a coordinate system $(u,\ub,\th^1,\th^2)$: Let $(\th^1,\th^2)$ be a coordinate system on the standard sphere $S_{0,0}$ and with the property that on each coordinate patch the metric $\gamma$ is smooth, bounded and positive definite. 
Next, the coordinates on the initial hypersurfaces are defined by requiring $\th^A$ to 
be constant along null generators of the initial hypersurface. In the spacetime, we also define $u$ and $\ub$ to be solutions to the eikonal equations:
$$g^{\mu\nu}\partial_\mu u\partial_\nu u=0,\quad g^{\mu\nu}\partial_\mu\ub\partial_\nu \ub=0.$$
{\color{black}Furthermore, we propagate $\th^1, \th^2$ along $\Lb$ direction: defining $\th^1, \th^2$ by solving
$$\Lb \th^1=0, \quad \mbox{and} \quad \Lb \th^2=0.$$
}
\noindent According to this coordinate system $(u,\ub,\th^1,\th^2)$, we express $e_3$ and $e_4$ as
$$e_3=\Omega^{-1}\frac{\partial}{\partial u}, \quad e_4=\Omega^{-1}\left(\frac{\partial}{\partial \ub}+d^A\frac{\partial}{\partial \th^A}\right)$$
for some $d^A$ satisfying $d^A=0$ on $\Hb_0$. Hence, the metric $g$ takes the form
$$g=-2\Omega^2(du\otimes d\ub+d\ub\otimes du)+\gamma_{AB}(d\th^A-d^Ad\ub)\otimes (d\th^B-d^Bd\ub).$$

\subsection{Ricci Coefficients and Curvature Components}\label{seceqn}
We then define the geometric quantities with a null frame $e_3$, $e_4$ introduced above and an frame ${e_1,e_2}$ tangent to the 2-spheres $S_{u,\ub}$. Using the indices $A,B\in\{1,2\}$, we define the Ricci coefficients:
 \begin{equation}
\begin{split}
&\chi_{AB}=g(D_A e_4,e_B),\, \,\, \quad \chib_{AB}=g(D_A e_3,e_B),\\
&\eta_A=-\frac 12 g(D_3 e_A,e_4),\quad \etab_A=-\frac 12 g(D_4 e_A,e_3),\\
&\tilde{\omega}=-\frac 14 g(D_4 e_3,e_4),\quad\,\,\, \omegab=-\frac 14 g(D_3 e_4,e_3),\\
&\zeta_A=\frac 1 2 g(D_A e_4,e_3)
\end{split}
\end{equation}
where $D_A=D_{e_{A}}$. We separate the trace and traceless part of $\chi$ and $\chib$. Denote $\chih$ and $\chibh$ to be the traceless parts of $\chi$ and $\chib$, respectively. 

{\color{black}Recall 
$R(X,Y, Z, W):=g(D_X D_Y Z-D_Y D_X Z-D_{[X,Y]}Z, W).$
We also introduce the  null curvature components}
\begin{equation}
\begin{split}
\a_{AB}&=R(e_A, e_4, e_B, e_4),\quad \, \,\,   \ab_{AB}=R(e_A, e_3, e_B, e_3),\\
\b_A&= \frac 1 2 R(e_A,  e_4, e_3, e_4) ,\quad \bb_A =\frac 1 2 R(e_A,  e_3,  e_3, e_4),\\
\rho&=\frac 1 4 R(e_4,e_3, e_4,  e_3),\quad \sigma=\frac 1 4  \,^*R(e_4,e_3, e_4,  e_3).
\end{split}
\end{equation}
Here $\, ^*R$ is the Hodge dual of $R$.  Denote $\nab$ to be the 
induced covariant derivative on $S_{u,\ub}$ and $\nab_3$, $\nab_4$ to be
the projections to $S_{u,\ub}$ of the covariant derivatives $D_3$, $D_4$ (see
precise definitions in \cite{KNI:book}). We further have
\begin{equation}
\begin{split}
&\tilde{\omega}=-\frac 12 \nab_4 (\log\Omega),\qquad \omegab=-\frac 12 \nab_3 (\log\Omega),\\
&\eta_A=\zeta_A +\nab_A (\log\Omega),\quad \etab_A=-\zeta_A+\nab_A (\log\Omega)
\end{split}
\end{equation}

\section{Deformation Formula}\label{deformation formula}
With the notation in Section \ref{sec.setting}, we follow the calculation in \cite{KLR} and derive the null expansion $\tr\chi'$ for 2-sphere {\color{black}$M:=\{(u, \ub, \o): u=1-R(\o)\}$} on $\Hb_{\ub}$. \\

We use null frames
\begin{equation}\label{nfr}
e'_3=e_3, \quad e'_a=e_a-\O e_a(R)e_3, \quad e'_4=e_4-2\O e^a(R) e_a+\O^2 |\nab R|^2 e_3. 
\end{equation}
By definition $e_3(u)=\O^{-1}$. Hence $e'_a(u+R-1)=e_a(R)-e_a(R)\O e_3(u)=0$. Recall $e_3$ is orthogonal to any vector tangent to $\Hb$. We thus have 
\begin{equation*}
g(e'_a, e'_b)=g(e_a, e_b)=\d_{ab}, \quad g(e'_4, e'_a)=g(e'_4, e'_4)=0, \quad g(e'_3, e'_4)=-2.
\end{equation*}
It is straightforwrd to check 
\begin{proposition}\label{1st deformation}(\cite{KLR})
The trace of the null second fundamental form $\chi'$, relatively to the new frames (\ref{nfr}) is given by 
\begin{equation}
\tr\chi'=\tr\chi-2\O\D R-4\O\eta\cdot\nab R-4\O^2\chibh_{bc}\nab^b R \nab^c R-\O^2 \tr\chib |\nab R|^2-8\O^2\omb|\nab R|^2.
\end{equation}
\end{proposition}

In anticipation for later use, we rewrite $e^a e_a (R)$ with $\D'_M R$, where $\D'_M$ is the Laplace-Beltrami operator on $M$:
\begin{proposition}\label{change laplacian}
\begin{equation}
e^a e_a (R)=\D'_{M}R-2\O \chibh_{ab}\nab^a R \nab^b R. 
\end{equation}
\end{proposition}

\begin{proof}
{\color{black}We denote $(e_3 e_a)R=e_3 e_a (R):=D^2_{e_{3,} e_a}R.$} Employing commutator formulas in \cite{KR:Trapped}, we get
$$(e_3 e_a)R=-\chibh_{ac} \nab^c R-\f12 \tr\chib \nab_a R.$$
Recall the definition $\chib_{ab}:=g(D_{e_a}e_3, e_b)$. {\color{black} Along each $\Hb_{\ub}$, we could rewrite $R(\ub, \o)=R(\o)$, thus we have $(D_{e_3}e_3)R=0$.} And we deduce 
\begin{equation}\label{eaeaR}
\begin{split}
&e^a e_a (R)\\
=& (e^{a'}+\O e^a(R)e_3)(e_{a'}+\O e_a(R)e_3) R\\
=& (e^{a'}+\O e^a(R)e_3)(e_{a'}(R))-\l D_{e^{a'}+\O e^a(R)e_3} e_{a'}+\O e_a (R)e_3\r R\\
=& e^{a'}(e_{a'}(R))-(D_{e^{a'}}e_{a'}) R+ \O e^a(R) e_3 (e_a (R))-\O e^a (R) (D_{e_3}e_{a'}) R\\
&-\O e_a(R) D_{e^{a'}}e_3 R-\O e^a(R) \O e_a (R) (D_{e_3}e_3) R\\
=& e^{a'} e_{a'}(R)+\O e^a(R) e_3 e_a (R)-\O e_a(R) ( D_{e^a} e_3 ) R\\
=& e^{a'} e_{a'}(R)+\O e^a(R) e_3 e_a (R)-\O e^a(R) \chib_{ab} e^b(R)\\
=& e^{a'} e_{a'}(R)-\f{\O}{2} \tr\chib \nab^a R \nab_a R- \O \chibh_{ac}\nab^a R \nab^c R-\O e^a(R) \chib_{ab} e^b(R)\\
=& e^{a'} e_{a'}(R)-\O \tr\chib \nab^a R \nab_a R- 2\O \chibh_{ac}\nab^a R \nab^c R.
\end{split}
\end{equation}

On the other side
$$D_{e_{a'}}e_{a'}=\overline{D_{e_{a'}}e_{a'}}-\f12 g(D_{e_{a'}}e_{a'}, e'_4)e'_3-\f12 g(D_{e_{a'}}e_{a'}, e'_3)e'_4.$$
Here $\overline{D_{e_{a'}}e_{a'}}$ is the projection of $D_{e_{a'}}e_{a'}$ to $T_{p}M$ at point $p$.
From $e'_3=e_3$ and $e'_4=e_4-2\O e^a(R) e_a+\O^2 |\nab R|^2 e_3$, we thus have
\begin{equation*}
\begin{split}
D_{e_{a'}}e_{a'}(R)=&\overline{D_{e_{a'}}e_{a'}}(R)-\f12 g(D_{e_{a'}}e_{a'}, e'_3)e'_4(R)\\
=&\overline{D_{e_{a'}}e_{a'}}(R)-\f12 g(D_{e_{a'}}e_{a'}, e'_3)\cdot -2\O e^a(R)e_a(R)\\
=&\overline{D_{e_{a'}}e_{a'}}(R)+\f12 g(D_{e_{a'}}e_3, e'_a)\cdot -2\O e^a(R)e_a(R)\\
=&\overline{D_{e_{a'}}e_{a'}}(R)+\f12 \tr\chib\cdot -2\O e^a(R)e_a(R)\\
=&\overline{D_{e_{a'}}e_{a'}}(R)-\O\tr\chib \nab^a R\nab_a R.
\end{split}
\end{equation*}

Hence,
\begin{equation}\label{Laplacian operator definition}
\begin{split}
e^{a'} e_{a'}(R)=&e^{a'}{\color{black}\big(e_{a'}(R)\big)}-D_{e^{a'}}e_{a'}(R)\\
=&e^{a'}{\color{black}\big(e_{a'}(R)\big)}-\overline{D_{e^{a'}}e_{a'}}(R)+\O\tr\chib\nab^a R \nab_a R\\
=&\D'_M R+\O\tr\chib \nab^a R \nab_a R.
\end{split}
\end{equation}

Together with (\ref{eaeaR}), we arrive at
\begin{equation*}
\begin{split}
e^a e_a(R)=&e^{a'} e_{a'}(R)-\O \tr\chib \nab^a R \nab_a R- 2\O \chibh_{ac}\nab^a R \nab^c R\\
=&\D'_M R+\O\tr\chib \nab^a R \nab_a R-\O \tr\chib \nab^a R \nab_a R-2\O \chibh_{ac}\nab^a R \nab^c R\\
=&\D'_M R-2\O \chibh_{ac}\nab^a R \nab^c R.
\end{split}
\end{equation*}

Here $\D'_M$ is the induced Laplace-Beltrami operator on $M$. 
\end{proof}

We then conclude
\begin{proposition}\label{newdeformation}
The trace of the null second fundamental form $\chi'$, relatively to the new frames (\ref{nfr}) is 
\begin{equation}
\tr\chi'=\tr\chi-2\O\D'_M R-4\O\eta\cdot\nab R-\O^2 \tr\chib |\nab R|^2-8\O^2\omb|\nab R|^2.
\end{equation}
\end{proposition}

\begin{remark}
Note that $e_a (R)=e_{a'} (R)$. 
\end{remark}

\section{A Priori Estimates}\label{a priori estimates section} 
Along each incoming null hypersurface $\Hb_{\ub}$, to get the location of the 2-dimensional MOTS: $(u, \ub, \o)=(1-R(\ub, \o), \ub, \o)$, we will solve a quasilinear elliptic equation for $R(\ub, \o)$. To solve it, we employ the method of continuity. In this section, we derive the a priori estimates in need.  \\

On $\underline{H}_{\ub}$, we define operator $L$
\begin{equation}
\begin{split}
L(R(\omega)):=&\D'_M R(\o)+2\eta_b\nab^b R(\o)+\f12 \O \tr\chib |\nab R(\o)|^2\\
&+4\O \omb |\nab R(\o)|^2-\f{\O^{-1}}{2} \tr\chi. 
\end{split}
\end{equation}
Here $R(\o)$ is an abbreviation for $R(\ub, \o)$. For each fixed $\ub$, $\D'_M$ is the Laplace-Beltrami operator on the unknown 2-dimensional MOTS: 
$$(u, \ub, \o)=(1-R(\ub, \o), \ub, \o).$$ 
Notice that $\tr\chi'=0$ is equivalent to $L(R(\o))=0$:
\begin{equation}\label{LR=0}
\begin{split}
&\D'_M R(\o)+2\eta_b\nab^b R(\o)+\f12 \O \tr\chib |\nab R(\o)|^2\\
&+4\O \omb |\nab R(\o)|^2-\f{\O^{-1}}{2} \tr\chi\\
=&0.
\end{split}
\end{equation}  
Given $\o \in \mathbb{S}^2$, we solve for $R(\o)$ via (\ref{LR=0}).

In later sections, we will employ continuity arguments. For $0\leq\lambda\leq 1$, we will need a priori estimates for solutions to 
\begin{equation}\label{eqn2 continuity}
\D'_M R(\o)+\f12\O\tr\chib|\nab R(\o)|^2-\f{1}{R(\o)}+\f{\ub a}{2R(\o)^2}[1+(f(\ub, \o)-1)\lambda]=0, 
\end{equation}
and
\begin{equation}\label{eqn1 continuity}
\begin{split}
&\D'_M R(\o)+\f12 \O\tr\chib|\nab R(\o)|^2-\f{1}{R(\o)}+\f{\ub a f(\ub, \o)}{2R(\o)^2}\\
&+\lambda \l 2\eta_b\nab^b R(\o)+4\O \omb |\nab R(\o)|^2-\f{\O^{-1}}{2} \tr\chi+\f{1}{R(\o)}-\f{\ub a f(\ub, \o)}{2R(\o)^2}\r=0. 
\end{split}
\end{equation}
In this section we derive these a priori bounds. Since $|f(\ub, \o)-1|\leq \f{1}{21}$, we have 
$$\f{20}{21}\leq 1+(f(\ub, \o)-1)\lambda \leq \f{22}{21}.$$
Recall in \cite{AL} we obtained estimates:

$$|\chibh,\omb, \eta, \tr\chib+\f{2}{R(\o)}|\leq \f{\ub \at{\color{black}b^{\f14}}}{R(\o)^2},$$
$$|\O-1|\leq \f{\ub \at \bf}{R(\o)},$$
$$|\tr\chi-\f{2}{R(\o)}+\f{1}{R(\o)^2}\int_0^{\ub}|\chih_0|^2(\ub',\o)d\ub'|\leq \f{\ub\at\bf}{R(\o)^2}.$$
Therefore, for any fixed $\lambda$ with $0\leq\lambda\leq 1$, (\ref{eqn2 continuity}) and (\ref{eqn1 continuity}) are equivalent \footnote{For (\ref{eqn1 continuity}), since $\f{20}{21}\leq 1+(f(\o)-1)\lambda \leq \f{22}{21}$, we rename $1+(f(\o)-1)\lambda$ to $f(\o)$ with lower and upper bounds $\f{20}{21}$ and $\f{22}{21}$, respectively. } to
\begin{equation}\label{3.3}
\begin{split}
&\Delta'_M R(\o)-\f{1}{R(\o)}|\nab R(\o)|^2-\f{1}{R(\o)}+\f{\ub a f(\ub, \o)}{2R(\o)^2}\\
&+\f{\ub\at c_{1a}}{R(\o)^2}\nab^a R(\o)+\f{\ub\at c_{2bc}}{R(\o)^2}\nab^b R(\o)\nab^c R(\o)+\f{\ub\at}{R(\o)^2} c_3\\
=&0.
\end{split}
\end{equation}
Here $\f{20}{21} \leq f(\ub, \o) \leq \f{22}{21}$, $|c_1, c_2, c_3|\leq \bf$ and $c_1, c_2, c_3$ depend not on $\nab R(\omb), \nab^2 R(\omb)$ but only on $R(\omb)$.\\

We are ready to derive a priori estimates for $R(\o)$ in (\ref{LR=0}).

\subsection{$C^0$ Estimate}
Denote $M_0(\o):=\ub a f(\ub, \o).$ From (\ref{3.3}), we have
\begin{equation*}
\begin{split}
&\Delta'_M R(\o)-\f{1}{R(\o)}|\nab R(\o)|^2-\f{1}{R(\o)}+\f{M_0(\o)}{2R(\o)^2}\\
&+\f{\ub\at c_{1a}}{R(\o)^2}\nab^a R(\o)+\f{\ub\at c_{2bc}}{R(\o)^2}\nab^b R(\o)\nab^c R(\o)+\f{\ub\at}{R(\o)^2} c_3\\
=&0.
\end{split}
\end{equation*}

Let $R(\o_{max}):=\max_{\o\in S^2}R(\o)$ and $R(\o_{min}):=\min_{\o\in S^2}R(\o)$. Since $\nab R(\o_{max})=0,$ we arrive at
\begin{equation*}
\begin{split}
0=&\D'_M R(\o_{max})-\f{1}{R(\o_{max})}+\f{M_0(\o_{max})}{2R(\o_{max})^2}+\f{\ub\at}{R(\o_{max})^2}c_3\\
\leq& -\f{1}{R(\o_{max})}+\f{M_0(\o_{max})}{2R(\o_{max})^2}+\f{\ub\at}{R(\o_{max})^2}c_3\\
=& \f{-2R(\o_{max})+M_0(\o_{max})+2\ub\at c_3}{2R(\o_{max})^2},\\
\end{split}
\end{equation*} 
where $|c_3|\leq b^{\f14}\leq a^{\f18}$. {\color{black}
Note that $M_0(\o)=\ub a f(\o)$ and $20/21\leq f(\o)\leq 22/21$, the above inequality implies
\begin{equation*}
\begin{split}
R(\o_{max})\leq& \f12M_0(\o_{max})+\ub\at c_3\leq \f12M_0(\o_{max})+\ub a^{\f58}\\
=&[\f12+\f{\ub a^{\f58}}{\ub a f(\o_{max})}]M_0(\o_{max})=[\f12+O(\f{1}{a^{\f38}})]M_0(\o_{max})\\
=&[\f12+o(1)] \max_{\o}M_0(\o).
\end{split}
\end{equation*} 
}
Similarly, with $\nab R(\o_{min})=0$, we deduce
\begin{equation*}
\begin{split}
0=&\D'_M R(\o_{min})-\f{1}{R(\o_{min})}+\f{M_0(\o_{min})}{2R(\o_{min})^2}+\f{\ub\at}{R(\o_{min})^2}c_3\\
\geq& -\f{1}{R(\o_{min})}+\f{M_0(\o_{min})}{2R(\o_{min})^2}+\f{\ub\at}{R(\o_{min})^2}c_3\\
=& \f{-2R(\o_{min})+M_0(\o_{min})+2\ub\at c_3}{2R(\o_{min})^2},\\
\end{split}
\end{equation*} 
where $|c_3|\leq b^{\f14}\leq a^{\f18}$. 

Hence
\begin{equation*}
R(\o_{min})\geq [\f12+o(1)]M_0(\o_{min})\geq [\f12+o(1)]\min_{\o}M_0(\o). 
\end{equation*}
To conclude, we obtain
$$[\f12+o(1)]\f{20}{21} \ub a\leq R(\o)\leq [\f12+o(1)]\f{22}{21}\ub a.$$

\subsection{$W^{1,2}$ Estimate}
Integrating (\ref{3.3}) on $M$ yields 
\begin{equation*}
\begin{split}
&\int_M \Delta'_M R(\o)-\int_M \f{1}{R(\o)}|\nab R(\o)|^2-\int_M \f{1}{R(\o)}+\int_M \f{M_0(\o)}{2R(\o)^2}\\
&+\int_M \f{\ub\at c_{1a}}{R(\o)^2}\nab^a R(\o)+\int_M \f{\ub\at c_{2bc}}{R(\o)^2}\nab^b R(\o)\nab^c R(\o)+\int_M\f{\ub\at}{R(\o)^2} c_3\\
=&0.
\end{split}
\end{equation*}
Together with $\f38 \min_{\o}M_0(\o) \leq |R(\o)|\leq \f58 \max_{\o}M_0(\o)$, we have
\begin{equation}\label{W12}
\int_M |\nab R(\o)|^2 \lesssim \ub^{{\color{black}2}} a^{{\color{black}2}}. 
\end{equation}
{\color{black} 
\begin{remark}
For later application, let $k\in \mathbb{N}^{+}$ and $p\geq 1$ we then define \textit{scale-invariant} $\widetilde{W}^{k,p}$ norm. Recall $\ub a/4 \leq R(\o)\leq \ub a$. For $f=f(\o)$ with $\o\in \mathbb{S}^2$, on $M$ we define
\begin{equation}\label{scale invariant WKP}
\|f(\o)\|_{\widetilde{W}^{k,p}(M)}:=\sum_{0\leq i\leq k}\|R(\o)^{-\f{2}{p}}\cdot\big(R(\o)\nab\big)^if(\o)\|_{L^p(M)}.
\end{equation}
Thus, it's easy to check $\|1\|_{\widetilde{W}^{k,p}(M)}\approx 1$ and the above (\ref{W12}) and $C^0$ estimate imply
$$\|R(\o)\|_{\widetilde{W}^{1,2}(M)}\lesssim \ub a\lesssim R(\o).$$ 
 
\end{remark}
}

\subsection{$C^1$ Estimate}\label{C1 Estimate}
To achieve gradient estimate, we will use Bochner's formula:
$$\D'_M |\nab' R|^2=2|\nab'^2 R|^2+2{\color{black}\mbox{Ric}_{M}}(\nab' R, \nab' R)+2\nab'^a R \nab'_{a} (\D'_M R).$$
Here $\nab'$ denotes the covariant derivative on $M$. {\color{black} 
Note $\nab' R(\o)=\nab R(\o)$. When there is no danger of confusion, for simplicity, we replace $\nab' R(\o)$ by $\nab R(\o)$.}
In order to get a desired $L^{\infty}$ estimate for $\nab R$, we consider $\D'_M\l h(R)|\nab R|^2\r,$
where 
$$h(R)=1+\f{8}{\ub^2 a^2}(R-\f{\ub a}{2})^2.$$
Together with Bochner's formula, we have
\begin{equation}\label{weightedBochner}
\begin{split}
&\D'_M\l h(R)|\nab R|^2\r\\
=&h'(R)\D'_M R\cdot |\nab R|^2+h''(R)|\nab R|^4+\f{2h'(R)\nab'^a R}{h(R)}\cdot \nab'_{a}\l h(R)|\nab R|^2\r-\f{2h'(R)h'(R)}{h(R)}|\nab R|^4\\
&+h(R)\l 2|\nab{\color{black}'}^2 R|^2+2{\color{black}\mbox{Ric}_{M}}(\nab R, \nab R)+2\nab'^{a} R \nab'_{a}(\D'_M R)\r\\
=&h'(R)|\nab R|^2\cdot(-\f12\O \tr\chib |\nab R|^2+\f{\O^{-1}}{2}\tr\chi-2\eta^b \nab_b R-4\O \omb |\nab R|^2)\\
&+h''(R)|\nab R|^4+\f{2h'(R)\nab'^a R}{h(R)}\cdot \nab'_{a}\l h(R)|\nab R|^2\r-\f{2h'(R)h'(R)}{h(R)}|\nab R|^4\\
&+h(R)\l 2|\nab{\color{black}'}^2 R|^2+2{\color{black}\mbox{Ric}_{M}}(\nab R, \nab R)\r\\
&+h(R)\cdot 2\nab'^a R \nab'_a (-\f12\O \tr\chib |\nab R|^2+\f{\O^{-1}}{2}\tr\chi-2\eta^b \nab_b R-4\O \omb |\nab R|^2)\\
\end{split}
\end{equation}

We first focus on the term $h(R)\cdot2\nab^a R \nab'_a (-\f12 \O \tr\chib|\nab R|^2)$. Here
$$h(R)\cdot2\nab^a R \nab'_a (-\f12 \O \tr\chib|\nab R|^2)=2h(R)\nab^a R \nab'_a (-\f12 \O \tr\chib)|\nab R|^2-h(R)\O\tr\chib\nab'^a R \nab'_a(|\nab R|^2).$$
Employing $e'_a=e_a-\O e_a(R)e_3$ and the estimate for $\nab_3 (\O \tr\chib)$ in \cite{AL}, we derive
\begin{equation*}
\begin{split}
&2h(R)\nab^a R \nab'_a (-\f12 \O \tr\chib)|\nab R|^2\\
=&-2h(R)\nab^a R \nab_a (\f12 \O \tr\chib)|\nab R|^2+h(R)\O\nab_a R \nab^a R \nab_3 (\O\tr\chib)|\nab R|^2\\
=& -h(R)|\nab R|^4 (\f{2}{R^2}+\f{c}{R^2})-2h(R)\nab^a R \nab_a (\f12 \O \tr\chib)|\nab R|^2\\
=& -h(R)|\nab R|^4 (\f{2}{R^2}+\f{c}{R^2})-h(R)|\nab R|^2 \nab^a R \f{c_a}{R^2},
\end{split}
\end{equation*}
where $|c, c_a|\leq {\color{black}1/a^{\f14}}$.
And
\begin{equation*}
\begin{split}
&-h(R)\O\tr\chib\nab'^a R \nab'_a(|\nab R|^2)\\
=&-{\O\tr\chib\nab'^a R}\cdot\nab'_a(h(R)|\nab R|^2)+{\O\tr\chib h'(R)}|\nab R|^4.
\end{split}
\end{equation*}

Similarly, for $h(R)\cdot2\nab^a R \nab'_a (\f{\O^{-1}}{2}\tr\chi)$ we have
\begin{equation*}
\begin{split}
&h(R)\cdot2\nab^a R \nab'_a (\f{\O^{-1}}{2}\tr\chi)\\
=&2h(R)\nab^a R \nab_a (\f{\O^{-1}}{2}\tr\chi)-2h(R)\O\nab^a R \nab_a R \nab_3 (\f{\O^{-1}}{2}\tr\chi).\\
\end{split}
\end{equation*}
Define $\rhoc:=\rho-\f12\chih\cdot\chibh$, from \cite{AL} we have 
$$\nab_3 \tr\chi=-\f{1}{2} \tr\chib \tr\chi+2\check{\rho}+2\omb \tr\chi+2\div \eta+2|\eta|^2.$$
In \cite{AL}, we also have
\begin{equation*}
\begin{split}
&\nab_4\chih+\trch \chih=-2 \tilde{\omega} \chih-\alpha,\\
&\nab_4\chibh +\frac 1 2 \trch \chibh=\nab\widehat{\otimes} \etab+2\tilde{\omega} \chibh-\frac 12 \trchb \chih +\etab\widehat{\otimes} \etab,\\
&\nab_4\rho+\frac 32\trch\rho=\div\beta-\frac 12\chibh\cdot\alpha+\zeta\cdot\beta+2\etab\cdot\beta.
\end{split}
\end{equation*}
Therefore, for $\rhoc=\rho-\f12 \chih\cdot\chibh$, it satisfies
$$\nab_4 \check{\rho}=-\f32 \tr\chi \check{\rho}+\div \b+\zeta\cdot \b+2\etb\cdot \b-\f12 \chih\cdot\nab\hat{\otimes}\etb-\f12\chih\cdot(\etb\hat{\otimes}\etb)+\f14 \tr\chib |\chih|^2.$$

Employ the pointwise estimate in \cite{AL}, we derive
$$\rhoc |_{S_{R,\ub}}=-\f{\ub a f(\ub, \o)}{2R^3}+\f{\ub\at c_3}{R^3}$$
and 
\begin{equation}\label{trchi u}
\nab_3 \tr\chi |_{S_{1-R, \ub}}=\f{1}{R}(\f{2}{R}-\f{\ub a f(\ub, \o)}{R^2})-\f{\ub a f(\ub, \o)}{R^3}+\f{\ub\at c_3}{R^3}=\f{2}{R^2}-\f{2\ub a f(\ub, \o)}{R^3}+\f{\ub\at c_3}{R^3},
\end{equation}
with $|c_3|\leq b^{\f14}$. 

In later sections, it is crucial to obtain a lower bound for 
$$-\f{2}{R^2}+\f{2\ub a f(\ub, \o)}{R^3}=\f{-2R+2\ub a f(\ub, \o)}{R^3}.$$
Here we have {\color{black}a key estimate with positive lower bound}
$$-2R+2\ub a f(\ub, \o) \geq -2R+\f53 \ub a \geq -2R+\f{5}{3}\cdot \f{12}{7}\cdot [1+o(1)]\cdot R=\f{6}{7}\cdot [1+o(1)]\cdot R>\f12 R.$$

Therefore, 
\begin{equation*}
\begin{split}
&h(R)\cdot2\nab^a R \nab'_a (\f{\O^{-1}}{2}\tr\chi)\\
=&2h(R)\nab^a R \nab_a (\f{\O^{-1}}{2}\tr\chi)-2h(R)\O\nab^a R \nab_a R \nab_3 (\f{\O^{-1}}{2}\tr\chi)\\
=&h(R)\nab^a R \f{c_a}{R^2}-h(R)|\nab R|^2 (\f{2}{R^2}-\f{2\ub a f(\ub, \o)}{R^3}+\f{c}{R^2})\\
=&h(R)\nab^a R \f{c_a}{R^2}+h(R)|\nab R|^2 (\f{-2R+2\ub a f(\ub, \o)}{R^3}+\f{c}{R^2})\\
\geq&h(R)\nab^a R \f{c_a}{R^2}+h(R)|\nab R|^2 (\f{1}{2R^2}+\f{c}{R^2})\\
\end{split}
\end{equation*}
with $|c, c_a|\leq {\color{black}1/a^{\f14}}$. {\color{black}Notice the coefficient of $|\nab R|^2$ is positive!}

In the same manner, we deal with other terms. Together with (\ref{weightedBochner}), for $|c, c_a|\leq {\color{black}1/a^{\f14}}$ we obtain
\begin{equation*}
\begin{split}
&\D'_M\l h(R)|\nab R|^2\r\\
=&h'(R)|\nab R|^2\cdot(-\f12\O \tr\chib |\nab R|^2+\f{\O^{-1}}{2}\tr\chi-2\eta^b \nab_b R-4\O \omb |\nab R|^2)\\
&+h''(R)|\nab R|^4+\f{2h'(R)\nab'^a R}{h(R)}\cdot \nab'_{a}\l h(R)|\nab R|^2\r-\f{2h'(R)h'(R)}{h(R)}|\nab R|^4\\
&+h(R)\l 2|\nab{\color{black}'}^2 R|^2+2{\color{black}\mbox{Ric}_{M}}(\nab R, \nab R)\r\\
&+h(R)\cdot 2\nab'^a R \nab'_a (-\f12\O \tr\chib |\nab R|^2+\f{\O^{-1}}{2}\tr\chi-2\eta^b \nab_b R-4\O \omb |\nab R|^2)\\
\geq&h(R)\l 2|\nab{\color{black}'}^2 R|^2+2{\color{black}\mbox{Ric}_{M}}(\nab R, \nab R)\r\\
&+\l \f{h'(R)}{R}+h''(R)-\f{2h'(R)h'(R)}{h(R)}-\f{2h(R)}{R^2}-\f{2h'(R)}{R}\r \cdot (1+c) \cdot |\nab R|^4\\
&+\l h'(R)\cdot \f{\O^{-1}}{2}\tr\chi+h(R)(\f{1}{2R^2}+\f{c}{R^2})\r |\nab R|^2 \\
&+\f{2h'(R)\nab'^a R}{h(R)}\cdot \nab'_{a}\l h(R)|\nab R|^2\r-({\O\tr\chib+8\O\omb)\nab'^a R}\cdot\nab'_a(h(R)|\nab R|^2)\\
&-h(R)|\nab R|^2\nab^a R \f{c_a}{R^2}+h'(R)|\nab R|^2\nab^a R \f{c_a}{R}+h(R)\nab^a R \f{c_a}{R^2}-4h(R)\nab^a R \nab'_a\nab'_b R\,\eta^b.
\end{split}
\end{equation*}

Recall 
$$h(R)=1+\f{8}{\ub^2 a^2}(R-\f{\ub a}{2})^2.$$ 
With $C^0$ estimates $|R-\f{\ub a}{2}|\leq \f{\ub a}{20}$ and the estimates in \cite{AL}, it is straightforward to check
\begin{equation*}
\begin{split}
& h'(R)=\f{16}{\ub^2 a^2} (R-\f{\ub a}{2}), \quad \quad h''(R)=\f{16}{\ub^2 a^2}, \quad \quad |h(R)-1|\leq \f{1}{50},\\
&|h'(R)|\leq \f{4}{5\ub a}, \quad \quad |\O^{-1}\tr\chi|\leq \f{3\ub a}{20 R^2},\\
&\l \f{h'(R)}{R}+h''(R)-\f{2h'(R)h'(R)}{h(R)}-\f{2h(R)}{R^2}-\f{2h'(R)}{R}\r \cdot (1+c) \geq \f{2}{\ub^2 a^2},\\
& \l h'(R)\cdot \f{\O^{-1}}{2}\tr\chi+h(R)(\f{1}{2R^2}+\f{c}{R^2})\r\geq \f{1}{6R^2}.
\end{split}
\end{equation*}

For $2 {\color{black}\mbox{Ric}_{M}}(\nab R, \nab R)$ term, relying on an estimate (\ref{RicciBound}) in appendix we have
\begin{equation*}
\begin{split}
2{\color{black}\mbox{Ric}_{M}}(\nab R, \nab R)\geq& -\f{\ub \at}{R}\cdot\f{|\nab R|^4}{R^2}-\f{\ub \at}{R^2} \cdot|\nab^2 R|\cdot |\nab R|^2-\f{\ub a}{R^3}\cdot\f{\ub \at}{R}\cdot|\nab R|^3.
\end{split}
\end{equation*}
Here from
\begin{equation*}
\begin{split}
&\nab'_a \nab'_b R=\nab'_a \nab_b R\\
=&\nab_a\nab_b-\O\nab_a R \nab_3\nab_b R\\
=&\nab_a\nab_b R+\O\chibh_{bc}\nab_a R\nab^c R+\f12\O\tr\chib\nab_a R \nab_b R,
\end{split}
\end{equation*}
we have 
$$\nab_a\nab_b R=\nab'_a \nab'_b R-\O\chibh_{bc}\nab_a R\nab^c R-\f12\O\tr\chib\nab_a R \nab_b R.$$
Therefore, 
$$-\f{\ub \at}{R^2} \cdot|\nab^2 R|\cdot |\nab R|^2=-\f{\ub \at}{R^2} \cdot|\nab'^2 R|\cdot |\nab R|^2+\f{c}{R^2}|\nab R|^4,$$
{\color{black}
where
$$\f{|c|}{R^2}\leq |\f{\ub \at}{R^2}\O\tr\chib|+|\f{\ub\at}{R^2}\O\chibh|,$$
and
$$|c|\leq|\f{\ub\at}{R}|+|\ub\at\f{\ub\at b^{\f14}}{R^2}|\lesssim\f{1}{\at}+\f{b^{\f14}}{a}\lesssim \f{1}{\at}\ll 1.$$
}

To handle $\nab'^2 R$ terms, we notice that 
\begin{equation*}
\begin{split}
&|\nab'^2 R|^2-4 \nab^a R \nab'_a\nab'_b R \, \eta^b\\
\geq &|\nab'^2 R|^2-4 |\nab R| |\nab'^2 R| |\eta|\\
\geq & -4|\nab R|^2 |\eta|^2\\
=& -\f{c}{R^2}|\nab R|^2, 
\end{split}
\end{equation*}
and
\begin{equation*}
\begin{split}
&|\nab'^2 R|^2-\f{\ub\at}{R^2} |\nab R| |\nab'^2 R|\\
\geq & -\f{\ub^2 a}{4R^4}|\nab R|^2\\
=& -\f{c}{R^2}|\nab R|^2,
\end{split}
\end{equation*}
where $0<c\leq {\color{black}1/a^{\f14}}$.  {\color{black}Thus, all $\nab'^2 R$ terms are bounded from below by $|\nab R|^2$ terms.}

In sum, we hence have 
\begin{equation*}
\begin{split}
\D'_M \l h(R)|\nab R|^2\r\geq&|\nab R|^4\cdot\f{1}{\ub^2 a^2}+|\nab R|^2\cdot\f{1}{8R^2}-\f{1}{R^2}\cdot {\color{black}O\big(\f{1}{a^{\f12}}\big)}.
\end{split}
\end{equation*}

Denote $[h(R)|\nab R|^2](\tilde{\o}_{max}):=\max_{\o\in S^2} [h(R(\o))|\nab R(\o)|^2]$.  We hence derive 
$$0\geq \D'_{M}[h(R)|\nab R|^2](\tilde{\o}_{max})\geq |\nab R|^4(\tilde{\o}_{max})\f{1}{\ub^2 a^2}+|\nab R|^2(\tilde{\o}_{max})\f{1}{8R^2}-\f{1}{R^2}\cdot  {\color{black}O\big(\f{1}{a^{\f12}}\big)}.$$
This implies
$$|\nab R|^2(\tilde{\o}_{max})\lesssim {\color{black}\f{1}{a^{\f14}}}\ll1.$$
Let $|\nab R|(\o_{max}):=\max_{\o\in S^2} |\nab R|$.  We derive
\begin{equation*}
\begin{split}
|\nab R|^2 (\o_{max})=&\f{1}{[h(R)](\o_{max})}[h(R)|\nab R|^2](\o_{max})\\
\leq& \f{1}{[h(R)](\o_{max})}[h(R)|\nab R|^2](\tilde{\o}_{max})\\
\leq& \f{50}{49}[h(R)|\nab R|^2](\tilde{\o}_{max}){\color{black}\lesssim \f{1}{a^{\f14}}}\ll 1.
\end{split}
\end{equation*}

Therefore, we conclude 
\begin{equation}
|\nab R| (\o) {\color{black}\lesssim \f{1}{a^{\f18}}}\ll1 \quad \mbox{for all} \quad \o\in S^2. 
\end{equation}

\subsection{$W^{2,p}$ Estimate for $p<\infty$}

On $M$: $u=1-R(\theta_1, \theta_2)$, the induced metric reads
$$g'_{\theta_i \theta_j}=g_{\theta_i \theta_j}+\f{\partial (1-R)}{\partial \theta_i} \f{\partial{(1-R)}}{\partial \theta_j} g(\f{\partial}{\partial u},\f{\partial}{\partial u})=g_{\theta_i, \theta_j}.$$
Since
$$\D'_M R=\f{1}{\sqrt{|g'|}}\f{\partial}{\partial \theta_k} \l \sqrt{|g'|}g'^{\theta_k \theta_l}\f{\partial}{\partial \theta_l} R \r,$$
the leading term for $\D'_M R(\o)$ is $g^{\theta_i \theta_j} \f{\partial^2}{\partial \theta_i \partial \theta_j}R$.  

Let $\{j_{\alpha}\}$ be the partition of unity on $M$: $\sum_{\alpha}j_{\alpha}=1$. We deduce
\begin{equation*}
\begin{split}
&\D'_M (j_{\alpha}R)=\f{1}{\sqrt{|g'|}}\f{\partial}{\partial \theta_k} \l \sqrt{|g'|}g'^{\theta_k \theta_l}\f{\partial}{\partial \theta_l} (j_{\alpha} R) \r\\
=&(\D'_M j_{\alpha}) R+\f{1}{\sqrt{|g'|}}\f{\partial}{\partial \theta_k} (\sqrt{|g'|}g'^{\theta_k \theta_l} R)\f{\partial j_{\alpha}}{\partial \theta_l}+j_{\alpha} \D'_M R+\f{1}{\sqrt{|g'|}}\f{\partial}{\partial \theta_k} \l \sqrt{|g'|}g'^{\theta_k \theta_l} j_{\alpha} \r \f{\partial R}{\partial \theta_l}. 
\end{split}
\end{equation*}
{\color{black}
\noindent Hence, 
\begin{equation*}
\begin{split}
g'^{\theta_k \theta_l}&\f{\partial^2}{\partial \theta_k \partial \theta_l}(j_{\a}R)=(\D'_M j_{\alpha}) R+\f{1}{\sqrt{|g'|}}\f{\partial}{\partial \theta_k} (\sqrt{|g'|}g'^{\theta_k \theta_l} R)\f{\partial j_{\alpha}}{\partial \theta_l}\\
&+j_{\alpha} \D'_M R+\f{1}{\sqrt{|g'|}}\f{\partial}{\partial \theta_k} \l \sqrt{|g'|}g'^{\theta_k \theta_l} j_{\alpha} \r \f{\partial R}{\partial \theta_l}-\f{1}{\sqrt{|g'|}}\f{\partial}{\partial \theta_k} \l \sqrt{|g'|}g'^{\theta_k \theta_l}\r \f{\partial (j_{\alpha} R)}{\partial \theta_l}. 
\end{split}
\end{equation*}
Notice $g'^{\theta_k \theta_l}=g^{\theta_k \theta_l}\approx \ub^{-2}a^{-2}\approx R(\o)^{-2}$. To make sure that the coefficients of elliptic operator are uniformly bounded of size $1$, we use 
\begin{equation*}
\begin{split}
R&(\o)^2\cdot g'^{\theta_k \theta_l}\f{\partial^2}{\partial \theta_k \partial \theta_l}(j_{\a}R)=R(\o)^2\cdot(\D'_M j_{\alpha}) R+\f{R(\o)^2}{\sqrt{|g'|}}\f{\partial}{\partial \theta_k} (\sqrt{|g'|}g'^{\theta_k \theta_l} R)\f{\partial j_{\alpha}}{\partial \theta_l}\\
&+j_{\alpha}\cdot R(\o)^2 \cdot \D'_M R+\f{R(\o)^2}{\sqrt{|g'|}}\f{\partial}{\partial \theta_k} \l \sqrt{|g'|}g'^{\theta_k \theta_l} j_{\alpha} \r \f{\partial R}{\partial \theta_l}-\f{R(\o)^2}{\sqrt{|g'|}}\f{\partial}{\partial \theta_k} \l \sqrt{|g'|}g'^{\theta_k \theta_l}\r \f{\partial (j_{\alpha} R)}{\partial \theta_l}. 
\end{split}
\end{equation*}
Together with the fact $|\nab R(\o)|\lesssim 1/a^{\f18}\leq 1$, to use $W^{2,p}$
elliptic estimate, we only need to check that $R(\o)^2\cdot g'^{\theta_k \theta_l}\in C(M)$. We hence only need to verify $e_a'(g'_{\theta_k \theta_l})$ is in $L^{\infty}(M)$, which implies $g'^{\theta_k \theta_l}\in C(M)$ and $R(\o)^2\cdot g'^{\theta_k \theta_l}\in C(M)$.    
\noindent Since $e_a'=e_a-\O e_a(R)e_3$ and $g'_{\theta_k \theta_l}=g_{\theta_k \theta_l}$, we have
\begin{equation}
\begin{split}
e_a'(g'_{\theta_k \theta_l})=&\big(e_a-\O e_a(R)e_3\big)g_{\theta_k \theta_l}=e_a(g_{\theta_k \theta_l})-\O e_a(R)e_3(g_{\theta_k \theta_l})\\
=&e_a(g_{\theta_k \theta_l})-4\O e_a(R)\cdot\chib_{\theta_k \theta_l}.
\end{split}
\end{equation}
From the estimate in \cite{AL} and $|\nab R(\o)|$ is uniformly bounded, we have $e_a'(g'_{\theta_k \theta_l})\in L^{\infty}(M)$. 
In scale-invariant $\widetilde{W}^{2,p}(M)$ norms, we have 
\begin{equation*}
\begin{split}
\|j_{\alpha} R\|_{\widetilde{W}^{2,p}(M)}\lesssim& \|j_{\alpha} R\|_{\tilde{L}^p(M)}+\|j_{\alpha}\cdot R(\o)^2\cdot \D'_M R\|_{\tilde{L}^p(M)}+\|\f{R(\o)^2}{\sqrt{|g'|}}\f{\partial}{\partial \theta_k} (\sqrt{|g'|}g'^{\theta_k \theta_l} R)\f{\partial j_{\alpha}}{\partial \theta_l}\|_{\tilde{L}^p(M)}\\
&+\|R(\o)^2\cdot(\D'_M j_{\alpha}) R\|_{\tilde{L}^p(M)}+\|\f{R(\o)^2}{\sqrt{|g'|}}\f{\partial}{\partial \theta_k} \l \sqrt{|g'|}g'^{\theta_k \theta_l} j_{\alpha} \r \f{\partial R}{\partial \theta_l}\|_{\tilde{L}^p(M)}\\
&+\|\f{R(\o)^2}{\sqrt{|g'|}}\f{\partial}{\partial \theta_k} \l \sqrt{|g'|}g'^{\theta_k \theta_l}  \r \f{\partial (j_{\alpha} R)}{\partial \theta_l}\|_{\tilde{L}^p(M)}.
\end{split}
\end{equation*}
Together with (\ref{3.3}) and the derived $C^1$ estimate $\|\nab R\|_{L^{\infty}(M)}\leq 1$, we conclude
\begin{equation}\label{3.24new}
\begin{split}
&\|R(\o)\|_{\widetilde{W}^{2,p}(M)}=\|\bigg(\sum_{\a}j_{\a}R(\o)\bigg)\|_{\widetilde{W}^{2,p}(M)}\\
\lesssim& \|R(\o)^2
\cdot\D'_M R(\o)\|_{\tilde{L}^p(M)}+\|R(\o)\|_{\widetilde{W}^{1,p}(M)}\\
\approx& R(\o)^{-\f{2}{p}}\cdot\|R(\o)^2\cdot\D'_M R(\o)\|_{L^p(M)}+\|R(\o)\|_{\widetilde{W}^{1,p}(M)}\\
\lesssim& \f{R(\o)^{-\f{2}{p}+2}}{R(\o)}\|\nab R(\o)\|^2_{L^{2p}(M)}+R(\o)^{-\f{2}{p}+2}\cdot\|\f{1}{R(\o)}\|_{L^p(M)}+\f{R(\o)^{-\f{2}{p}+2}}{R(\o)}\|\nab R(\o)\|_{L^p(M)}+R(\o)\\
\lesssim& R(\o). 
\end{split}
\end{equation}
}

\subsection{$C^{1,q}$ Estimate for $0<q<1$}
Let's first focusing on finding a MOTS $M:  u=1-R(\o)$ among 2-spheres with positive injectivity radius and bounded curvatures\footnote{In Section 3.5, we prove that the solution we construct through the method of continuity
satisfies these requirements.}. Under these restrictions, for any $q\in (0,1)$ Sobolev's inequality (Theorem 2.21 in \cite{Au})  implies
\begin{equation}\label{3.21}
\begin{split}
\|R(\o)\|_{C^{1,q}(M)}\lesssim {\color{black}R(\o)^{-1-q}\cdot \|R(\o)\|_{\widetilde{W}^{2,\f{2}{1-q}}(M)}}\lesssim R(\o)^{-q}.  
\end{split}
\end{equation}

\subsection{$C^{2,q}$ Estimate for $0<q<1$}
By Morrey's inequality (Lemma 2.22 in \cite{Au}) and Sobolev inequality (Theorem 2.21 in \cite{Au}), we have
$$\|j_{\alpha}R\|_{C^{0, q}(M)}\leq \|\f{\partial}{\partial \theta_i} (j_{\alpha}R)\|_{L^s(M)}, \quad \mbox{for} \quad s=\f{2}{1-q},$$
$$\|\f{\partial}{\partial \theta_i} (j_{\alpha} R)\|_{C^{0, q}}\leq \|\f{\partial^2}{\partial \theta_j \partial \theta_k} (j_{\alpha}R)\|_{L^s(M)}, \quad \mbox{for} \quad s=\f{2}{1-q},$$
$$\|\f{\partial}{\partial \theta_i} (j_{\alpha}R)\|_{C^0(M)}\leq \|\f{\partial^2}{\partial \theta_j \partial \theta_k} (j_{\alpha}R)\|_{L^s(M)}, \quad \mbox{for} \quad s>n.$$
Together with the bound (\ref{3.24new}) for $W^{2,p}$ estimate, we have $R$ is bounded in $C^{1,q}(M)$ norm.

For (\ref{LR=0}):
\begin{equation*}
\begin{split}
&\D'_M R(\o)+2\eta_b\nab^b R(\o)+\f12 \O \tr\chib |\nab R(\o)|^2\\
&+4\O \omb |\nab R(\o)|^2-\f{\O^{-1}}{2} \tr\chi\\
=&0.
\end{split}
\end{equation*}  
rewrite this equation in coordinate. Since we have derived the $C^{1, q}(M)$ bound of $R(\o)$, $C^{2,q}$ estimate for $\D'_M$ on $M$ follows from partition of unity and the standard Schauder's estimate.  
Together with (\ref{3.21}) and regularity theory of elliptic equations, we conclude
\begin{equation}
\|R(\o)\|_{C^{2,q}(M)}\lesssim R(\o)^{-1-q},
\end{equation}
and 
\begin{equation}\label{higher regularity}
\|R(\o)\|_{C^{k,q}(M)}\lesssim R(\o)^{1-k-q}, \quad \mbox{for all} \quad k\in \mathbb{N}.
\end{equation}

\begin{remark}
(\ref{higher regularity}) implies 
$$\|\nab^3 R\|_{C^0(M)}\leq {1}/{R^2} \quad \quad \mbox{and} \quad \quad \|\nab^4 R\|_{C^0(M)}\leq {1}/{R^3}.$$
Since $\|\nab R\|_{C^0(M)}\ll 1$, by interpolation inequalities we get
\begin{equation*}
\|\nab^2 R\|_{C^0(M)}\ll {1}/{R}.
\end{equation*}
Recall that $\f13 \ub a\leq R(\o)\leq \f23 \ub a$. In coordinates for $i,j=1,2$, we have
\begin{equation}\label{nab nab R}
|\f{\partial^2}{\partial \theta_i \partial \theta_j} R|\ll \ub a.
\end{equation}
\end{remark}

\section{Continuity Argument I} \label{Continuity1}
In this paper, we will apply the method of continuity twice. For this section, we first construct a solution to the following equation:
\begin{equation}
\D'_M R(\o)+\f12\O\tr\chib|\nab R(\o)|^2-\f{1}{R(\o)}+\f{\ub a f(\ub, \o)}{2 R(\o)^2}=0.\\
\end{equation}

We introduce $\lambda$ and let 
$$F(R(\o),\lambda):=\D'_M R(\o)+\f12\O\tr\chib|\nab R(\o)|^2-\f{1}{R(\o)}+\f{\ub a}{2 R(\o)^2}[1+(f(\ub, \o)-1)\lambda].$$
When $\lambda=0$, we have that $R(\o)={\ub a}/{2}$ is a solution to 
$$F(R(\o),0)=\D'_M R(\o)+\f12\O\tr\chib|\nab R(\o)|^2-\f{1}{R(\o)}+\f{\ub a}{2 R(\o)^2}=0.$$
For $0\leq \tilde{\lambda} \leq 1$, we assume that $\tilde{R}(\o)$ is a solution to 
$$\D'_{S_{1-\tR,\ub}}\tR+\f12\l\O\tr\chib|_{1-\tR,\ub}\r|\nab \tR(\o)|^2-\f{1}{\tR(\o)}+\f{\ub a}{2\tR^2(\o)}[1+(f(\ub, \o)-1)\tilde{\lambda}]=0. $$
Therefore, 
\begin{equation}
\begin{split}
&F_R(\tR(\o),\lambda)[W]\\
=&\lim_{\epsilon\rightarrow 0}\f{1}{\epsilon}\l F(\tR+\epsilon W, \lambda)-F(\tR, \lambda) \r\\
=&\lim_{\epsilon\rightarrow 0}\f{1}{\epsilon}\l \D'_{S_{1-(\tR+\epsilon W),\ub}}(\tR+\epsilon W)+\f12 \big(\O\tr\chib|_{1-(\tR+\epsilon W), \ub}\big) |\nab (\tR+\epsilon W)|^2\\
&-\f{1}{\tR+\epsilon W}+\f{\ub a}{2(\tR+\epsilon W)^2}[1+(f(\ub, \o)-1){\lambda}] \\
&\quad \quad \quad -\D'_{S_{1-\tR,\ub}}\tR -\f12 \big(\O\tr\chib|_{1-\tR, \ub}\big)|\nab \tR|^2+\f{1}{\tR}-\f{\ub a}{2\tR^2}[1+(f(\ub, \o)-1){\lambda}] \r\\
=&\lim_{\epsilon\rightarrow 0}\f{1}{\epsilon}\l \D'_{S_{1-(\tR+\epsilon W),\ub}}(\tR+\epsilon W)-\D'_{S_{1-\tR,\ub}} \tR\r\\
&+\lim_{\epsilon\rightarrow 0}\f{1}{\epsilon} \l \f12 \big(\O\tr\chib|_{1-(\tR+\epsilon W), \ub}\big)|\nab(\tR+\epsilon W)|^2-\f12 \big(\O\tr\chib|_{1-\tR, \ub}\big)|\nab \tR|^2\r\\
&+\l \f{W}{\tR^2}-\f{W}{\tR^3} \ub a [1+(f(\ub, \o)-1)\lambda]\r\\
=&I_1+I_2+I_3.  
\end{split}
\end{equation}

Before we list the detailed description of $I_1, I_2$ and $I_3$, let's first state three useful propositions. 
\begin{proposition}\label{commute}
With commutator formulas in \cite{KR:Trapped}, for $U$ a vector field and $f$ a scalar function, we have
\begin{equation}
\begin{split}
[\nab_3, \div ]U=&-\f12 \tr\chib \div U-\chibh\cdot\nab U+\beb\cdot U+\f12(\eta+\etb)\cdot \nab_3 U\\
&-\eta\cdot\chibh\cdot U+\f12 \tr\chib \eta \cdot U,
\end{split}
\end{equation}
and
\begin{equation}
\begin{split}
[\nab_3, \nab]f=\f12(\eta+\etb)D_3 f-\chib\cdot\nab f. 
\end{split}
\end{equation}
\end{proposition}

If we let $U=\nab f$, it follows
\begin{proposition}\label{commute laplacian}
\begin{equation}
\begin{split}
[D_3, \D]f=&-\f12 \tr\chib \D f-\chibh\cdot\nab^2 f+\beb\cdot\nab f+\f12(\eta+\etb)\cdot \nab_3 \nab f-\eta\cdot\chibh\cdot \nab f\\
&+\f12 \tr\chib\eta\cdot\nab f+\div\l\f12(\eta+\etb)D_3 f\r-\div(\chib\cdot\nab f)\\
=&-\tr\chib \D f-2\chibh\cdot\nab^2 f+\beb\cdot\nab f+\f12(\eta+\etb)\cdot \nab_3 \nab f-\eta\cdot\chibh\cdot \nab f\\
&+\f12 \tr\chib\eta\cdot\nab f+\div\l\f12(\eta+\etb)D_3 f\r-\div \chib\cdot \nab f.
\end{split}
\end{equation}
\end{proposition}

With the help of these two commutation lemmas, we further have
\begin{proposition}\label{limit laplacian}
\begin{equation}
\begin{split}
&\lim_{\epsilon\rightarrow 0}\f{1}{\epsilon} \l \D'_{S_{1-\tR-\epsilon W, \ub}}\tR-\D'_{S_{1-\tR, \ub}}\tR \r\\
=&-\O\nab_3\l\D_{S_{1-\tR, \ub}}\tR\r W\\
&-\O\nab_3\l 2\O \chibh_{ab}\nab^a \tR \nab_b \tR\r W.
\end{split}
\end{equation}
\end{proposition}
\begin{proof}
Set $g$ and $\theta_1, \theta_2$ be the induced metric and independent angular variables on $S_{u,\ub}$. Give a function $f$, we have
\begin{equation}\label{3.39}
\begin{split}
\D_{S_{u,\ub}} f=&\f{1}{\sqrt{\det g}}\f{\partial}{\partial \theta_i} (\sqrt{\det g}\,g^{\theta_i \theta_l}\f{\partial f}{\partial \theta_l})\\
=&g^{\theta_1 \theta_1}\f{\partial^2 f}{\partial \theta_1 \partial \theta_1}+g^{\theta_2 \theta_2}\f{\partial^2 f}{\partial \theta_2 \partial \theta_2}+2 g^{\theta_1 \theta_2}\f{\partial^2 f}{\partial \theta_1 \partial \theta_2}\\
&+\f{\partial}{\partial \theta_1} (g^{\theta_1 \theta_1})\f{\partial f}{\partial \theta_1}+\f{\partial}{\partial \theta_1} (g^{\theta_1 \theta_2})\f{\partial f}{\partial \theta_2}\\
&+\f{\partial}{\partial \theta_2} (g^{\theta_2 \theta_1})\f{\partial f}{\partial \theta_1}+\f{\partial}{\partial \theta_2} (g^{\theta_2 \theta_2})\f{\partial f}{\partial \theta_2}\\
&+\f12 g^{\theta_k \theta_j}\f{\partial g_{\theta_j \theta_k}}{\partial \theta_i} g^{\theta_i \theta_l}\f{\partial f}{\partial \theta_l},
\end{split}
\end{equation}
where $i,j,k,l=1,2$ and $g^{\theta_k \theta_j}$ depends on $u$. Here we also use the formula
$$\f{\partial}{\partial \theta_i}\det g=\det g \cdot g^{\theta_k \theta_j} \cdot \f{\partial g_{\theta_j \theta_k}}{\partial \theta_i}.$$
Combining $\partial \tR/\partial u=0$ and Proposition \ref{change laplacian}, we deduce
\begin{equation*}
\begin{split}
&\lim_{\epsilon\rightarrow 0}\f{1}{\epsilon} \l \D'_{S_{1-\tR-\epsilon W, \ub}}\tR-\D'_{S_{1-\tR, \ub}}\tR \r \\
=&\lim_{\epsilon\rightarrow 0}\f{1}{\epsilon} \l \D_{S_{1-\tR-\epsilon W, \ub}}\tR-\D_{S_{1-\tR, \ub}}\tR \r \\
&+\lim_{\epsilon\rightarrow 0}\f{1}{\epsilon} \l 2\O \chibh_{ab}  \nab^a \tR\nab^b \tR|_{S_{1-\tR-\epsilon W, \ub}}- 2\O \chibh_{ab}  \nab^a \tR\nab^b\tR|_{S_{1-\tR, \ub}} \r\\
=&-\f{\partial}{\partial u}\l\D_{S_{1-\tR, \ub}}\tR\r W+\D_{S_{1-\tR, \ub}}\l\f{\partial}{\partial u}\tR\r W\\
&-\O\nab_3 \l 2\O \chibh_{ab}\nab^a \tR\nab^b \tR \r W\\
=&-\O\nab_3\l\D_{S_{1-\tR, \ub}}\tR\r W\\
&-\O\nab_3\l 2\O \chibh_{ab}\nab^a \tR \nab^b \tR\r W.\\
\end{split}
\end{equation*}
\end{proof}

Together with Propositions \ref{commute} and \ref{commute laplacian}, we thus have
\begin{equation*}
\begin{split}
I_1=&\lim_{\epsilon\rightarrow 0}\f{1}{\epsilon}\l \D'_{S_{1-(\tR+\epsilon W),\ub}}(\tR+\epsilon W)-\D'_{S_{1-\tR,\ub}} \tR\r\\
=&-\O \nab_3 (\D_{S_{1-\tR, \ub}}\tR)W-\O\nab_3(2\O\chibh_{ab}\nab^a \tR \nab^b \tR)W+\D'_{S_{1-\tR, \ub}}W\\
=&\O ([\D_{S_{1-\tR, \ub}}, \nab_3]\tR)W-\O\nab_3(2\O\chibh_{ab}\nab^a \tR \nab^b \tR)W+\D'_{S_{1-\tR, \ub}}W\\
=&(\O\tr\chib\D_{S_{1-\tR,\ub}}\tR)W+2(\O\chibh\cdot\nab^2 \tR)W-(\O\beb\cdot\nab\tR)W\\
&-\f12\O(\eta+\etb)\cdot(-\f12\tr\chib\nab \tR-\chibh\cdot\nab\tR)W+(\O\eta\cdot\chibh\nab\tR-\f12\O\tr\chib\eta\cdot\nab\tR+\O\div\chib\cdot\nab\tR)W\\
&-\O\nab_3(2\O\chibh_{ab}\nab^a \tR \nab^b \tR)W+\D'_{S_{1-\tR, \ub}}W,\\
&\\
I_2=&\lim_{\epsilon\rightarrow 0}\f{1}{\epsilon} \l \f12\big(\O\tr\chib|_{1-(\tR+\epsilon W),\ub} \big)|\nab(\tR+\epsilon W)|^2-\f12\big(\O\tr\chib|_{1-\tR,\ub} \big)|\nab \tR|^2\r\\
=&(\O\tr\chib)|_{1-\tR,\ub}\cdot\nab\tR\cdot\nab W-\f12\f{\partial}{\partial u}(\O\tr\chib)|_{1-\tR,\ub}\cdot|\nab \tR|^2 W\\
&+{\color{black}\f12\O\tr\chib|_{1-\tR,\ub}\cdot\f{\partial g^{{\color{black}ab}}}{\partial u}|_{1-\tR,\ub}\cdot\nab_{{\color{black}a}} \tR \cdot\nab_{{\color{black}b}} \tR\cdot (-W)}\\
=&(\O\tr\chib)|_{1-\tR,\ub}\cdot\nab\tR\cdot\nab W-\f12\f{\partial}{\partial u}(\O\tr\chib)|_{1-\tR,\ub}\cdot|\nab \tR|^2 W\\
&{\color{black}+}\f12(\O\tr\chib)^2|_{1-\tR,\ub}\cdot|\nab \tR|^2 W{\color{black}+}\O^2\tr\chib\,\chibh^{{\color{black}ab}}|_{1-\tR,\ub}\cdot \nab_{{\color{black}a}} \tR \cdot \nab_{{\color{black}b}} \tR \cdot W,\\
&\\
I_3=&\f{W}{\tR^2}-\f{W}{\tR^3}\ub a[1+(f(\ub, \o)-1)\lambda].
\end{split}
\end{equation*}
{\color{black}In $I_2$, we use $\partial g^{ab}/\partial u=-2\O \chib^{ab}$. And in the following, we will} focus on the coefficients in front of $W$. Applying the estimates in \cite{AL}, in $I_1$ all the terms containing $\chibh, \beb, \eta, \etb, \div\chib$ {\color{black}are much smaller than $1/\tR^2$, which is the coefficient of $W$ in the first term of $I_3$. And hence these terms} could be considered as lower order terms (l.o.t.). Thus, we have
$$I_1=(\O\tr\chib\D_{S_{1-\tR,\ub}}\tR)W+\D'_{S_{1-\tR, \ub}}W+\mbox{l.o.t.}.$$
Since
 $$\O=1+\mbox{l.o.t.}, \quad \quad \tr\chib=-\f{2}{R}+\mbox{l.o.t.}, \quad \quad -\f12\f{\partial}{\partial u}(\O\tr\chib)=\f14 (\tr\chib)^2+\mbox{l.o.t.},$$
{\color{black} and $\chibh_{ab}$ could be treated as l.o.t.}, then for $I_1+I_2+I_3$ the coefficients in front of $W$ are 
$$\f{1}{\tR^3}\l-2\tR^2 \D_{S_{1-\tR,\ub}}\tR+\tR|\nab\tR|^2{\color{black}+}2\tR|\nab\tR|^2+\tR-\ub a[1+(f(\ub, \o)-1)\lambda] \r\cdot [1+o(1)].$$
Recall that 
$$0=\D_{S_{1-\tR, \ub}}\tR+\O\tr\chib|\nab \tR|^2-\f12(\O\tr\chib)|\nab \tR|^2+2\O\chibh_{ab}\nab^a\tR \nab^b \tR-\f{1}{\tR}+\f{\ub a}{2\tR^2}[1+(f(\ub, \o)-1)\tilde{\lambda}].$$
It follows
\begin{equation*}
\begin{split}
&\tR|\nab\tR|^2-2\tR^2\D_{S_{1-\tR,\ub}}\tR\\
=&-\tR|\nab\tR|^2\cdot[1+o(1)]+4\tR^2\O \chibh_{ab}\nab^a\tR\nab^b\tR-2\tR+\ub a[1+(f(\ub, \o)-1)\tilde{\lambda}].
\end{split}
\end{equation*}
Therefore, for $I_1+I_2+I_3$, the coefficient in front of $W$ are 
$$\f{1}{\tR^3}\l -\tR-\tR|\nab \tR|^2{\color{black}+}2\tR|\nab\tR|^2+4\O\tR^2\chibh_{ab}\nab^a \tR \nab^b \tR+\ub a (f(\ub, \o)-1)(\tilde{\lambda}-\lambda) \r \cdot [1+o(1)].$$
By $C^1$ a priori estimate, we have proved $|\nab \tR|\ll 1$. For $|f(\ub, \o)-1|\leq 1/21$, when $\tilde{\lambda}$ and $\lambda$ are close, the coefficients in front of $W$ is negative. {\color{black} For $q\in (0,1)$, by a priori estimates in Section \ref{a priori estimates section}, it holds that $\tR(\o)\in C^{2,q}(\mathbb{S}^2)$. According to the solvability condition of elliptic equations
($W$ has negative coefficient), we further have} operator $F_R(\tR(\o),\lambda)[W]{\color{black}: C^{2, q}(\mathbb{S}^2)\rightarrow C^{{\color{black}0},q}(\mathbb{S}^2)}$ is invertible for $W$ when $\lambda$ close to $\tilde{\lambda}$.
Together with continuity argument for $0\leq \lambda\leq 1$, when $\lambda=1$ we thus obtain a solution for
$$\D'_M R(\o)+\f12\O\tr\chib|\nab R(\o)|^2-\f{1}{R(\o)}+\f{\ub a f(\ub, \o)}{2 R(\o)^2}=0.$$

\section{Continuity Argument II}\label{Continuity2}
In Section \ref{Continuity1}, we have constructed one solution for
\begin{equation*}
\D'_M R(\o)+\f12 \O\tr\chib|\nab R(\o)|^2-\f{1}{R(\o)}+\f{\ub a f(\ub, \o)}{2R(\o)^2}=0.
\end{equation*}
Based on it, in this section we employ the method of continuity for one more time and we solve equation (\ref{LR=0}):
\begin{equation*}
\begin{split}
&\D'_M R(\o)+2\eta_b\nab^b R(\o)+\f12 \O \tr\chib |\nab R(\o)|^2\\
&+4\O \omb |\nab R(\o)|^2-\f{\O^{-1}}{2} \tr\chi\\
=&0.
\end{split}
\end{equation*}  
We construct $G(R(\o), \lambda)$ through
\begin{equation*}
\begin{split}
G(R(\o), \lambda)=&\D'_M R(\o)+\f12 \O\tr\chib|\nab R(\o)|^2-\f{1}{R(\o)}+\f{\ub a f(\ub, \o)}{2R(\o)^2}\\
&+\lambda \l 2\eta_b\nab^b R(\o)+4\O \omb |\nab R(\o)|^2-\f{\O^{-1}}{2} \tr\chi+\f{1}{R(\o)}-\f{\ub a f(\ub, \o)}{2R(\o)^2}\r.
\end{split}
\end{equation*}
We already have a solution to $G(R(\o), 0)=0$. And a solution $R=R(\o)$ to $G(R(\o), 1)=0$ will solve (\ref{LR=0}). 

Assume $\tR(\o)$ solving 
\begin{equation*}
\begin{split}
G(\tR(\o), \tilde{\lambda})=&\D'_M \tR(\o)+\f12 \O\tr\chib|\nab \tR(\o)|^2-\f{1}{\tR(\o)}+\f{\ub a f(\ub, \o)}{2\tR(\o)^2}\\
&+\tilde{\lambda} \l 2\eta_b\nab^b \tR(\o)+4\O \omb |\nab \tR(\o)|^2-\f{\O^{-1}}{2} \tr\chi+\f{1}{\tR(\o)}-\f{\ub a f(\ub, \o)}{2\tR(\o)^2}\r\\
=&0.
\end{split}
\end{equation*}
Therefore, 
\begin{equation}
\begin{split}
&G_R(\tR(\o),\lambda)[W]\\
=&\lim_{\epsilon\rightarrow 0}\f{1}{\epsilon}\l G(\tR+\epsilon W, \lambda)-G(\tR, \lambda) \r\\
=&\lim_{\epsilon\rightarrow 0}\f{1}{\epsilon}\l \D'_{S_{1-(\tR+\epsilon W),\ub}}(\tR+\epsilon W)-\D'_{S_{1-\tR,\ub}} \tR\r\\
&+\lim_{\epsilon\rightarrow 0}\f{1}{\epsilon} \l\f12 (\O\tr\chib)|_{1-(\tR+\epsilon W), \ub}|\nab(\tR+\epsilon W)|^2-\f12 (\O\tr\chib)|_{1-\tR,\ub} |\nab \tR|^2\r\\
&+\lim_{\epsilon\rightarrow 0}\f{1}{\epsilon}\l -\f{1}{\tR+\epsilon W}+\f{\ub a f(\ub, \o)}{2(\tR+\epsilon W)^2}+\f{1}{\tR}-\f{\ub a f(\ub, \o)}{2\tR^2} \r \\
&+\lim_{\epsilon\rightarrow 0}\f{\lambda}{\epsilon}\l 2\eta_b|_{1-(\tR+\epsilon W), \ub} \nab^b (\tR+\epsilon W)-2\eta_b|_{1-\tR, \ub} \nab^b \tR \r \\
&+\lim_{\epsilon\rightarrow 0}\f{\lambda}{\epsilon}\l 4(\O\omb)|_{1-(\tR+\epsilon W), \ub} |\nab(\tR+\epsilon W)|^2-4(\O\omb)|_{1-\tR, \ub} |\nab \tR|^2 \r \\
&+\lim_{\epsilon\rightarrow 0}\f{\lambda}{\epsilon}\l -(\f{\O^{-1}}{2} \tr\chi)|_{1-(R+\epsilon W, \ub)}+\f{1}{\tR+\epsilon W}-\f{\ub a f(\ub, \o)}{2(\tR+\epsilon W)^2} \\
&\quad\quad\quad +(\f{\O^{-1}}{2} \tr\chi)|_{1-\tR,\ub}-\f{1}{\tR(\o)}+\f{\ub a f(\ub, \o)}{2\tR(\o)^2}\r \\
=&I_1+I_2+I_3+I_4+I_5+I_6.  
\end{split}
\end{equation}
{\color{black}
Here $I_1$ and $I_2$ are the same as they are in Section \ref{Continuity1}
\begin{equation}
\begin{split}
I_1=&(\O\tr\chib\D_{S_{1-\tR,\ub}}\tR)W+2(\O\chibh\cdot\nab^2 \tR)W-(\O\beb\cdot\nab\tR)W\\
&-\f12\O(\eta+\etb)\cdot(-\f12\tr\chib\nab \tR-\chibh\cdot\nab\tR)W+(\O\eta\cdot\chibh\nab\tR-\f12\O\tr\chib\eta\cdot\nab\tR+\O\div\chib\cdot\nab\tR)W\\
&-\O\nab_3(2\O\chibh_{ab}\nab^a \tR \nab^b \tR)W+\D'_{S_{1-\tR, \ub}}W,\\
\end{split}
\end{equation}
\noindent $I_2=\l -\f{2}{\tR}\nab\tR\cdot \nab W{\color{black}+3}\f{|\nab \tR|^2}{\tR^2}W \r \cdot [1+o(1)].$ 

\noindent For $I_3-I_6$, we have
}
\begin{equation*}
\begin{split}
&I_3=\f{W}{\tR^2}-\f{W}{\tR^3} \ub a f(\o),\\
&I_4+I_5+I_6=\lambda\l 2\eta_b |_{1-\tR,\ub} \nab^b W-2(\f{\partial}{\partial u}\eta_b |_{1-\tR,\ub} \nab^b R)  W\r\\
&\quad\quad\quad\quad\quad\quad+\lambda\l 8 (\O \omb) |_{1-\tR,\ub} \nab^b \tR \nab_b W-4\f{\partial}{\partial u}(\O \omb) |_{1-\tR, \ub} |\nab \tR|^2 W\r\\
&\quad\quad\quad\quad\quad\quad+\lambda\l \f{\partial}{\partial u}(\f{\O^{-1}}{2} \tr\chi) |_{1-\tR,\ub}W-\f{1}{\tR^2}W+\f{\ub a f(\o)}{\tR^3}W\r.
\end{split}
\end{equation*}
Notice that, applying (\ref{trchi u}), we have
$$\lambda\l \f{\partial}{\partial u}(\f{\O^{-1}}{2} \tr\chi) |_{1-\tR,\ub}W-\f{1}{\tR^2}W+\f{\ub a f(\o)}{\tR^3}W\r=\f{\lambda \ub \at c_3}{R^3}{\color{black}W},$$
with $|c_3|\leq b^{\f14}$. \\

Here, we focus on the coefficients in front of $W$. Applying the estimates in \cite{AL}, all the terms containing $\chibh, \beb, \eta, \etb, \div\chib, \omb$ could be considered as lower order terms (l.o.t.). Thus, we have
$$I_1=(\O\tr\chib\D_{S_{1-\tR,\ub}}\tR)W+\D'_{S_{1-\tR, \ub}}W+\mbox{l.o.t.},$$
$$I_4+I_5+I_6=\mbox{l.o.t.}.$$
Since
$$\nab_3 \O=\mbox{l.o.t.}, \quad \nab_3 \tr\chib=-\f12 (\tr\chib)^2+\mbox{l.o.t.}, \quad \nab_3 \nab R=-\f12 \tr\chib \nab R+\mbox{l.o.t.}, \quad \tr\chib=-\f{2}{R}+\mbox{l.o.t.},$$
then, for $I_1+I_2+I_3+I_4+I_5+I_6$, the coefficients in front of $W$ are 
$$\f{1}{\tR^3}\l\tR-\ub a f(\o){\color{black}{\color{black}+}2\tR|\nab\tR|^2}+\tR|\nab\tR|^2-2\tR^2 \D_{S_{1-\tR,\ub}}\tR\r\cdot [1+o(1)].$$
Recall that 
\begin{equation*}
\begin{split}
0=&\D_{S_{1-\tR, \ub}}\tR+\O\tr\chib|\nab \tR|^2-\f12\O\tr\chib|\nab \tR|^2+2\O\chibh_{ab}\nab^a\tR \nab^b \tR-\f{1}{\tR}+\f{\ub a f(\ub, \o)}{2\tR^2}\\
&+\tilde{\lambda} \l 2\eta_b\nab^b \tR(\o)+4\O \omb |\nab \tR(\o)|^2-\f{\O^{-1}}{2} \tr\chi+\f{1}{\tR(\o)}-\f{\ub a f(\ub, \o)}{2\tR(\o)^2}\r.
\end{split}
\end{equation*}
It follows
\begin{equation*}
\begin{split}
&\tR|\nab\tR|^2-2\tR^2\D_{S_{1-\tR,\ub}}\tR\\
=&-\tR|\nab\tR|^2\cdot[1+o(1)]+4\tR^2\O \chibh_{ab}\nab^a\tR\nab^b\tR-2\tR+\ub a f(\ub, \o)\\
&+2\tilde{\lambda} \tR^2 \l 2\eta_b\nab^b \tR(\o)+4\O \omb |\nab \tR(\o)|^2-\f{\O^{-1}}{2} \tr\chi+\f{1}{\tR(\o)}-\f{\ub a f(\ub, \o)}{2\tR(\o)^2}\r.
\end{split}
\end{equation*}
Recall that 
\begin{equation*}
-\f{\O^{-1}}{2} \tr\chi+\f{1}{\tR(\o)}-\f{\ub a f(\ub, \o)}{2\tR(\o)^2}=\f{c}{\tR(\o)},
\end{equation*}
with $|c|\leq \f{1}{\at}$.  

Therefore, for $I_1+I_2+I_3+I_4+I_5+I_6$, the coefficients in front of $W$ are 
$$\f{1}{\tR^3}\l -\tR{\color{black}{\color{black}+}2\tR|\nab\tR|^2}-\tR|\nab \tR|^2+4\O\tR^2\chibh_{ab}\nab^a \tR \nab^b \tR \r \cdot [1+o(1)].$$
Since $$\|{\color{black}\tR} |\nab\tR|^2+4\O\tR^2\chibh_{ab}\nab^a \tR \nab^b \tR\|_{L^{\infty}(\mathbb{S}^2)}\ll \tR,$$ for $0\leq \lambda, \tilde{\lambda}\leq 1$ and $\lambda$ close to $\tilde{\lambda}$, we have
$$G_R(\tR(\o),\lambda)[W]=\D'_{1-\tR,\ub}W-\f{1}{\tR^2}\cdot [1+o(1)]\cdot W.$$
This operator is invertible for $W$. Hence, there exists a solution $R(\o)$ for $$G(R(\o),1)=0,$$ which satisfies 
$$\tr\chi'=0.$$

\section{Bounds of Injectivity Radius for $M$} 
Here we use Cheeger's lemma to show that $M$ has an injectivity radius larger than $0$. 
\begin{lemma}[Lemma 51 in \cite{Pe}]
Given $n\geq 2$ and $v, K \in (0, \infty)$. If a compact $n$-dimensional manifold $(M, g)$ satisfies 
$$|\mbox{sec}|\leq K,$$
$$\mbox{vol}B(p,1)\geq v,$$
for all $p\in M$, then $\mbox{inj}M\geq i_0>0$, where $i_0$ depends only on $n, K$ and $v$. 
\end{lemma} 

With $C^{2,q}$ estimates for $R(\ub, \o)$, it is straightforward that the curvatures for induced metric on $M$ is bounded and $\mbox{vol}B(p,1)$ is positive. Cheeger's lemma implies that $M$ has injectivity radius with a positive lower bound.

\section{On the Uniqueness of Apparent Horizon}\label{uniqueness}
In this section, we prove that along each $\Hb_{\ub}$ the solution to (\ref{LR=0}) is unique. \\

Let's assume we have two solutions satisfying (\ref{LR=0})
\begin{equation*}
\begin{split}
0=&\D_{\tR(\o)} \tR(\o)+2\eta_b(1-\tR,\o)\nab^b \tR(\o)+\f12 (\O \tr\chib)(1-\tR,\o) |\nab \tR(\o)|^2\\
&+2(\O\chibh_{ab})(1-\tR,\o)\nab^a \tR(\o) \nab^b \tR(\o)+4(\O \omb)(1-\tR,\o) |\nab \tR(\o)|^2\\
&-\f12 (\O^{-1} \tr\chi)(1-\tR,\o),
\end{split}
\end{equation*}
and
\begin{equation*}
\begin{split}
0=&\D_{R(\o)} R(\o)+2\eta_b(1-R,\o)\nab^b R(\o)+\f12 (\O \tr\chib)(1-R,\o) |\nab R(\o)|^2\\
&+2(\O\chibh_{ab})(1-R,\o)\nab^a R(\o) \nab^b R(\o)+4(\O \omb)(1-R,\o) |\nab R(\o)|^2\\
&-\f12 (\O^{-1} \tr\chi)(1-R,\o).
\end{split}
\end{equation*}
Notice
\begin{equation}\label{3.63}
\begin{split}
&\Delta_{R(\o)}(\tR(\o)-R(\o))\\
=&-({\Delta_{\tR(\o)}-\Delta_{R(\o)}}) \tR(\o)+\l{\Delta_{\tR(\o)} \tR(\o)-\Delta_{R(\o)}R(\o)}\r\\
=& I+II.
\end{split}
\end{equation}

To deduce the expression for $I$, recall (\ref{3.39})
\begin{equation*}
\begin{split}
\D_M f=&g^{\theta_1 \theta_1}\f{\partial^2 f}{\partial \theta_1 \partial \theta_1}+g^{\theta_2 \theta_2}\f{\partial^2 f}{\partial \theta_2 \partial \theta_2}+2 g^{\theta_1 \theta_2}\f{\partial^2 f}{\partial \theta_1 \partial \theta_2}\\
&+\f{\partial}{\partial \theta_1} (g^{\theta_1 \theta_1})\f{\partial f}{\partial \theta_1}+\f{\partial}{\partial \theta_1} (g^{\theta_1 \theta_2})\f{\partial f}{\partial \theta_2}\\
&+\f{\partial}{\partial \theta_2} (g^{\theta_2 \theta_1})\f{\partial f}{\partial \theta_1}+\f{\partial}{\partial \theta_2} (g^{\theta_2 \theta_2})\f{\partial f}{\partial \theta_2}\\
&+\f12 g^{\theta_k \theta_j}\f{\partial g_{\theta_j \theta_k}}{\partial \theta_i} g^{\theta_i \theta_l}\f{\partial f}{\partial \theta_l},
\end{split}
\end{equation*}
where $i,j,k=1,2$ and $g^{\theta_k \theta_j}$ depends on $u$. 

Based on the formula above, we decompose $I$ into
$$I=I_1+I_2+...+I_8,$$ 
where
$$I_1=-\l{g^{\theta_1 \theta_1}(1-\tR,\o)-g^{\theta_1 \theta_1}(1-R,\o)}\r \f{\partial^2}{\partial \theta_1 \partial \theta_1} \tR(\o),$$
$$I_2=-\l{g^{\theta_2 \theta_2}(1-\tR,\o)-g^{\theta_2 \theta_2}(1-R,\o)}\r\f{\partial^2}{\partial \theta_2 \partial \theta_2} \tR(\o),$$
$$I_3=-2\l{g^{\theta_1 \theta_2}(1-\tR,\o)-g^{\theta_1 \theta_2}(1-R,\o)}\r\f{\partial^2}{\partial \theta_1 \partial \theta_2} \tR(\o),$$
$$I_4=-\l{\f{\partial g^{\theta_1 \theta_1}}{\partial \theta_1}(1-\tR,\o)-\f{\partial g^{\theta_1 \theta_1}}{\partial \theta_1}(1-R,\o)}\r \f{\partial}{\partial \theta_1} \tR(\o),$$
$$I_5=-\l{\f{\partial g^{\theta_1 \theta_2}}{\partial \theta_1}(1-\tR,\o)-\f{\partial g^{\theta_1 \theta_2}}{\partial \theta_1}(1-R,\o)}\r \f{\partial}{\partial \theta_2} \tR(\o),$$
$$I_6=-\l{\f{\partial g^{\theta_2 \theta_1}}{\partial \theta_2}(1-\tR,\o)-\f{\partial g^{\theta_2 \theta_1}}{\partial \theta_2}(1-R,\o)}\r \f{\partial}{\partial \theta_1} \tR(\o),$$
$$I_7=-\l{\f{\partial g^{\theta_2 \theta_2}}{\partial \theta_2}(1-\tR,\o)-\f{\partial g^{\theta_2 \theta_2}}{\partial \theta_2}(1-R,\o)}\r \f{\partial}{\partial \theta_2} \tR(\o),$$
$$I_8=-\f12\l{( g^{\theta_k \theta_j} \f{\partial g_{\theta_j \theta_k}}{\partial \theta_i} g^{\theta_i \theta_l})(1-\tR,\o)-(g^{\theta_k \theta_j} \f{\partial g_{\theta_j \theta_k}}{\partial \theta_i} g^{\theta_i \theta_l})(1-R,\o)}\r \cdot \f{\partial \tR}{\partial \theta_l}(\o). $$
For $I_1$, we have
\begin{equation*}
\begin{split}
I_1=&-\l{g^{\theta_1 \theta_1}(1-\tR,\o)-g^{\theta_1 \theta_1}(1-R,\o)}\r \f{\partial^2}{\partial \theta_1 \partial \theta_1} \tR(\o)\\
=&\int_0^1\f{\partial g^{\theta_1 \theta_1}}{\partial u}\l1-\tau \tR-(1-\tau) R\r d\tau \\
&\quad \quad \quad \times\l{\tR(\o)-R(\o)}\r\cdot \f{\partial^2}{\partial \theta_1 \partial \theta_1} \tR(\o).\\
\end{split}
\end{equation*}
For $0\leq \tau \leq 1$, we utilize the estimates (\ref{nab nab R})
$$\f{\partial^2}{\partial \theta_1 \partial \theta_1} \tR(\o)\approx o(1) \cdot \ub a,$$ 
$$\f{\partial g^{\theta_1 \theta_1}}{\partial u}\l1-\tau \tR-(1-\tau) R \r \approx \f{1}{\ub^3 a^3},$$
and get
$$I_1= \f{1}{\ub^2 a^2} \cdot ({\tR-R}) \cdot o(1).$$
In the same fashion, for $2\leq i \leq 7$ we have
$$I_i= \f{1}{\ub^2 a^2} \cdot (\tR-R) \cdot o(1).$$

The last term $I_8$ is
\begin{equation*}
\begin{split}
I_8=-\f12\l{( g^{\theta_k \theta_j} \f{\partial g_{\theta_j \theta_k}}{\partial \theta_i} g^{\theta_i \theta_l})(1-\tR,\o)-( g^{\theta_k \theta_j} \f{\partial g_{\theta_j \theta_k}}{\partial \theta_i} g^{\theta_i \theta_l})(1-R,\o)}\r \cdot \f{\partial \tR}{\partial \theta_l}(\o). 
\end{split}
\end{equation*}
Thanks to identity
\begin{equation*}
\begin{split}
&{f(u')g(u')h(u')-f(u)g(u)h(u)}\\
=&f(u')g(u')\l{h(u')-h(u)}\r+f(u')\l{g(u')-g(u)}\r h(u)+\l{f(u')-f(u)}\r g(u)h(u), 
\end{split}
\end{equation*}
we conclude that 
\begin{equation*}
\begin{split}
I_8
=&\f12 g^{\theta_k \theta_j}(1-\tR,\o) \cdot \f{\partial g_{\theta_j \theta_k}}{\partial \theta_i} (1-\tR,\o)\cdot \l{\tR-R}\r \cdot \f{\partial \tR}{\partial \theta_l} (\o)\\
&\quad\quad\quad\times \int_0^1\f{\partial}{\partial u} g^{\theta_i \theta_l} (1-\tau \tR-(1-\tau) R, \o)d\tau\\
&+\f12 g^{\theta_k \theta_j}(1-\tR,\o) \cdot \l{\tR-R}\r \cdot \f{\partial \tR}{\partial \theta_l} (\o)\\
&\quad\quad\quad \times \int_0^1 \f{\partial^2 g_{\theta_j \theta_k}}{\partial u \partial \theta_i} (1-\tau \tR-(1-\tau) R,\o) d\tau \cdot  g^{\theta_i \theta_l} (1-R, \o)\\
&+\f12\f{\partial g_{\theta_j \theta_k}}{\partial \theta_i} (1-R,\o) \cdot g^{\theta_i \theta_l} (1-R,\o)\cdot \l{\tR-R}\r \cdot \f{\partial \tR}{\partial \theta_l} (\o)\\
&\quad\quad\quad \times \int_0^1\f{\partial}{\partial u}g^{\theta_k \theta_j}(1-\tau \tR-(1-\tau) R,\o)d\tau\\
=& \f{1}{\ub^2 a^2} \cdot ({\tR-R}) \cdot o(1).
\end{split}
\end{equation*}

For $II$, we first recall (\ref{trchi u}) $$\nab_3 \tr\chi|_{S_{1-R,\ub}}=\f{2}{R^2}-\f{2\ub a f(\ub, \o)}{R^3}+\f{\ub\at c_3}{R^3},$$
with $|c_3|\leq b^{\f14}$. We then analyze the leading term $II_1$ in $II$
\begin{equation*}
\begin{split}
II_1=&\f12 \l{(\O^{-1}\tr\chi)(1-\tR,\o)-(\O^{-1}\tr\chi)(1-R,\o)}\r\\
=&-\f12\int_0^1 \f{\partial}{\partial u}(\O^{-1}\tr\chi)(1-\tau \tR-(1-\tau)R,\o)d\tau \cdot \l{\tR-R}\r\\
=&\int_0^1\l-\f{1}{(\tau \tR+(1-\tau)R)^2}+\f{\ub a f(\ub, \o)}{(\tau \tR+(1-\tau)R)^3}+\f{\ub\at c_3}{(\tau \tR+(1-\tau)R)^3}\r d\tau \cdot(\tR-R)\\
=&\int_0^1 \l \f{-\tau \tR-(1-\tau)R+\ub a f(\ub, \o)+\ub\at c_3}{(\tau \tR+(1-\tau)R)^3} \r d\tau \cdot (\tR-R).
\end{split}
\end{equation*}  
Using $C^0$ estimate 
$$\f38 \ub a\leq R(\o), \tR(\o) \leq \f58\ub a,$$
here we obtain 
$$-\tau \tR-(1-\tau)R+\ub a f(\ub, \o)+\ub\at c_3\leq -\tau \f{3 \ub a}{8}-(1-\tau) \f{3\ub a}{8}+\f{15\ub a}{12}=\f{7\ub a}{8}.$$
$$-\tau \tR-(1-\tau)R+\ub a f(\ub, \o)+\ub\at c_3\geq -\tau \f{5 \ub a}{8}-(1-\tau) \f{5\ub a}{8}+\f{9\ub a}{12}=\f{\ub a}{8}.$$
Therefore,
\begin{equation*}
\begin{split}
&\f12 \l{(\O^{-1}\tr\chi)(1-\tR,\o)-(\O^{-1}\tr\chi)(1-R,\o)}\r\\
=&\int_0^1 \l \f{-\tau \tR-(1-\tau)R+\ub a f(\ub, \o)+\ub\at c_3}{(\tau \tR+(1-\tau)R)^3} \r d\tau \cdot (\tR-R)\\
=&\nu(\o)\cdot (\tR-R).\\
\end{split}
\end{equation*}  
We further have
$$\nu(\o)\geq \f{\ub a}{8} \f{1}{(\f{3\ub a}{8})^3}=\f{64}{81\ub^2 a^2}. $$

Similarly, we control $II_2$ (the other terms in $II$) and conclude
\begin{equation*}
\begin{split}
|II_2|\leq &\f{1}{\ub^2 a^2} \cdot \l{\tR-R}\r \cdot \mbox{o(1)}\\
&+\f{1}{\ub^2 a^2} \cdot \f{\partial}{\partial \theta_i}\l\tR-R\r \cdot \mbox{o(1)}.
\end{split}
\end{equation*}
Back to (\ref{3.63}), we arrive at 
\begin{equation}\label{3.77}
\begin{split}
&\Delta_{R(\o)} (\tR-R)(\o)-\nu(\o) (\tR-R)(\o) \\
&+\f{1}{\ub^2 a^2} \cdot (\tR-R)(\o) \cdot o(1)\\
&+\f{1}{\ub^2 a^2} \cdot \f{\partial}{\partial \theta_i}(\tR-R)(\o) \cdot o(1)\\
=&0,
\end{split}
\end{equation}
with
$$\nu(\o)\geq \f{64}{81\ub^2 a^2}. $$
With maximal principle, we conclude that 
$$\tR(\o)=R(\o) \quad \mbox{for} \quad \o\in \mathbb{S}^2.$$

\section{Regularity of Apparent Horizon}\label{regularity of horizon}
According to the previous section, along each $\Hb_{\ub}$ we have obtained a unique solution for (\ref{LR=0}). This solution corresponds to a unique MOTS on $\underline{H}_{\ub}$. 
For different $\ub$, collecting these MOTSs together yields an apparent horizon. In this section, we study its regularity. \\ 

By (\ref{LR=0}), for different $\ub'$ and $\ub$ we have
\begin{equation}\label{3.61}
\begin{split}
0=&\D_{R(\ub', \o)} R(\ub', \o)+2\eta_b(1-R(\ub',\o), \ub',\o)\nab^b R(\ub', \o)\\
&+\f12 (\O \tr\chib)(1-R(\ub',\o), \ub',\o) |\nab R(\ub', \o)|^2\\
&+2(\O\chibh_{ab})(1-R(\ub',\o), \ub',\o)\nab^a R(\ub', \o) \nab^b R(\ub', \o)\\
&+4(\O \omb)(1-R(\ub',\o), \ub',\o) |\nab R(\ub', \o)|^2\\
&-\f12 (\O^{-1} \tr\chi)(1-R(\ub',\o), \ub',\o),
\end{split}
\end{equation}
and
\begin{equation}\label{3.62}
\begin{split}
0=&\D_{R(\ub, \o)} R(\ub, \o)+2\eta_b(1-R(\ub,\o), \ub,\o)\nab^b R(\ub, \o)\\
&+\f12 (\O \tr\chib)(1-R(\ub,\o), \ub,\o) |\nab R(\ub, \o)|^2\\
&+2(\O\chibh_{ab})(1-R(\ub,\o), \ub,\o)\nab^a R(\ub, \o) \nab^b R(\ub, \o)\\
&+4(\O \omb)(1-R(\ub,\o), \ub,\o) |\nab R(\ub, \o)|^2\\
&-\f12 (\O^{-1} \tr\chi)(1-R(\ub,\o), \ub,\o).
\end{split}
\end{equation}
Notice
\begin{equation}\label{3.63}
\begin{split}
&\Delta_{R(\ub, \o)}\f{R(\ub', \o)-R(\ub, \o)}{\ub'-\ub}\\
=&-\f{\Delta_{R(\ub',\o)}-\Delta_{R(\ub,\o)}}{\ub'-\ub} R(\ub', \o)+\f{\Delta_{R(\ub',\o)} R(\ub',\o)-\Delta_{R(\ub, \o)}R(\ub, \o)}{\ub'-\ub}\\
=& I+II.
\end{split}
\end{equation}
The expression for $II$ follows from (\ref{3.61}) and (\ref{3.62}). 
To deduce the expression for $I$, recall (\ref{3.39})
\begin{equation*}
\begin{split}
\D_M f=&g^{\theta_1 \theta_1}\f{\partial^2 f}{\partial \theta_1 \partial \theta_1}+g^{\theta_2 \theta_2}\f{\partial^2 f}{\partial \theta_2 \partial \theta_2}+2 g^{\theta_1 \theta_2}\f{\partial^2 f}{\partial \theta_1 \partial \theta_2}\\
&+\f{\partial}{\partial \theta_1} (g^{\theta_1 \theta_1})\f{\partial f}{\partial \theta_1}+\f{\partial}{\partial \theta_1} (g^{\theta_1 \theta_2})\f{\partial f}{\partial \theta_2}\\
&+\f{\partial}{\partial \theta_2} (g^{\theta_2 \theta_1})\f{\partial f}{\partial \theta_1}+\f{\partial}{\partial \theta_2} (g^{\theta_2 \theta_2})\f{\partial f}{\partial \theta_2}\\
&+\f12 g^{\theta_k \theta_j}\f{\partial g_{\theta_j \theta_k}}{\partial \theta_i} g^{\theta_i \theta_l}\f{\partial f}{\partial \theta_l},
\end{split}
\end{equation*}
where $i,j,k=1,2$ and $g^{\theta_k \theta_j}$ depends on $u$. 

Based on the formula above, we decompose $I$ into
$$I=I_1+I_2+...+I_8,$$ 
where
$$I_1=-\f{g^{\theta_1 \theta_1}(1-R(\ub',\o), \ub', \o)-g^{\theta_1 \theta_1}(1-R(\ub,\o), \ub, \o)}{\ub'-\ub} \f{\partial^2}{\partial \theta_1 \partial \theta_1} R(\ub',\o),$$
$$I_2=-\f{g^{\theta_2 \theta_2}(1-R(\ub',\o),\ub', \o)-g^{\theta_2 \theta_2}(1-R(\ub,\o), \ub, \o)}{\ub'-\ub} \f{\partial^2}{\partial \theta_2 \partial \theta_2} R(\ub',\o),$$
$$I_3=-2\cdot\f{g^{\theta_1 \theta_2}(1-R(\ub',\o), \ub', \o)-g^{\theta_1 \theta_2}(1-R(\ub,\o), \ub, \o)}{\ub'-\ub} \f{\partial^2}{\partial \theta_1 \partial \theta_2} R(\ub',\o),$$
$$I_4=-\f{\f{\partial g^{\theta_1 \theta_1}}{\partial \theta_1}(1-R(\ub',\o), \ub', \o)-\f{\partial g^{\theta_1 \theta_1}}{\partial \theta_1}(1-R(\ub,\o), \ub, \o)}{\ub'-\ub} \f{\partial}{\partial \theta_1} R(\ub',\o),$$
$$I_5=-\f{\f{\partial g^{\theta_1 \theta_2}}{\partial \theta_1}(1-R(\ub',\o), \ub', \o)-\f{\partial g^{\theta_1 \theta_2}}{\partial \theta_1}(1-R(\ub,\o), \ub, \o)}{\ub'-\ub} \f{\partial}{\partial \theta_2} R(\ub',\o),$$
$$I_6=-\f{\f{\partial g^{\theta_2 \theta_1}}{\partial \theta_2}(1-R(\ub',\o), \ub', \o)-\f{\partial g^{\theta_2 \theta_1}}{\partial \theta_2}(1-R(\ub,\o), \ub, \o)}{\ub'-\ub} \f{\partial}{\partial \theta_1} R(\ub',\o),$$
$$I_7=-\f{\f{\partial g^{\theta_2 \theta_2}}{\partial \theta_2}(1-R(\ub',\o), \ub', \o)-\f{\partial g^{\theta_2 \theta_2}}{\partial \theta_2}(1-R(\ub,\o), \ub, \o)}{\ub'-\ub} \f{\partial}{\partial \theta_2} R(\ub',\o),$$
$$I_8=-\f12\f{\l g^{\theta_k \theta_j} \f{\partial g_{\theta_j \theta_k}}{\partial \theta_i} g^{\theta_i \theta_l}\r(1-R(\ub',\o), \ub', \o)-\l g^{\theta_k \theta_j} \f{\partial g_{\theta_j \theta_k}}{\partial \theta_i} g^{\theta_i \theta_l}\r(1-R(\ub,\o), \ub, \o)}{\ub'-\ub} \cdot \f{\partial R}{\partial \theta_l}(\ub',\o). $$

For $I_1$, we have
\begin{equation*}
\begin{split}
I_1=&-\f{g^{\theta_1 \theta_1}(1-R(\ub',\o), \ub', \o)-g^{\theta_1 \theta_1}(1-R(\ub,\o), \ub, \o)}{\ub'-\ub} \f{\partial^2}{\partial \theta_1 \partial \theta_1} R(\ub',\o)\\
=&\int_0^1\f{\partial g^{\theta_1 \theta_1}}{\partial u}\l1-\tau R(\ub', \o)-(1-\tau) R(\ub, \o), \tau\ub'+(1-\tau)\ub, \o \r d\tau \\
&\quad \quad \quad \times\f{R(\ub', \o)-R(\ub, \o)}{\ub'-\ub}\cdot \f{\partial^2}{\partial \theta_1 \partial \theta_1} R(\ub', \o)\\
&-\int_0^1 \f{\partial g^{\theta_1 \theta_1}}{\partial \ub}\l 1-\tau R(\ub', \o)-(1-\tau)R(\ub,\o), \tau\ub'+(1-\tau)\ub, \o \r d\tau\cdot\f{\ub'-\ub}{\ub'-\ub}\cdot \f{\partial^2}{\partial \theta_1 \partial \theta_1}R(\ub',\o). 
\end{split}
\end{equation*}
For $0\leq \tau \leq 1$ and when $\ub'$ is close to $\ub$, we utilize the estimates (\ref{nab nab R})
$$\f{\partial^2}{\partial \theta_1 \partial \theta_1} R(\ub',\o)\approx o(1) \cdot \ub a, $$
$$\f{\partial g^{\theta_1 \theta_1}}{\partial u}\l1-\tau R(\ub', \o)-(1-\tau) R(\ub, \o), \tau\ub'+(1-\tau)\ub, \o \r \approx \f{1}{\ub^3 a^3},$$
$$\f{\partial g^{\theta_1 \theta_1}}{\partial \ub}\l 1-\tau R(\ub', \o)-(1-\tau)R(\ub,\o), \tau\ub'+(1-\tau)\ub, \o \r d\tau \approx \f{1}{\ub^3 a^2}$$
and get
$$I_1= \f{1}{\ub^2 a^2} \cdot \f{R(\ub', \o)-R(\ub, \o)}{\ub'-\ub} \cdot o(1)+\f{a}{\ub^2 a^2}\cdot o(1).$$
In the same fashion, for $2\leq i \leq 7$ we have
$$I_i= \f{1}{\ub^2 a^2} \cdot \f{R(\ub', \o)-R(\ub, \o)}{\ub'-\ub} \cdot o(1)+\f{a}{\ub^2 a^2}\cdot o(1).$$
The last term $I_8$ is
\begin{equation*}
\begin{split}
I_8=-\f12\f{\l g^{\theta_k \theta_j} \f{\partial g_{\theta_j \theta_k}}{\partial \theta_i} g^{\theta_i \theta_l}\r(1-R(\ub',\o),\ub', \o)-\l g^{\theta_k \theta_j} \f{\partial g_{\theta_j \theta_k}}{\partial \theta_i} g^{\theta_i \theta_l}\r(1-R(\ub,\o), \ub, \o)}{\ub'-\ub} \cdot \f{\partial R}{\partial \theta_l}(\ub',\o). 
\end{split}
\end{equation*}
Thanks to identity
\begin{equation*}
\begin{split}
&\f{f(\ub')g(\ub')h(\ub')-f(\ub)g(\ub)h(\ub)}{\ub'-\ub}\\
=&f(\ub')g(\ub')\f{h(\ub')-h(\ub)}{\ub'-\ub}+f(\ub')\f{g(\ub')-g(\ub)}{\ub'-\ub}h(\ub)+\f{f(\ub')-f(\ub)}{\ub'-\ub}g(\ub)h(\ub), 
\end{split}
\end{equation*}
we conclude that 
\begin{equation*}
\begin{split}
I_8
=&\f12 g^{\theta_k \theta_j}(1-R(\ub',\o), \ub', \o) \cdot \f{\partial g_{\theta_j \theta_k}}{\partial \theta_i} (1-R(\ub',\o), \ub', \o)\cdot \f{R(\ub',\o)-R(\ub, \o)}{\ub'-\ub} \cdot \f{\partial R}{\partial \theta_l} (\ub',\o)\\
&\quad\quad\quad\times \int_0^1\f{\partial}{\partial u} g^{\theta_i \theta_l} (1-\tau R(\ub', \o)-(1-\tau) R(\ub, \o), \tau \ub'+(1-\tau)\ub, \o)d\tau\\
&-\f12 g^{\theta_k \theta_j}(1-R(\ub',\o), \ub', \o) \cdot \f{\partial g_{\theta_j \theta_k}}{\partial \theta_i} (1-R(\ub',\o), \ub', \o)\cdot \f{\ub'-\ub}{\ub'-\ub} \cdot \f{\partial R}{\partial \theta_l} (\ub',\o)\\
&\quad\quad\quad\times \int_0^1\f{\partial}{\partial \ub} g^{\theta_i \theta_l} (1-\tau R(\ub', \o)-(1-\tau) R(\ub, \o), \tau \ub'+(1-\tau)\ub, \o)d\tau\\
&+\f12 g^{\theta_k \theta_j}(1-R(\ub',\o), \ub', \o) \cdot \f{R(\ub',\o)-R(\ub, \o)}{\ub'-\ub} \cdot \f{\partial R}{\partial \theta_l} (\ub',\o)\\
&\quad\quad\quad \times \int_0^1 \f{\partial^2 g_{\theta_j \theta_k}}{\partial u \partial \theta_i} (1-\tau R(\ub', \o)-(1-\tau) R(\ub, \o), \tau \ub'+(1-\tau)\ub, \o) d\tau\\
&\quad\quad\quad\quad\quad\quad \times  g^{\theta_i \theta_l} (1-R(\ub, \o), \ub, \o)\\
&-\f12 g^{\theta_k \theta_j}(1-R(\ub',\o), \ub', \o) \cdot \f{\ub'-\ub}{\ub'-\ub} \cdot \f{\partial R}{\partial \theta_l} (\ub',\o)\\
&\quad\quad\quad \times \int_0^1 \f{\partial^2 g_{\theta_j \theta_k}}{\partial \ub \partial \theta_i} (1-\tau R(\ub', \o)-(1-\tau) R(\ub, \o), \tau \ub'+(1-\tau)\ub, \o) d\tau\\
&\quad\quad\quad\quad\quad\quad \times  g^{\theta_i \theta_l} (1-R(\ub, \o), \ub, \o)\\
&+\f12\f{\partial g_{\theta_j \theta_k}}{\partial \theta_i} (1-R(\ub,\o), \ub, \o) \cdot g^{\theta_i \theta_l} (1-R(\ub,\o), \ub, \o)\cdot \f{R(\ub',\o)-R(\ub, \o)}{\ub'-\ub} \cdot \f{\partial R}{\partial \theta_l} (\ub',\o)\\
&\quad\quad\quad \times \int_0^1\f{\partial}{\partial u}g^{\theta_k \theta_j}(1-\tau R(\ub', \o)-(1-\tau) R(\ub, \o), \tau \ub'+(1-\tau)\ub, \o)d\tau\\
&-\f12\f{\partial g_{\theta_j \theta_k}}{\partial \theta_i} (1-R(\ub,\o), \ub, \o) \cdot g^{\theta_i \theta_l} (1-R(\ub,\o), \ub, \o)\cdot \f{\ub'-\ub}{\ub'-\ub} \cdot \f{\partial R}{\partial \theta_l} (\ub',\o)\\
&\quad\quad\quad \times \int_0^1\f{\partial}{\partial \ub}g^{\theta_k \theta_j}(1-\tau R(\ub', \o)-(1-\tau) R(\ub, \o), \tau \ub'+(1-\tau)\ub, \o)d\tau\\
=& \f{1}{\ub^2 a^2} \cdot \f{R(\ub', \o)-R(\ub, \o)}{\ub'-\ub} \cdot o(1)+\f{a}{\ub^2 a^2}\cdot o(1).
\end{split}
\end{equation*}
For $II$, we first analyze the leading term
\begin{equation*}
\begin{split}
&\f12 \f{(\O^{-1}\tr\chi)(1-R(\ub',\o), \ub', \o)-(\O^{-1}\tr\chi)(1-R(\ub,\o), \ub, \o)}{\ub'-\ub}\\
=&\f12 \int_0^1 \f{\partial}{\partial \ub}(\O^{-1}\tr\chi)(1-\tau R(\ub',\o)-(1-\tau)R(\ub,\o), \tau \ub'+(1-\tau)\ub, \o)d\tau \cdot \f{\ub'-\ub}{\ub'-\ub}\\
&-\f12\int_0^1 \f{\partial}{\partial u}(\O^{-1}\tr\chi)(1-\tau R(\ub',\o)-(1-\tau)R(\ub,\o), \tau \ub'+(1-\tau)\ub, \o)d\tau\\
&\quad\quad\quad\times \f{R(\ub',\o)-R(\ub, \o)}{\ub'-\ub}\\
=&-\tilde{\nu}(\ub, \o; \ub')\f{a}{\ub^2 a^2}+\f{\nu(\ub, \o; \ub')}{\ub^2 a^2}\cdot \f{R(\ub',\o)-R(\ub,\o)}{\ub'-\ub}.
\end{split}
\end{equation*}  
Here $\tilde{\nu}(\ub, \o; \ub')$ is defined through
\begin{equation}\label{tilde nu}
\begin{split}
&-\tilde{\nu}(\ub, \o; \ub')\f{a}{\ub^2 a^2}\\
=&\f12 \int_0^1 \f{\partial}{\partial \ub}(\O^{-1}\tr\chi)(1-\tau R(\ub',\o)-(1-\tau)R(\ub,\o), \tau \ub'+(1-\tau)\ub, \o)d\tau \cdot \f{\ub'-\ub}{\ub'-\ub}\\
=&-\f12\int_0^1 [1+o(1)]\cdot |\chih|^2 (1-\tau R(\ub',\o)-(1-\tau)R(\ub,\o), \tau \ub'+(1-\tau)\ub, \o) d\tau.\\ 
\end{split}
\end{equation}  
In the same manner, we define $\nu(\ub, \o; \ub')$ through
\begin{equation*}
\begin{split}
&\f{\nu(\ub, \o; \ub')}{\ub^2 a^2}\cdot \f{R(\ub',\o)-R(\ub,\o)}{\ub'-\ub}\\
=&-\f12\int_0^1 \f{\partial}{\partial u}(\O^{-1}\tr\chi)(1-\tau R(\ub',\o)-(1-\tau)R(\ub,\o), \tau \ub'+(1-\tau)\ub, \o)d\tau\\
&\quad\quad\quad\times\f{R(\ub',\o)-R(\ub, \o)}{\ub'-\ub}.
\end{split}
\end{equation*}  
Using (\ref{trchi u}), we have
\begin{equation}\label{new new nu}
\begin{split}
&\f{\nu(\ub, \o; \ub')}{\ub^2 a^2}\cdot \f{R(\ub',\o)-R(\ub,\o)}{\ub'-\ub}\\
=&\int_0^1 \l \f{-\tau R(\ub', \o)-(1-\tau)R(\ub, \o)+[\tau\ub'+(1-\tau)\ub] a f(\ub, \o)+[\tau\ub'+(1-\tau)\ub]\at c_3}{\l\tau R(\ub',\o)+(1-\tau)R(\ub,\o)\r^3} \r d\tau\\
&\quad\quad\quad \times \f{R(\ub',\o)-R(\ub,\o)}{\ub'-\ub},
\end{split}
\end{equation}  
with $|c_3|\leq b^{\f14}$. \\
Together with $C^0$ estimate, when $\ub'$ is close to $\ub$, we conclude
$$\nu(\ub, \o; \ub')\geq \f{1}{8}\f{1}{(\f38)^3}=\f{64}{81},$$
$$\nu(\ub, \o; \ub')\leq \f{17}{24}\f{1}{(\f58)^3}=\f{1088}{375}.$$
We then move to $II_2$ (the other terms in $II$). In a similar fashion, we conclude
\begin{equation*}
\begin{split}
|II_2|\leq &\f{a}{\ub^2 a^2}\cdot o(1)+\f{1}{\ub^2 a^2} \cdot \f{R(\ub', \o)-R(\ub, \o)}{\ub'-\ub} \cdot o(1)\\
&+\f{1}{\ub^2 a^2} \cdot \f{\partial}{\partial \theta_i}\f{R(\ub', \o)-R(\ub, \o)}{\ub'-\ub} \cdot o(1).
\end{split}
\end{equation*}
Back to (\ref{3.63}), we arrive at 
\begin{equation}\label{3.77}
\begin{split}
&\Delta_{R(\ub, \o)} \l \f{R(\ub',\o)-R(\ub,\o)}{\ub'-\ub}\r-\f{\nu(\ub, \o; \ub')}{\ub^2 a^2} \l \f{R(\ub',\o)-R(\ub,\o)}{\ub'-\ub}\r \\
&+\tilde{\nu}(\ub, \o; \ub')\f{a f(\o)}{\ub^2 a^2}+\f{1}{\ub^2 a^2} \cdot \f{R(\ub', \o)-R(\ub, \o)}{\ub'-\ub} \cdot o(1)\\
&+\f{1}{\ub^2 a^2} \cdot \f{\partial}{\partial \theta_i}\f{R(\ub', \o)-R(\ub, \o)}{\ub'-\ub} \cdot o(1)\\
=&\f{a}{\ub^2 a^2}\cdot o(1).
\end{split}
\end{equation}
We then construct an elliptic equation for $h(\ub, \o; \ub')$ satisfying
\begin{equation}\label{new h}
\D_{R(\ub,\o)}h(\ub, \o; \ub')-\f{\nu(\ub, \o; \ub')}{\ub^2 a^2}h(\ub, \o; \ub')+\tilde{\nu}(\ub, \o; \ub')\f{a{\color{black}f(\o)}}{\ub^2 a^2}=0.
\end{equation}
For this equation, we have a unique solution $h(\ub, \o; \ub')$ and it is smooth. 

With the help of $h(\ub,\o;\ub')$, (\ref{3.77}) could be rewritten as
\begin{equation}\label{eqn dq}
\begin{split}
&\Delta_{R(\ub, \o)} \l \f{R(\ub',\o)-R(\ub,\o)}{\ub'-\ub}-h(\ub, \o; \ub')\r-\f{\nu(\ub, \o; \ub')}{\ub^2 a^2} \l \f{R(\ub',\o)-R(\ub,\o)}{\ub'-\ub}-h(\ub, \o; \ub')\r \\
&+\f{1}{\ub^2 a^2} \cdot \l\f{R(\ub', \o)-R(\ub, \o)}{\ub'-\ub}-h(\ub, \o; \ub')\r \cdot o(1)\\
&+\f{1}{\ub^2 a^2} \cdot \l\f{\partial}{\partial \theta_i}\f{R(\ub', \o)-R(\ub, \o)}{\ub'-\ub}-h(\ub, \o; \ub')\r \cdot o(1)\\
=&\f{a}{\ub^2 a^2}\cdot o(1).
\end{split}
\end{equation}
Multiplying $\f{R(\ub',\o)-R(\ub,\o)}{\ub'-\ub}-h(\ub, \o; \ub')$ with (\ref{eqn dq}) and integrating by parts on $M_{\ub}$ we arrive at
\begin{equation*}
\begin{split}
&\int_{M_{\ub}} |\nab' \l \f{R(\ub',\o)-R(\ub,\o)}{\ub'-\ub}-h(\ub, \o; \ub')\r |^2\\
&+\int_{M_{\ub}} \f{\nu(\ub, \o; \ub')}{\ub^2 a^2} \l \f{R(\ub',\o)-R(\ub,\o)}{\ub'-\ub}-h(\ub, \o; \ub')\r ^2\\
=& \int_{M_{\ub}} \f{1}{\ub^2 a^2} \l \f{R(\ub',\o)-R(\ub,\o)}{\ub'-\ub}-h(\ub, \o; \ub')\r ^2 \cdot o(1)\\
&+\int_{M_{\ub}} \f{1}{\ub^2 a^2} \l \f{R(\ub',\o)-R(\ub,\o)}{\ub'-\ub}-h(\ub, \o; \ub')\r \cdot a  \cdot o(1)\\
&+\int_{M_{\ub}} \f{1}{\ub^2 a^2} \l \f{R(\ub',\o)-R(\ub,\o)}{\ub'-\ub}-h(\ub, \o; \ub')\r \\
&\quad\quad\quad\times \f{\partial}{\partial \theta_i} \l \f{R(\ub', \o)-R(\ub, \o)}{\ub'-\ub}-h(\ub, \o; \ub') \r \cdot o(1).
\end{split}
\end{equation*}
This implies
\begin{equation}\label{3.81} 
\begin{split}
&\int_{M_{\ub}} |\nab' \l \f{R(\ub',\o)-R(\ub,\o)}{\ub'-\ub}-h(\ub, \o; \ub')\r |^2\\
&+\int_{M_{\ub}} \f{2}{\ub^2 a^2} \l \f{R(\ub',\o)-R(\ub,\o)}{\ub'-\ub}-h(\ub, \o; \ub')\r ^2\\
\leq &\int_{M_{\ub}} \f{a^2}{\ub^2 a^2}\cdot o(1) \leq o(1) \cdot a^2. 
\end{split}
\end{equation}
Denote $\nu(\ub, \o)$ and $\tilde{\nu}(\ub, \o)$ to be the limits of $\nu(\ub, \o; \ub')$ and $\tilde{\nu}(\ub, \o; \ub')$ as $\ub'$ approaches $\ub$:
$$\nu(\ub, \o):=\lim_{\ub'\rightarrow \ub}\nu(\ub, \o; \ub'), \quad \quad \tilde{\nu}(\ub, \o):=\lim_{\ub'\rightarrow \ub}\tilde{\nu}(\ub, \o; \ub').$$
Let $h(\ub, \o)$ be the unique solution to 
$$\D_{R(\ub,\o)}h(\ub, \o)-\f{\nu(\ub, \o)}{\ub^2 a^2}h(\ub, \o)+\tilde{\nu}(\ub, \o)\f{a f(\o)}{\ub^2 a^2}=0.$$
Back to (\ref{3.81}), since $\ub'$ is an arbitrary number close to $\ub$ {\color{black}and 
$h(\ub, \o; \ub')$ being a smooth function with respect to all variables, we have that $\f{R(\ub',\o)-R(\ub,\o)}{\ub'-\ub}$ is uniformly bounded in $L^2(M_{\ub})$. According to the standard result for difference quotient \footnote{{\color{black}See Theorem 3 of Section 5.8.2 in \cite{Evans}.}}, we have
$$\f{\partial R}{\partial \ub} (\ub, \o) \,\, \mbox{exists},\,\, \mbox{and by} \,\, (\ref{3.81}) \,\, \mbox{it holds} \,\, \f{\partial R}{\partial \ub} (\ub, \o)-h(\ub, \o) \in L^2(M_{\ub}).$$}
Thus,
$$R(\ub, \o)\in H^1(\mathbb{R}^3).$$
Employing (\ref{3.81}) for one more time, we conclude
$$\f{\partial R}{\partial \ub} (\ub, \o)-h(\ub, \o) \in H^1(M_{\ub}).$$
By Sobolev's embedding, this implies
$$\f{\partial R}{\partial \ub} (\ub, \o)-h(\ub, \o) \in L^6(M_{\ub}).$$
Hence,
$$R(\ub, \o)\in W^{1,6}(\mathbb{R}^3).$$
Apply Sobolev's embedding again
\begin{equation}\label{3.49new}
W^{k,p}(\mathbb{R}^n)\subset C^{r,\alpha}(\mathbb{R}^n) \quad \mbox{for} \quad \f{k-r-\alpha}{n}=\f{1}{p}, \quad \mbox{where} \quad \alpha\in (0,1).
\end{equation}
For $k=1, p=6, n=3, r=0$, we obtain $\alpha=1/2$. Therefore, 
$$R(\ub, \o)\in C^{0,\f12}(\mathbb{R}^3).$$
Utilize maximal principle for (\ref{eqn dq}), we have 
\begin{equation*}
\|\f{R(\ub',\o)-R(\ub,\o)}{\ub'-\ub}-h(\ub, \o)\|_{L^{\infty}(M_{\ub})}\leq a\cdot o(1). 
\end{equation*}
By elliptic estimates for (\ref{eqn dq}), we then conclude for $p>1$
\begin{equation*}
\begin{split}
&\|\f{R(\ub',\o)-R(\ub, \o)}{\ub'-\ub}-h(\ub, \o)\|_{W^{2,p}(M_{\ub})}\\
\lesssim & \f{1}{\ub^{2}a^{2}}\|\f{R(\ub',\o)-R(\ub, \o)}{\ub'-\ub}-h(\ub, \o)\|_{L^p(M_{\ub})}\\
&+\f{1}{\ub^2 a^2} \|o(1)\cdot a\|_{L^{p}(M_{\ub})}.
\end{split}
\end{equation*}
With Sobolev's embedding (\ref{3.49new}), for $k=2, p=3/2, n=2, r=0$, we get $\f{2-\alpha}{2}=\f23$. This implies $\alpha=2/3$. Hence, 
$$\f{R(\ub',\o)-R(\ub, \o)}{\ub'-\ub} \in C^{0,\f23}(M_{\ub}).$$
With $C^{2,\alpha}$ estimates for (\ref{eqn dq}), it follows
\begin{equation}\label{3.96new}
\begin{split}
&\|\f{R(\ub',\o)-R(\ub, \o)}{\ub'-\ub}-h(\ub, \o)\|_{C^{2,\alpha}(M_{\ub})}\\
\leq & \f{1}{\ub^{2+\alpha}a^{2+\alpha}}\|\f{R(\ub',\o)-R(\ub, \o)}{\ub'-\ub}-h(\ub, \o)\|_{C^{0}(M_{\ub})}\\
&+\f{1}{\ub^2 a^2} \|o(1)\cdot a\|_{C^{\alpha}(M_{\ub})}\\
\leq& \f{1}{\ub^{2+\alpha}a^{2+\alpha}}\cdot a \cdot o(1). 
\end{split}
\end{equation}
Define 
\begin{equation*}
\D_{\tau_n} R (\ub, \o):=\f{R(\ub+\tau_n, \o)-R(\ub, \o)}{\tau_n}. 
\end{equation*}
By (\ref{3.96new}), for any $\tau_n \rightarrow 0$, $\D_{\tau_n}R(\ub, \o)$ is bounded in $C^{2,\alpha}(M_{\ub})$. We also have the fact that for $\alpha'<\alpha$:
$$C^{2, \alpha} \hookrightarrow C^{2, \alpha'} \, \mbox{is a compact embedding}.$$
Consequently, there exists a subsequence $\tau'_n\rightarrow 0$ and $X(\ub, \o)\in C^{2,\alpha'}(M_{\ub})$, such that 
$$\D_{\tau'_n} R(\ub, \o) \rightarrow X(\ub, \o) \quad \mbox{in} \quad C^{2,\alpha'}(M_{\ub}).$$
From the derivation of (\ref{3.77}), we see that there exist elliptic operators 
$$L^{\tau_n}_{\theta_1, \theta_2}=\D_{R(\ub,\o)}+o(1)\cdot\f{1}{\ub^2 a^2}\cdot \f{\partial}{\partial \theta_i}-[1+o(1)]\cdot \f{\nu(\ub, \o; \ub+\tau_n)}{\ub^2 a^2},$$ 
$$L_{\theta_1, \theta_2}=\D_{R(\ub,\o)}+o(1)\cdot\f{1}{\ub^2 a^2}\cdot \f{\partial}{\partial \theta_i}-[1+o(1)]\cdot \f{\nu(\ub, \o)}{\ub^2 a^2},$$ 
and $l\in C^{\infty}(\ub, \o)$. With them, we could rewrite (\ref{3.77}) as
\begin{equation*}
L^{\tau_n}_{\theta_1, \theta_2} (\D_{\tau_n} R)=\D_{\tau_n}l.
\end{equation*}
Along $\tau'_n \rightarrow 0$, we have
\begin{equation*}
L_{\theta_1, \theta_2} (\lim_{n\rightarrow+\infty} \D_{\tau'_n} R(\ub,\o))=L_{\theta_1, \theta_2} X(\ub, \o)=\f{\partial l}{\partial \ub}. 
\end{equation*}
For any $\tau''_n$ (subsequence of $\tau_n$) $\rightarrow 0$, if $\lim_{n\rightarrow +\infty}\D_{\tau_n''}R(\ub, \o)$ exists, then as $n\rightarrow+\infty$ we have
\begin{equation*}
L_{\theta_1, \theta_2} \l \lim_{n\rightarrow+\infty} \D_{\tau''_n} R(\ub,\o)\r=\f{\partial l}{\partial \ub}. 
\end{equation*}
From the invertibility of this operator $L_{\theta_1, \theta_2}$, it follows
$$\lim_{n\rightarrow+\infty} \D_{\tau''_n} R(\ub,\o)=X(\ub, \o).$$
Using the following conclusion in analysis:
\begin{proposition}
Suppose $\{Y_n\}^{\infty}_{n=1}$ is a bounded sequence. If all of its convergent subsequences have the same limit, then the sequence is convergent. 
\end{proposition}
We hence conclude that $\D_{\tau_n}R(\ub, \o)$ converges and
$$\lim_{n\rightarrow+\infty} \D_{\tau_n} R(\ub,\o)=X(\ub, \o) \quad \mbox{for all} \quad \tau_n\rightarrow 0.$$
As a consequence, {\color{black} not only $\partial R/ \partial \ub$ is in $C^{0,\f12}(M_{\ub})$,  but also it holds} $\partial R/ \partial \ub \in C^{2,\alpha}(M_{\ub})$. \\

Taking $\partial/\partial \ub$ on both sides of (\ref{LR=0}): $L(R(\ub, \o))=0$ , we arrive at 
\begin{equation}\label{horizon}
\begin{split}
&\D \f{\partial}{\partial \ub}R(\ub, \o)+[\f{\partial}{\partial \ub}, \D]R(\ub, \o)+2\eta_b\f{\partial}{\partial \ub}\nab^b R(\ub, \o)+2\f{\partial \eta_b}{\partial \ub}\nab^b R(\ub, \o) \\
&+\f12 \O \tr\chib \f{\partial}{\partial \ub}|\nab R(\ub, \o)|^2+\f12  \f{\partial (\O \tr\chib)}{\partial \ub}|\nab R(\ub, \o)|^2\\
&+2\O \chibh_{ab} \f{\partial}{\partial \ub}\l\nab^a R(\ub, \o) \nab^b R(\ub, \o)\r+2 \f{\partial (\O \chibh_{ab})}{\partial \ub}\nab^a R(\ub, \o) \nab^b R(\ub, \o)\\
&+4\O \omb \f{\partial}{\partial \ub}|\nab R(\ub, \o)|^2+4 \f{\partial (\O \omb)}{\partial \ub}|\nab R(\ub, \o)|^2\\
&-\f{\partial}{\partial \ub}(\f{\O^{-1}}{2} \tr\chi)\\
=&0.
\end{split}
\end{equation}
Here to get $[\f{\partial}{\partial \ub}, \D]R(\ub,\o)$, we appeal to the fact 
$$\f{\partial }{\partial \ub}=\O e_4-b^A\f{\partial}{\partial \theta^A}$$
and
\begin{equation*}
\begin{split}
[e_4, \D]R
=&-\tr\chi \D R-2\chih\cdot\nab^2 R+\beta\cdot\nab R+\f12(\eta+\etb)\cdot \nab_4 \nab R-\etb\cdot\chih\cdot \nab R\\
&+\f12 \tr\chi\etb\cdot\nab R+\div\l\f12(\eta+\etb)D_4 R\r-\div \chi\cdot \nab R.
\end{split}
\end{equation*}
By standard elliptic estimates, {\color{black}(with sufficiently smooth initial data in \cite{AL})} we have
$$\f{\partial R}{ \partial \ub} \in C^{\infty}(M_{\ub}).$$ 

Our next goal is to study 
$$\f{\f{\partial R}{\partial \ub}(\ub', \o)-\f{\partial R}{\partial \ub}(\ub, \o)}{\ub'-\ub}.$$
Through
\begin{equation*}
\begin{split}
&\D_{R(\ub,\o)} \f{\f{\partial R}{\partial \ub}(\ub', \o)-\f{\partial R}{\partial \ub}(\ub, \o)}{\ub'-\ub}\\
=& \f{\D_{R(\ub',\o)}\f{\partial R}{\partial \ub}(\ub', \o)-\D_{R(\ub,\o)}\f{\partial R}{\partial \ub}(\ub, \o)}{\ub'-\ub}-\f{\D_{R(\ub',\o)}-\D_{R(\ub,\o)}}{\ub'-\ub} \f{\partial R}{\partial \ub} (\ub', \o),
\end{split}
\end{equation*}
we derive an elliptic equation for $$\f{\f{\partial R}{\partial \ub}(\ub', \o)-\f{\partial R}{\partial \ub}(\ub, \o)}{\ub'-\ub}.$$
With the same treatment as for 
$\f{R(\ub', \o)-R(\ub, \o)}{\ub'-\ub},$
we conclude that 
$$\f{\partial^2 R}{\partial \ub^2} \quad \mbox{is well-defined, and} \quad \f{\partial^2 R}{\partial \ub^2}\in C^{\infty}(M_{\ub}).$$
Iterating this procedure, {\color{black}(with sufficiently smooth initial data in \cite{AL})}we then improve the regularity for $R(\ub, \o)$ to
\begin{equation}
\f{\partial^{k} R}{\partial \ub^k}\in C^{\infty}(M) \quad \mbox{for any} \quad k\in \mathbb{Z}^+. 
\end{equation}

\section{Dynamical Horizon}\label{section DH}
In this section, we show that {\color{black} under more restrictive initial condition} the apparent horizon constructed is spacelike. \\

Our apparent horizon is a three dimensional hypersurface: $u=1-R(\ub, \theta_1, \theta_2)$.
With the derived estimate $$|\f{\partial R(\ub, \theta_1, \theta_2)}{\partial \ub}-h(\theta_1, \theta_2)|\leq a\cdot o(1),$$
the components of the induced metric read
\begin{equation*}
g'_{\theta_i \theta_j}=g_{\theta_i \theta_j}+\f{\partial (1-R)}{\partial \theta_i}\cdot \f{\partial(1-R)}{\partial \theta_j} \cdot g(\f{\partial}{\partial u}, \f{\partial}{\partial u})=g_{\theta_i \theta_j},
\end{equation*}
\begin{equation*}
g'_{\ub\, \ub}=g_{\ub\, \ub}+2\f{\partial (1-R)}{\partial \ub}\cdot \f{\partial \ub}{\partial \ub} \cdot g(\f{\partial}{\partial u}, \f{\partial}{\partial \ub})=4\f{\partial R}{\partial \ub}=4h(\ub, \theta_1, \theta_2)\cdot[1+o(1)],
\end{equation*}
\begin{equation*}
g'_{\theta_i \ub}=g_{\theta_i \ub}+\f{\partial (1-R)}{\partial \theta_i}\cdot \f{\partial \ub}{\partial \ub}\cdot g(\f{\partial}{\partial u}, \f{\partial}{\partial \ub})=2\f{\partial R}{\partial \theta_i}.
\end{equation*}

Here we will consider a special open set of initial data. With these data, we will show that $h(\ub, \theta_1, \theta_2){\color{black}=[\f12+o(1)]a}$, which will imply that the apparent horizon constructed is spacelike. 

\subsection{Construction of Special Initial Data}\label{Construction of Special Initial Data} {\color{black}  In previous sections, to ensure the existence and smoothness of apparent horizon, we need initial data along $H_0^{[0,\d]}$ satisfying
\begin{equation*}
\int_0^{\ub}|\chih|^2(0,\ub',\o)d\ub'=f(\ub,\o)\ub a, \quad \mbox{with} \quad \f{20}{21}\leq f(\ub, \o)\leq \f{22}{21}.
\end{equation*}

To prove that the apparent horizon is spacelike, we use more restrictive initial condition:
\begin{enumerate}
\item we require initial data along $H_0^{[0,\d]}$ satisfying
$$\int_0^{\ub}|\chih|^2(0,\ub',\o)d\ub'=f(\ub,\o)\ub a, \,\, \mbox{with} \,\, 1-
\f{1}{c}\leq f(\ub, \o)\leq 1+\f{1}{c}, \,\, \mbox{where} \,\, 1\ll c \ll b \leq \at.$$
{\color{black}\noindent Note: by Remark \ref{Remark B.2} in Appendix \ref{Appendix B}, these initial data could be obtained by rescaling initial data satisfying
\begin{equation}\label{special initial data 1}
\int_0^{\d}|\chih|^2(0,\ub',\o)d\ub'=f(\d,\o)\d a, \,\, \mbox{with} \,\, 1-
\f{1}{c}\leq f(\d, \o)\leq 1+\f{1}{c}, \,\, \mbox{where} \,\, 1\ll c \ll b \leq \at.
\end{equation}
}

\item Moreover, we require that $|\chih_0|^2(\ub, \o)$ is {\color{black}close to} a constant out of {\color{black}two small discs $\{D^1_{\ub}\subset \mathbb{S}^2\}$ and $\{D^2_{\ub}\subset \mathbb{S}^2\}$}
: \footnote{Here $D^1_{\ub}$ and $D^2_{\ub}$ could change for different $\ub$.}
\begin{equation}\label{special initial data 2}
\begin{split}
&0\leq |\chih_0|^2(\ub,\o)\leq a, \quad \mbox{for} \quad \o\in D^1_{\ub}\cup D^2_{\ub}, \\
&|\chih_0|^2(\ub, \o)={\color{black}[1+o(1)]\cdot} a, \quad \mbox{for} \quad \o\in\mathbb{S}^2\setminus (D^1_{\ub}\cup D^2_{\ub}), \,\,\mbox{with} \quad |o(1)|\leq 1/c,\\
&\iint_{D^1_{\ub}\cup D^2_{\ub}}1\cdot d\o\lesssim \f{1}{c^2}, \quad \mbox{where} \,\, 1\ll c \ll b \leq \at.
\end{split}
\end{equation}
{\color{black}\noindent Note: in Appendix \ref{Appendix C}, we will construct examples of initial data satisfying both (\ref{special initial data 1}) and (\ref{special initial data 2}) at the same time.}

\end{enumerate}

}

\noindent Using the following equation in \cite{AL} 
$$\nab_3\chih+\f12\tr\chib \chih=\nab \hat{\otimes}\eta+2\omb\chih-\f12\tr\chi\chibh+\eta\hat{\otimes}\eta,$$
together with the estimates derived there, we have
$$|u|^2|\chih|^2(u,\ub,\o)\approx|\chih|^2(0, \ub, \o)=|\chih|^2(0, \ub,\o)\cdot [1+o(1)].$$
Back to (\ref{tilde nu}), we hence get 
{\color{black}
\begin{equation*}
\begin{split}
&-\tilde{\nu}(\ub, \o; \ub')\f{a}{\ub^2 a^2}\\
=&-\f12\int_0^1 [1+o(1)]\cdot |\chih|^2 (1-\tau R(\ub',\o)-(1-\tau)R(\ub,\o), \tau \ub'+(1-\tau)\ub, \o) d\tau\\ 
=&-\f12\int_0^1 [1+o(1)]\cdot \f{|\chih_0|^2 (\tau \ub'+(1-\tau)\ub, \o)}{\l\tau R(\ub',\o)+(1-\tau)R(\ub,\o)\r^2} d\tau.
\end{split}
\end{equation*}  
When $\ub'$ is very close to $\ub$, we have
$$\f{[\tau R(\ub',\o)+(1-\tau)R(\ub,\o)]^2}{1+o(1)}=[\f14+o(1)]\ub^2 a^2.$$
This implies
$$\tilde{\nu}(\ub, \o; \ub')=\f{2}{a}\int_0^1 [1+o(1)]\cdot |\chih_0|^2 (\tau \ub'+(1-\tau)\ub, \o) d\tau.$$
Denote
$$\bar{\tilde{\nu}}(\ub, \o; \ub'):=\f{1}{|\mathbb{S}^2|} \iint_{\mathbb{S}^2}\tilde{\nu}(\ub, \o; \ub') \, d\o.$$
With the help of Fubini's theorem, we then have 
\begin{equation*}
\begin{split}
\bar{\tilde{\nu}}(\ub, \o; \ub'):=&\f{1}{|\mathbb{S}^2|}\f{2}{a}\iint_{\mathbb{S}^2}\bigg(\int_0^1 [1+o(1)]\cdot |\chih_0|^2 (\tau \ub'+(1-\tau)\ub, \o) d\tau\bigg)d\o\\
=&\f{1}{|\mathbb{S}^2|}\f{2}{a}\int^1_0\bigg(\iint_{\mathbb{S}^2} [1+o(1)]\cdot |\chih_0|^2 (\tau \ub'+(1-\tau)\ub, \o) d\o\bigg)d\tau\\
=&\f{1}{|\mathbb{S}^2|}\f{2}{a}\cdot |\mathbb{S}^2| \cdot[1+o(1)]\cdot a.
\end{split}
\end{equation*} 
For the last step, we use that on any $D_{\tau \ub'+(1-\tau)\ub}$, it holds $|\chih_0|^2 (\tau \ub'+(1-\tau)\ub, \o)= [1+o(1)]a$ for almost all the points $\{\o\}$ on $\mathbb{S}^2$ except a small portion (with area $\lesssim \f{1}{c^2}$) of $\mathbb{S}^2$. For sufficient large $c$, we hence have
$$\bar{\tilde{\nu}}(\ub, \o; \ub')=2+o(1).$$
}
Back to (\ref{new h})
\begin{equation}\label{new new h}
\D_{R(\ub,\o)}h(\ub, \o; \ub')-\f{\nu(\ub, \o; \ub')}{\ub^2 a^2}h(\ub, \o; \ub')+\tilde{\nu}(\ub, \o; \ub')\f{a}{\ub^2 a^2}=0.
\end{equation}
With (\ref{new new nu}), we have
\begin{equation*}
\begin{split}
&\f{\nu(\ub, \o; \ub')}{\ub^2 a^2}\\
=&\int_0^1 \l \f{-\tau R(\ub', \o)-(1-\tau)R(\ub, \o)+[\tau\ub'+(1-\tau)\ub] a f(\ub, \o)+[\tau\ub'+(1-\tau)\ub]\at c_3}{\l\tau R(\ub',\o)+(1-\tau)R(\ub,\o)\r^3} \r d\tau.\\
\end{split}
\end{equation*} 
For $\ub'$ sufficiently close to $\ub$ and for sufficiently large $c$, we have
$$R(\ub,\o)=\f{\ub a}{2}\cdot[1+o(1)], \quad R(\ub',\o)=\f{\ub' a}{2}\cdot[1+o(1)], \quad f(\ub, \o)=1+o(1),$$
$$\bar{\tilde{\nu}}(\ub, \o; \ub')=2+o(1), \quad \nu(\ub,\o;\ub')=4+o(1).$$ 
Integrate (\ref{new new h}) over $M_{\ub}$. We then have
$$\iint_{M_{\ub}}\f{\nu(\ub, \o; \ub')}{\ub^2 a^2}h(\ub, \o; \ub')=\iint_{M_{\ub}}\tilde{\nu}(\ub, \o; \ub')\f{a}{\ub^2 a^2}.$$ 
This implies
$$\bar{h}(\ub; \ub'):=\f{1}{|M_{\ub}|}\iint_{M_{\ub}}h(\ub, \o; \ub')\, d\o=[\f12+o(1)]a.$$
From (\ref{new new h}), we have
\begin{equation}\label{new h bar}
\begin{split}
\D_{R(\ub,\o)}[h(\ub, \o; \ub')-\bar{h}(\ub; \ub')]&-\f{\nu(\ub, \o; \ub')}{\ub^2 a^2}[h(\ub, \o; \ub')-\bar{h}(\ub; \ub')]\\
&+\tilde{\nu}(\ub, \o; \ub')\f{a}{\ub^2 a^2}-\f{\nu(\ub, \o; \ub')}{\ub^2 a^2}\cdot \bar{h}(\ub;\ub')=0.
\end{split}
\end{equation}
Notice that
\begin{equation*}
\begin{split}
&\iint_{M_{\ub}}|\tilde{\nu}(\ub, \o; \ub')\f{a}{\ub^2 a^2}-\f{\nu(\ub, \o; \ub')}{\ub^2 a^2}\cdot \bar{h}(\ub;\ub')|\\
=&\iint_{M_{\ub}}|\bar{\tilde{\nu}}(\ub, \o; \ub')\f{a}{\ub^2 a^2}-\f{\nu(\ub, \o; \ub')}{\ub^2 a^2}\cdot \bar{h}(\ub;\ub')|+o(1)\cdot a\\
=&o(1)\cdot a+o(1)\cdot a\\
=&o(1)\cdot a.
\end{split}
\end{equation*}
Together with Sobolev embedding and $W^{2,2}$ elliptic estimates for (\ref{new h bar}), we have
\begin{equation*}
\begin{split}
&\|h(\ub,\o;\ub')-\bar{h}(\ub; \ub')\|_{L^{\infty}(M_{\ub})}\\
\leq& \|h(\ub,\o;\ub')-\bar{h}(\ub; \ub')\|_{W^{2,2}(M_{\ub})}\\
\leq& \|\tilde{\nu}(\ub, \o; \ub')\f{a}{\ub^2 a^2}-\f{\nu(\ub, \o; \ub')}{\ub^2 a^2}\cdot \bar{h}(\ub;\ub')\|_{L^{2}(M_{\ub})}\cdot \ub a\\
\leq&\|\tilde{\nu}(\ub, \o; \ub')\f{a}{\ub^2 a^2}-\f{\nu(\ub, \o; \ub')}{\ub^2 a^2}\cdot \bar{h}(\ub;\ub')\|^{\f12}_{L^{\infty}(M_{\ub})}\cdot\|\tilde{\nu}(\ub, \o; \ub')\f{a}{\ub^2 a^2}-\f{\nu(\ub, \o; \ub')}{\ub^2 a^2}\cdot \bar{h}(\ub;\ub')\|^{\f12}_{L^{1}(M_{\ub})}\cdot \ub a\\
\leq& o(1)\cdot a.
\end{split}
\end{equation*}
Since $\bar{h}(\ub;\ub')=[\f12+o(1)]a$, we thus have
$$h(\ub,\o;\ub')=[\f12+o(1)]a.$$
Following the procedure in the previous section, we conclude that 
\begin{equation}\label{key DH}
h(\ub, \theta_1, \theta_2)=[\f12+o(1)]a.
\end{equation}

\subsection{Spacelike Apparent Horizon}
Our parameter $a$ is a fixed large positive constant. Hence the tangent vectors $\f{\partial}{\partial \theta_1}, \f{\partial}{\partial \theta_2}, \f{\partial}{\partial \ub}$ are all spacelike. 
Let $\lambda_1, \lambda_2, \lambda_3$ be any real numbers with $\lambda_1^2+\lambda_2^2+\lambda_3^3\neq 0$. Recall $h(\ub, \theta_1, \theta_2)=[1/2+o(1)]a$. Then
\begin{equation*}
\begin{split}
&g'(\lambda_1 \f{\partial}{\partial \theta_1}+\lambda_2 \f{\partial}{\partial \theta_2}+\lambda_3 \f{\partial}{\partial \ub},\lambda_1 \f{\partial}{\partial \theta_1}+\lambda_2 \f{\partial}{\partial \theta_2}+\lambda_3 \f{\partial}{\partial \ub})\\
=&\lambda_1^2 \cdot g_{\theta_1 \theta_1}+\lambda_2^2 \cdot g_{\theta_2 \theta_2}+4\lambda_1 \lambda_3 \f{\partial R}{\partial \theta_1}+4\lambda_2 \lambda_3 \f{\partial R}{\partial \theta_2}+\lambda_3^2 \cdot h(\ub, \theta_1, \theta_2)\cdot [1+o(1)]\\
>& 0.
\end{split}
\end{equation*}
Therefore, our apparent horizon formed is spacelike. And according to the definitions given by Ashtekar and Krishnan in \cite{AK03, AK} and by Ashtekar and Galloway in \cite{AG}, it is a dynamical horizon.  

\section{The Area Law}\label{entropy}
In this section, we study the area of MOTS.
\begin{proposition}
For the MOTS $M_{\ub}$ constructed, we have
\begin{equation}
\lim_{\ub\rightarrow 0}\mbox{Area}(M_{\ub})=0.
\end{equation}
\end{proposition}

\begin{proof}
On each $M_{\ub}$ we have induced metric
\begin{equation*}
g'_{\theta_i \theta_j}=g_{\theta_i \theta_j}+\f{\partial (1-R)}{\partial \theta_i}\cdot \f{\partial(1-R)}{\partial \theta_j} \cdot g(\f{\partial}{\partial u}, \f{\partial}{\partial u})=g_{\theta_i \theta_j}.
\end{equation*}
Thus, along $\Hb_{\ub}$
$$\sqrt{\det g'}(u,\ub, \theta_1, \theta_2)=\sqrt{\det g}(u, \ub, \theta_1, \theta_2).$$
The first variation formula states
$$\f{\partial}{\partial \ub}g_{AB}=2\O\chi_{AB},$$
which implies
$$\f{\partial}{\partial \ub} \det g=\det g \cdot g^{AB}\cdot 2\O \chi_{AB},$$
and
$$\f{\partial}{\partial \ub} \sqrt{\det g}=\sqrt{\det g}\cdot \O \tr\chi.$$
Hence, we have
$$\f{\partial}{\partial \ub} \log \f{\sqrt{\det g} (u, \ub, \theta_1, \theta_2)}{\sqrt{\det g} (u, 0, \theta_1, \theta_2)}=\O \tr\chi (u, \ub, \theta_1, \theta_2).$$
This gives
$$|\f{\sqrt{\det g} (u, \ub, \theta_1, \theta_2)}{\sqrt{\det g} (u, 0, \theta_1, \theta_2)}-1|\leq \f{\ub \at b^{\f14}}{|u|}\ll 1,$$
and
$$\f12 \sqrt{\det g} (u, 0, \theta_1, \theta_2)\leq \sqrt{\det g} (u, \ub, \theta_1, \theta_2)\leq \f32 \sqrt{\det g} (u, 0, \theta_1, \theta_2).$$
Therefore, we conclude
\begin{equation*}
\begin{split}
\mbox{Area}(M_{\ub})=&\iint_{\mathbb{S}^2}\sqrt{\det g} (u, \ub, \theta_1, \theta) d\theta_1 d\theta_2\\
\leq& \f32\iint_{\mathbb{S}^2}\sqrt{\det g} (u, 0, \theta_1, \theta) d\theta_1 d\theta_2\\
=& \f32\cdot 4\pi\cdot |1-u|^2\\
\leq& \f32\cdot 4\pi \cdot \f{25}{64}\ub^2 a^2.
\end{split}
\end{equation*}
For the last inequality, we use along the apparent horizon $u=1-R(\ub, \theta_1, \theta_2)$ and a priori estimate 
$$\f38 \ub a\leq R(\ub, \theta_1, \theta_2) \leq \f58 \ub a.$$
Similarly, 
\begin{equation*}
\begin{split}
\mbox{Area}(M_{\ub})=&\iint_{\mathbb{S}^2}\sqrt{\det g} (u, \ub, \theta_1, \theta) d\theta_1 d\theta_2\\
\geq& \f12\iint_{\mathbb{S}^2}\sqrt{\det g} (u, 0, \theta_1, \theta) d\theta_1 d\theta_2\\
=& \f12\cdot 4\pi\cdot |1-u|^2\\
\geq& \f12\cdot 4\pi \cdot \f{9}{64}\ub^2 a^2.
\end{split}
\end{equation*}
In summary, we establish
$$\f12\cdot 4\pi \cdot \f{9}{64}\ub^2 a^2\leq \mbox{Area}(M_{\ub})\leq \f32\cdot 4\pi \cdot \f{25}{64}\ub^2 a^2, \quad \mbox{and} \quad \lim_{\ub\rightarrow 0}\mbox{Area}(M_{\ub})=0.$$
\end{proof}
We now turn to prove
\begin{proposition}
Given two different MOTSs $M_{\ub'}$ and $M_{\ub}$, we have
\begin{equation}\label{change of area 3}
\mbox{Area}(M_{\ub'})>\mbox{Area}(M_{\ub}) \quad \mbox{for} \quad \ub'>\ub.
\end{equation}
\end{proposition}
\begin{proof}
The apparent horizon AH: $u=1-R(\ub, \theta_1, \theta_2)$ is foliated by $\{M_{\ub}\}$. Let $g'_{\mu\nu}$ be the induced metric on AH and denote $h_{\mu\nu}$ to be the induced metric on $M_{\ub}$. 
Since AH is spacelike, there exists a vector field $r^{\mu}$ satisfies $h_{\mu\nu}r^{\mu}=0$ ($r^{\mu}$ orthogonal to each $M_{\ub}$) and $g'_{\mu\nu}r^{\mu}r^{\nu}=1$. 

Define a function $\lambda: \mbox{AH} \rightarrow \mathbb{R}$, such that each $M_{\ub}$ is a level set of $\lambda$. Assume $C(s)$ is an integral curve of $r^{\mu}$ and $s$ is the affine parameter. We have
$$\l\f{\partial}{\partial \lambda}\r^{\mu}=\l\f{\partial}{\partial s}\r^{\mu}\cdot\f{d s}{d \lambda}=wr^{\mu}, \quad \quad \mbox{where} \quad \quad w:=\f{ds}{d\lambda} \quad \mbox{and} \quad w>0.$$
Denote $A(\lambda)$ to be the area of each level set of $\lambda$: $A(\lambda):=\iint_{\mathbb{S}^2}\sqrt{\det h}(\lambda, \theta_1, \theta_2)d\theta_1 d\theta_2$. We calculate
$$\f{d A(\lambda)}{d \lambda}=\f12\iint_{\mathbb{S}^2}\f{1}{\sqrt{\mbox{det }h}}\f{\partial (\mbox{det} h)}{\partial \lambda} d\theta_1 d\theta_2=\f12 \iint_{\mathbb{S}^2} \sqrt{\mbox{det} h}\cdot h^{\mu\nu}\cdot \f{\partial h_{\mu\nu}}{\partial \lambda}.$$
To proceed, we combine the facts
$$\f{\partial h_{\mu\nu}}{\partial \lambda}=\mathcal{L}_{\f{\partial}{\partial \lambda}} h_{\mu\nu}=\mathcal{L}_{wr}h_{\mu\nu},$$
and 
$$\mathcal{L}_{wr}h_{\mu\nu}=w\mathcal{L}_r h_{\mu\nu}+r^{\gamma}h_{\gamma \mu}D_{\nu} w+r^{\gamma}h_{\gamma \nu}D_{\mu}w=w\mathcal{L}_r h_{\mu\nu}.$$
For the last equality, we use $h_{\mu\nu}r^{\mu}=0$. Therefore, we arrive at 
\begin{equation}\label{change of area}
\f{d A(\lambda)}{d \lambda}=\f12 \iint_{\mathbb{S}^2} w\cdot \sqrt{\mbox{det} h}\cdot h^{\mu\nu}\cdot \mathcal{L}_r h_{\mu\nu}.
\end{equation}
Recall AH is spacelike, at each point $P$ on AH there exists a vector field $t^{\mu}$ orthogonal to the tangent plane of AH at P and satisfying $g_{\mu\nu}t^{\mu}t^{\nu}=-1$. Together with $g_{\mu\nu}r^{\mu}r^{\nu}=g'_{\mu\nu}r^{\mu}r^{\nu}=1$, we have that $t^{\mu}+r^{\mu}$ is an outgoing null vector and $t^{\mu}-r^{\mu}$ is an incoming null vector. Therefore, there exist positive functions $l_3$ and $l_4$ satisfying 
$$t^{\mu}+r^{\mu}=l_4 \cdot e'_4, \quad \quad \quad t^{\mu}-r^{\mu}=l_3 \cdot e'_3.$$ 
Hence
$$r^{\mu}=\f12 \cdot l_4 \cdot e'_4-\f12 \cdot l_3 \cdot e'_3.$$
As a consequence, it follows 
$$h^{\mu\nu}\cdot \mathcal{L}_r h_{\mu\nu}=\f12 \cdot l_4 \cdot \tr\chi'-\f12 \cdot l_3 \cdot \tr\chib'=-\f12\cdot l_3 \cdot \tr\chib.$$
In the last equality, we utilize $\tr\chib'=\tr\chib$ and along AH $\tr\chi'=0$. Combining (\ref{change of area}) and $\tr\chib<0$\footnote{See \cite{AL}.}, $w>0$, $l_3>0$, we conclude
\begin{equation*}
\f{d A(\lambda)}{d \lambda}=\f12 \iint_{\mathbb{S}^2} w\cdot\sqrt{\mbox{det} h}\cdot -\f12 \cdot l_3 \cdot \tr\chib>0.
\end{equation*}
\end{proof}

Inspired by the works \cite{AK03, AK} of Ashtekar and Krishnan, we define the entropy of each $M_{\ub}$ to be its area. The area law (\ref{change of area 3}) shows that the entropy of $M_{\ub}$ grows as $\ub$ increases. This is corresponding to the second law of black hole mechanics along our apparent horizon.

\section{Towards Isolated Horizon}
In this section, we further prove Theorem \ref{thm1.9}.  

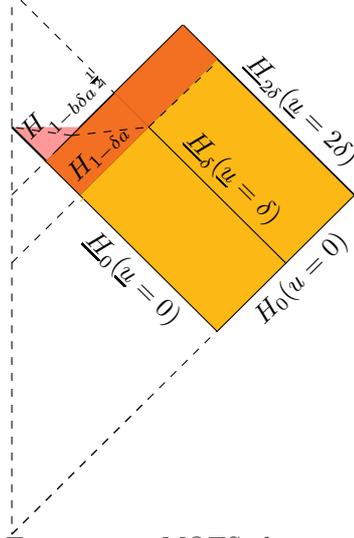
\begin{minipage}[!t]{0.4\textwidth}
\begin{tikzpicture}[scale=0.9]
%\draw (1.5,-0.5) node[very near end, sloped,below]{$H_{u_{\infty}}(u=u_{\infty})$}--(2,0);
\draw [white](3,-1)-- node[midway, sloped, below,black]{$H_0(u=0)$}(5,1);
\draw [white](2,2)--node [midway,sloped,above,black] {$\Hb_{\delta}(\ub=\delta)$}(4,0);
\draw [white](1,1)--node [midway,sloped, below,black] {$\Hb_{0}(\ub=0)$}(3,-1);
\draw [dashed] (0, 4)--(0, -4);
\draw [dashed] (0, -4)--(4,0)--(0,4);
\draw [dashed] (0,0)--(2,2);
\draw [dashed] (0,1)--(1.5,2.5);
\draw [dashed] (0,-4)--(2,-2);
\draw [dashed] (0,2)--(3,-1);
\draw [thick] (1,1)--(3,-1)--(4,0)--(2,2)--(1,1);
\draw [thick] (1,1)--(0.5,1.5)--(1.5,2.5)--(2,2)--(1,1);
\fill[yellow!70!red] (1,1)--(3,-1)--(4,0)--(2,2)--(1,1);
\fill[yellow!30!red](1,1)--(0.5,1.5)--(1.5,2.5)--(2,2)--(1,1);
\fill[red!40!](1, 2)--(0,2)--(0.5, 1.5)--(1, 2);
\draw [yellow!70!red](1,0.9)-- node[midway,sloped,above,black]{$H_{1-\delta a}$}(2,1.9);
%\draw [->] (3.3,-0.6)-- node[midway, sloped, above,black]{$L$}(3.6,-0.3);
%\draw [->] (1.4,1.3)-- node[midway, sloped, below,black]{$L$}(1.7,1.6);
%\draw [->] (3.3,0.6)-- node[midway, sloped, below,black]{$\Lb$}(2.7,1.2);
%\draw [->] (2.4,-0.3)-- node[midway, sloped, above,black]{$\Lb$}(1.7,0.4);
\draw [dashed] (0.875,1.875)--(2,2);
\draw[thick] (0,2)--(0.5, 1.5);
\draw [dashed] (0.875, 1.875)--(0,2);
\draw[thick] (1.5, 2.5)--(2.5, 3.5);
\draw[thick] (2.5, 3.5)--(3,3);
\draw[thick] (2,2)--(3,3);
\draw[thick] (3,3)--(5,1);
\draw[thick] (5,1)--(4,0);
\fill[yellow!30!red](1.5,2.5)--(2.5,3.5)--(3,3)--(2,2)--(1.5,2.5);
\fill[yellow!70!red] (2,2)--(3,3)--(5,1)--(4,0)--(2,2);
\draw[dashed] (2,2)--(3.02, 2.98);
\draw [thin](2,2)--node [midway,sloped,above,black] {$\Hb_{\delta}(\ub=\delta)$}(4,0);
\draw [thin](3,3)--node [midway,sloped,above,black] {$\Hb_{2\delta}(\ub=2\delta)$}(5,1);
\draw [thin](0.6,1.6)-- node[midway,sloped,above,black]{$H_{1-b \delta \at}$}(1.5,2.5);
%\draw [thick](0.5, 1.5)--(1, 2);
\end{tikzpicture}
\end{minipage}
\begin{minipage}[!t]{0.5\textwidth}

For $\d\leq \ub \leq 2\d$, with initial data satisfying
\begin{equation}\label{late input}
\int_{\d}^{\ub}|\chih_0|^2 (\ub', \o)d\ub'=0,
\end{equation}
we will see that there exists a unique MOTS along each $\Hb_{\ub}$ and all the MOTS $\{M_{\ub}\}_{\d\leq \ub \leq 2\d}$ together constitute a smooth 3-dimensional apparent horizon.\\

Note that, intuitively, this new piece of apparent horizon {\color{black} will be tilted toward outgoing null direction, and eventually it should be close to an isolated horizon. }For the precise definition of isolated horizon, interested readers are refereed to \cite{AK}.\\

\end{minipage}

To construct MOTS along each $\Hb_{\ub}$, we start from deriving estimates.

\subsection{A Priori Estimates}
For $\d\leq \ub \leq 2\d$, we denote
$$M_0(\ub, \o):=\int_0^{\ub} |\chih_0|^2 (0, \ub', \o)d\ub' \quad \mbox{and} \quad M_{0}^*(\o):=\int_0^{\d} |\chih_0|^2 (0, \ub', \o)d\ub'.$$
The condition (\ref{late input}) on initial data implies
$$M_0(\ub, \o)=M_0^*(\o) \quad \mbox{for} \quad \d \leq \ub \leq 2\d.$$
Note $u=1-R(\ub, \o)$. From the estimates in \cite{AL}, for $\d\leq \ub \leq 2\d$ we still have
$$\tr\chi=\f{2}{R(\ub, \o)}-\f{M_0^*(\o)}{R(\ub, \o)^2}+l.o.t., \quad \tr\chib=-\f{2}{R(\ub, \o)}+l.o.t., \quad \O=1+l.o.t..$$
And $\eta_b$, $\omb$ behave as lower order terms. Equation of MOTS (\ref{1.12}) is transferred to
$$\D'_M R(\ub, \o)-\f{1}{R(\ub, \o)}|\nab R(\ub, \o)|^2-\f{1}{R(\ub, \o)}+\f{M_0^*(\o)}{2R(\ub, \o)^2}+l.o.t.=0.$$
Recall
$$M_0^*(\o)=\delta a f(\delta, \o), \quad \mbox{with} \quad \f{20}{21}\leq f(\delta, \o) \leq \f{22}{21}.$$
With maximal principle, we derive $C^0$ estimates for $R(\ub, \o)$ 
\begin{equation}\label{C1 later}
\f{10}{21}[1+o(1)]\delta a \leq R(\ub, \o)\leq \f{11}{21}[1+o(1)]\delta a.
\end{equation}
For $C^1$ (gradient) estimates, we use Bochner's formula and by exploring the structures of equation (\ref{1.12}). Here we construct
$$h \big( R(\ub, \o) \big)=1+\f{8}{\delta^2 a^2}\l R(\ub, \o)-\f{\delta a}{2}\r^2.$$
With Bochner's formula and (\ref{1.12}), we then calculate $\D'_M\l h\big(R(\ub, \o)\big)|\nab R(\ub, \o)|^2\r$. With the same method as in Section \ref{C1 Estimate}, we conclude
\begin{equation*}
|\nab R (\ub, \o)|\ll1 \quad \mbox{for all} \quad \o\in S^2. 
\end{equation*}
By standard elliptic estimates, we further obtain
$$\|\nab^2 R(\ub, \o)\|_{L^p(M)}\lesssim (\delta a)^{-1+\f{2}{p}}, \quad \quad \|R(\ub, \o)\|_{C^{1,q}(M)}\lesssim (\delta a)^{-q}, \quad \quad \|R(\ub, \o)\|_{C^{2,q}(M)}\lesssim (\delta a)^{-1-q}.$$

\subsection{Existence of MOTS}
With a priori estimates, to solve (\ref{1.12}) we also employ the method of continuity.  The same argument as in Section \ref{Continuity1} and Section \ref{Continuity2} also works. The only modifications are 
\begin{enumerate}
\item when using $M_0(\ub, \o)$, we have $M_0(\ub, \o)=M_0^*(\o)=\d a f(\d, \o);$
\item  when showing the linearized operators are invertible, derived estimate (\ref{C1 later}) is used. 
\end{enumerate}

\subsection{Uniqueness of MOTS}

On a fixed $\Hb_{\ub}$, assume we have two solutions $\tR(\o)$ and $R(\o)$ satisfying $\tr\chi'=0$.  We then derive an elliptic equation for $\tR(\o)-R(\o)$. Together with a priori estimates and bounds derived in \cite{AL}, by a similar argument as in Section \ref{uniqueness} we have
\begin{equation*}
\begin{split}
&\D_{1-R(\o),\ub}\l \tR(\o)-R(\o)\r-\nu(\o)\l \tR(\o)-R(\o)\r\\
&+\f{1}{\delta^2 a^2}\cdot \l \tR(\o)-R(\o) \r \cdot o(1)\\
&+\f{1}{\delta^2 a^2}\cdot \f{\partial}{\partial \theta_i} \l \tR(\o)-R(\o)\r\cdot o(1)\\
=0,
\end{split}
\end{equation*}
with
$$\nu(\o)\geq \f{64}{81\delta^2 a^2}.$$
By maximal principle, we conclude that 
$$\tR(\o)=R(\o) \quad \mbox{for} \quad \o\in \mathbb{S}^2.$$

\subsection{Piecewise Smoothness of Apparent Horizon}
Along each incoming null hypersurface $\underline{H}_{\ub}$, there exists a unique 2-dimensional MOTS $M_{\ub}$, which is corresponding to the unique $C^2$ solution $R(\ub, \o)$ to $\tr\chi'=0$. Varying $\ub'$, $u=1-R(\ub', \o)$ is thus a three dimensional hypersurface. 

For different $\ub'$ and $\ub$, both in $[\d, 2\d]$, we then derive an elliptic equation for $\f{R(\ub', \o)-R(\ub, \o)}{\ub'-\ub}.$ Together with a priori estimates and bounds derived in \cite{AL}, as in Section \ref{regularity of horizon} we obtain
\begin{equation*}
\begin{split}
&\Delta_{1-R(\ub, \o), \ub} \l \f{R(\ub',\o)-R(\ub,\o)}{\ub'-\ub}-h(\ub, \o; \ub')\r-\f{\nu(\ub, \o; \ub')}{\ub^2 a^2} \l \f{R(\ub',\o)-R(\ub,\o)}{\ub'-\ub}-h(\ub, \o; \ub')\r \\
&+\f{1}{\ub^2 a^2} \cdot \l\f{R(\ub', \o)-R(\ub, \o)}{\ub'-\ub}-h(\ub, \o; \ub')\r \cdot o(1)\\
&+\f{1}{\ub^2 a^2} \cdot \l\f{\partial}{\partial \theta_i}\f{R(\ub', \o)-R(\ub, \o)}{\ub'-\ub}-h(\ub, \o; \ub')\r \cdot o(1)\\
=&\f{a}{\ub^2 a^2}\cdot o(1).
\end{split}
\end{equation*}
Here both $h(\ub, \o; \ub')$ and $\nu(\ub, \o; \ub')$ are smooth functions respect to $\ub$ and $\o$. And there exist smooth functions $h(\ub, \o)$ and $\nu(\ub, \o)$ such that 
$$h(\ub, \o)=\lim_{\ub'\rightarrow\ub}h(\ub, \o; \ub'), \quad \quad \nu(\ub,\o)=\lim_{\ub'\rightarrow\ub}\nu(\ub, \o; \ub').$$
When $\ub'$ is close to $\ub$, by a similar argument as in Section \ref{regularity of horizon} and Section \ref{section DH} we have
$$|h(\ub, \o; \ub')| \leq o(1)\cdot a, \quad \quad \quad \f{64}{81} \leq \nu(\ub, \o; \ub')\leq \f{1088}{375},$$
and
$$|h(\ub, \o)| \leq o(1)\cdot a, \quad \quad \quad \f{64}{81} \leq \nu(\ub, \o)\leq \f{1088}{375}.$$
Through elliptic estimates and standard argument for difference quotient, we have 
\begin{equation}\label{piecewise smooth}
|\f{\partial R(\ub,\o)}{\partial \ub}-h(\ub, \o)|\leq o(1)\cdot a, \quad \mbox{where} \quad |h(\ub, \o)|\leq o(1)\cdot a,
\end{equation}
$$\f{\partial^k R}{\partial \ub^k}\in C^{\infty}(M) \quad \mbox{for any} \quad k\in \mathbb{Z}^{+}.$$

\begin{remark}
For $\d\leq \ub \leq 2\d$, (\ref{piecewise smooth}) implies 
\begin{equation}\label{piecewise1}
|\f{\partial R(\ub,\o)}{\partial \ub}|\leq |h(\ub, \o)|+o(1)\cdot a \lesssim o(1)\cdot a.
\end{equation}
But in Section \ref{section DH}, for $0\leq \ub \leq \d$ we have 
$$|\f{\partial R(\ub,\o)}{\partial \ub}-h(\ub, \o)|\leq o(1)\cdot a, \quad \mbox{where} \quad h(\ub, \o)=[\f12+o(1)] a.$$
This implies for $0\leq \ub \leq \d$
\begin{equation}\label{piecewise2}
\f{\partial R(\ub,\o)}{\partial \ub}=h(\ub, \o)+o(1)\cdot a \geq 0.3a.
\end{equation}
Therefore, with the initial data as in Section \ref{section DH}, $\partial R(\ub, \o)/\partial \ub$ has a jump at $\ub=\d$. And (\ref{piecewise1}), (\ref{piecewise2}) together show that the constructed apparent horizon $\{M_{\ub}\}_{0\leq \ub \leq 2\d}$ is only piecewise smooth and it is not $C^1$ at $\ub=\delta$.

\end{remark}

\appendix
\section{Ricci Curvature of $M$} \label{AppendixA}
On each MOTS: $(1-R(\ub,\o), \ub, \o)$, here we derive a lower bound for ${\color{black}\mbox{Ric}_{M}}(\nab R, \nab R)$.\\ 

For the $2$-dimensional manifold $M$, we have 
$$R_{ij}=K'g_{ij},$$
where $K'$ is the Gaussian curvature of $M$ respect to frames $e_a', e_b', e_3', e_4'$.
Gauss equation gives
\begin{equation*}
K'=-\rho'+\f12 \chih'\cdot \chibh'-\f14 \tr\chi' \tr\chib'.
\end{equation*}
Let $F_a:=\O (\nab_a R)$ and $F=F^c e_c$. Recall that 
$$e'_b=e_b-F_b e_3, \quad e'_3=e_3, \quad e'_4=e_4-2F+|F|^2 e_3.$$
Therefore, we have
\begin{equation*}
\begin{split}
\chib'(e_a', e_b')=&g(D_{e_a'}e_3, e_b')=g(D_{e_a-F_a e_3}e_3, e'_b)\\
=&g(D_{e_a}e_3, e'_b)-F_a g(D_{e_3}e_3, e'_b)\\
=&g(D_{e_a}e_3, e_b-F_b e_3)=g(D_{e_a}e_3, e_b)-F_b g(D_{e_a}e_3, e_3)\\
=&\chib(e_a, e_b).
\end{split}
\end{equation*}
It follows 
$$\tr\chib'=\tr\chib, \quad \quad \chibh'=\chibh.$$
In the same fashion, we have
\begin{equation*}
\begin{split}
\chi'=&\chi_{ab}-(\nab_a F_b+\nab_b F_a)+\nab_3(F_a F_b)-(\zeta_b+\eta_b)F_a-(\zeta_a+\eta_a)F_b\\
&+|F|^2\chib_{ab}-F_b F^c\chib_{ac}-F_a F^c \chib_{bc}-4\omb F_a F_b.
\end{split}
\end{equation*}
Since 
$$\nab_3 F_a=\nab_3 (\O\nab R)=-\O\chib\cdot\nab R-2\omb\O\nab R,$$
contracting with metric, we then arrive at
$$\tr\chi'=\tr\chi-2\O\D R-4\O\eta\cdot\nab R-4\O^2\chibh_{bc}\nab^b R\nab^c R-\O^2\tr\chib|\nab R|^2-8\O^2\omb |\nab R|^2.$$
Furthermore notice that on $M$, we have $\tr\chi'=0$. Hence, 
\begin{equation*}
\begin{split}
\chih'_{ab}=&\chi'_{ab}-\f12\tr\chi' g_{ab}=\chi'\\
=&\chi_{ab}-(\nab_a F_b+\nab_b F_a)+\nab_3(F_a F_b)-(\zeta_b+\eta_b)F_a-(\zeta_a+\eta_a)F_b\\
&+|F|^2\chib_{ab}-F_b F^c\chib_{ac}-F_a F^c \chib_{bc}-4\omb F_a F_b\\
=&\chi_{ab}-\nab_a \O\nab_b R-\O \nab_a \nab_b R-\nab_b \O \nab_a R-\O\nab_b\nab_a R\\
&-\O^2 \chib_{ac}\nab^c R \nab_b R-2\omb \O^2 \nab_a R \nab_b R-\O^2 \chib_{bc} \nab^c R \nab_a R-2\omb \O^2 \nab_b R \nab_a R\\
&-\O(\zeta_b+\eta_b)\nab_a R-\O(\zeta_a+\eta_a)\nab_b R+\O^2\chib_{ab}|\nab R|^2\\
&-\O^2\chib_{ac}\nab_b R \nab^c R-\O^2\chib_{bc} \nab_a R \nab^c R-4\O^2\omb \nab_a R \nab_b R.
\end{split}
\end{equation*}
For $\rho'$, we have
\begin{equation*}
\begin{split}
\rho'=&\f14 R(e'_4, e'_3, e'_4, e'_3)\\
=&\f14 R(e_4-2F^c e_c+|F|^2e_3, e_3, e_4-2F^b e_b+|F|^2e_3, e_3)\\
=&\rho+2\O \beb_b \nab^b R+\O^2 \underline{\alpha}_{bc}\nab^c R \nab^b R.
\end{split}
\end{equation*}
Therefore, on $M$ we conclude
\begin{equation*}
\begin{split}
K'=&-\rho-2\O \beb_b \nab^b R-\O^2 \underline{\alpha}_{bc}\nab^c R \nab^b R\\
&+\f12\chibh^{ab}\chi_{ab}-\chibh^{ab}\nab_a \O \nab_b R-\O \chibh^{ab}\nab_a\nab_b R-\f12 \O^2 \chibh^{ab}\chib_{ac}\nab^c R \nab_b R\\
&-\O^2\omb \chib^{ab}\nab_a R \nab_b R-\f12 \O^2\chibh^{ab}\chib_{bc}\nab^c R\nab_a R-\O^2\omb \chih^{ab}\nab_a R \nab_b R\\
&-\f12\O\chibh^{ab} (\zeta_b+\eta_b)\nab_a R-\f12\O \chibh^{ab} (\zeta_a+\eta_a)\nab_b R+\f12\O^2|\nab R|^2 \chibh^{ab}\chib_{ab}\\
&-\f12\O^2\chibh^{ab}\chib_{ac}\nab_b R \nab^c R-\f12\O^2 \chibh^{ab}\chib_{bc}\nab_a R\nab^c R-2\O^2\omb \chibh^{ab}\nab_a R \nab_b R.
\end{split}
\end{equation*}
Here 
$$-\rho+\f12\chibh^{ab}\chi_{ab}=-\rho+\f12\chibh^{ab}\chih_{ab}=-\check{\rho}=\f{\ub a}{2R^3}f(\o)\cdot[1+o(1)]>0,$$
and 
\begin{equation*}
\begin{split}
&2| K'+\rho-\f12\chibh^{ab}\chi_{ab} |\leq\f{\ub \at}{R}\cdot\f{|\nab R|^2}{R^2}+\f{\ub \at}{R^2} \cdot|\nab^2 R|+\f{\ub a}{R^3}\cdot\f{\ub \at}{R}\cdot|\nab R|.
\end{split}
\end{equation*}
Hence 
\begin{equation*}
\begin{split}
2K'=&2\l K'+\rho-\f12\chibh^{ab}\chi_{ab}\r+2\l -\rho+\f12\chibh^{ab}\chi_{ab}\r\\
\geq& -2|K'+\rho-\f12\chibh^{ab}\chi_{ab}|+2\l -\rho+\f12\chibh^{ab}\chi_{ab}\r\\
\geq&-\f{\ub \at}{R}\cdot\f{|\nab R|^2}{R^2}-\f{\ub \at}{R^2} \cdot|\nab^2 R|-\f{\ub a}{R^3}\cdot\f{\ub \at}{R}\cdot|\nab R|.
\end{split} 
\end{equation*}

Therefore, we conclude
\begin{equation}\label{RicciBound}
\begin{split}
2{\color{black}\mbox{Ric}_{M}}(\nab R, \nab R)=&2K'|\nab R|^2\\
\geq& -\f{\ub \at}{R}\cdot\f{|\nab R|^4}{R^2}-\f{\ub \at}{R^2} \cdot|\nab^2 R|\cdot |\nab R|^2-\f{\ub a}{R^3}\cdot\f{\ub \at}{R}\cdot|\nab R|^3.
\end{split}
\end{equation}

\section{Construction of Initial Data Along $H_0^{[0,\delta]}$}\label{Appendix B}
 
\textbf{Goal of This Section.} By four steps, we will construct initial data along $H_0^{[0,\delta]}$ such that for any $\ub\in (0,\d]$ we have
\begin{equation}\label{goal of initial data}
\int_0^{\ub}|\chih|^2(0,\ub', \o)d\ub'=f(\ub, \o)\ub a, \quad \quad \mbox{with} \quad \quad \f{20}{21}\leq f(\ub, \o)\leq \f{22}{21}.
\end{equation}

\begin{remark}\label{Remark B.2}
In the below, we will focus on achieving (\ref{goal of initial data}). But with a similar argument, for any large positive constant $c$, we could also find initial data such that
\begin{equation}\label{1-c 1+c}
\int_0^{\ub}|\chih|^2(0,\ub', \o)d\ub'=f(\ub, \o)\ub a, \quad \quad \mbox{with} \quad \quad 1-\f{1}{c}\leq f(\ub, \o)\leq 1+\f{1}{c}
\end{equation}
{\color{black}by rescaling initial data satisfying
$$\int_0^{\d}|\chih|^2(0,\ub', \o)d\ub'=f(\d,\o) \d a,\quad \quad \mbox{with} \quad \quad 1-\f{1}{c}\leq f(\d, \o)\leq 1+\f{1}{c}.$$}
We use (\ref{1-c 1+c}) in Section \ref{Construction of Special Initial Data}.
\end{remark}

\begin{remark}
A first trying to achieve (\ref{goal of initial data}) is to require
\begin{equation}\label{first trying}
\f{20}{21}\cdot a \leq |\chih|^2(0, \ub', \o)\leq \f{22}{21}\cdot a
\end{equation}
for any $\o\in \mathbb{S}^2$, where $0\leq \ub' \leq \d$. If we have (\ref{first trying}), then (\ref{goal of initial data}) follows. However, by a topological argument, on any fixed $\S$, traceless two tensor $\chih_{ab}$ must vanish on at least one point. This first trying fails. In the below, we will give a more sophisticated approach to achieve (\ref{goal of initial data}). 

\end{remark}

\textbf{Step One.}
We choose a smooth function $\chih_0(\ub, \o)$. And we require that for fixed $\o$, $\chih_0(\ub, \o)\in C_{c}^{\infty}([0,\delta])$ in $\ub$ variable and {\color{black}for all $\o\in \mathbb{S}^2$},
$$\int_0^{\delta}|\chih_0|^2(\ub', \o)d\ub'=\d a.$$

\textbf{Step Two.}
We choose a number $\f{99}{100}$ with property 
$$\f{20}{21}\leq \f{99}{100} \leq 1 \leq \f{100}{99} \leq \f{22}{21}.$$ 
With this number, we decompose $(0,\d]$ into dyadic pieces
$$(0,\d]=\bigcup_{k=0}^{+\infty}[(\f{99}{100})^{k+1}\d, (\f{99}{100})^k \d].$$
In $[(\f{99}{100})^{k+1}\d, (\f{99}{100})^k \d]$, we let 
\begin{equation}\label{prescribe chi_0}
|\chih|^2(0,\ub, \o):=|\chih_0|^2\bigg([(\f{100}{99})^{k+1}\ub-\d]\cdot 99, \o \bigg).
\end{equation}
Thus, 
\begin{equation*}
\begin{split}
&\int_{(\f{99}{100})^{k+1}\d}^{(\f{99}{100})^{k}\d}|\chih|^2(0,\ub', \o)d\ub'\\
=&\int_{(\f{99}{100})^{k+1}\d}^{(\f{99}{100})^{k}\d}|\chih_0|^2\bigg([(\f{100}{99})^{k+1}\ub'-\d]\cdot 99, \o \bigg)d\ub'\\
=&(\f{99}{100})^{k+1}\cdot\f{1}{99}\cdot\int_{(\f{99}{100})^{k+1}\d}^{(\f{99}{100})^{k}\d}|\chih_0|^2\bigg([(\f{100}{99})^{k+1}\ub'-\d]\cdot 99, \o \bigg)d \bigg([(\f{100}{99})^{k+1}\ub'-\d]\cdot 99 \bigg)\\
=&(\f{99}{100})^{k+1}\cdot\f{1}{99}\cdot \int_0^{\d}|\chih_0|^2(\ub', \o)d\ub'\\
=&(\f{99}{100})^{k+1}\cdot\f{1}{99}\cdot \d a\\
=&(\f{99}{100})^k\cdot \f{1}{100}\cdot \d a.
\end{split}
\end{equation*}
And for fixed $\o$, $|\chih|^2(0,\ub, \o)\in C_c([(\f{99}{100})^{k+1}\d, (\f{99}{100})^{k}\d])$ in $\ub$ variable.  Furthermore, we have
\begin{equation*}
\begin{split}
&\sum_{k=0}^{+\infty}\int_{(\f{99}{100})^{k+1}\d}^{(\f{99}{100})^{k}\d}|\chih|^2(0,\ub', \o)d\ub'\\
=&\sum_{k=0}^{+\infty} (\f{99}{100})^k\cdot \f{1}{100}\cdot \d a\\
=&\f{\f{1}{100}\cdot \d a}{1-\f{99}{100}}=\d a.
\end{split}
\end{equation*}

\textbf{Step Three.} For any $\ub\in (0,\delta]$, there exists $N_0\in \mathbb{N}$ such that 
$$(\f{99}{100})^{N_0+1}\d\leq \ub \leq (\f{99}{100})^{N_0}\d.$$ 
We thus have
$$ \int_0^{(\f{99}{100})^{N_0+1}\d}|\chih|^2(0, \ub', \o)d\ub' \leq \int_0^{\ub}|\chih|^2(0, \ub', \o)d\ub' \leq \int_0^{(\f{99}{100})^{N_0}\d}|\chih|^2(0, \ub', \o)d\ub'.$$
On one side
\begin{equation*}
\begin{split}
&\int_0^{\ub}|\chih|^2(0, \ub', \o)d\ub' \\
\leq&\int_0^{(\f{99}{100})^{N_0}\d}|\chih|^2(0, \ub', \o)d\ub'\\
=&\f{(\f{99}{100})^{N_0}\cdot\f{1}{100}\cdot \d a}{1-\f{99}{100}}=(\f{99}{100})^{N_0}\cdot \d a\\
=& \f{100}{99}\cdot (\f{99}{100})^{N_0+1}\cdot \d a\leq \f{100}{99}\ub a \leq \f{22}{21}\ub a. 
\end{split}
\end{equation*}
On the other side
\begin{equation*}
\begin{split}
&\int_0^{\ub}|\chih|^2(0, \ub', \o)d\ub' \\
\geq&\int_0^{(\f{99}{100})^{N_0+1}\d}|\chih|^2(0, \ub', \o)d\ub'\\
=&\f{(\f{99}{100})^{N_0+1}\cdot\f{1}{100}\cdot \d a}{1-\f{99}{100}}=(\f{99}{100})^{N_0+1}\cdot \d a\\
=&\f{99}{100}\cdot (\f{99}{100})^{N_0}\cdot \d a \geq \f{99}{100}\ub a\geq \f{20}{21}\ub a.
\end{split}
\end{equation*}
Putting the above inequalities together, for any $\ub\in (0, \d]$ we have
$$\int_0^{\ub}|\chih|^2(0,\ub', \o)d\ub'=f(\ub, \o)\ub a, \quad \quad \mbox{with} \quad \quad \f{20}{21}\leq f(\ub, \o)\leq \f{22}{21}.$$

\textbf{Step Four.}With initial data $|\chih|^2(0, \ub, \o)$ prescribed along $H_0^{(0,\d]}$, we need to further check the hyperbolic part. Here we need to note that $\partial_{\ub} \chih$ and $\a$ are very large and tend to $+\infty$ as $\ub\rightarrow 0$. However, in \cite{AL} by a method of renormalization we avoid using $\a$.  Replace $\d$ by $\ub$ in \cite{AL}, for any $\ub\in (0, \d]$, we consider the region $(u,\ub')\in [b\ub \at, 1]\times [0,\ub]$.
All the arguments in \cite{AL} hold for this region. Then let $\ub\rightarrow \d$. With initial data $|\chih|^2(0, \ub, \o)$ prescribed along $u=0$ and Minkowski data prescribed along $\ub=0$, we then have the existence of Einstein vacuum equation in the whole colored region above.  And we have a sequence of $M_{\ub}$, of which the radius $\rightarrow 0$ as $\ub\rightarrow 0$.

\begin{remark}
Another way to construct arbitrary small MOTS is using the following initial data along $u=0$: fix any large natural number $N_1$, we then pick up another natural number $N_2$ such that $1\ll N_1 \ll N_2$ . For $\ub\leq (\f{99}{100})^{N_2+1}$, Minkowski data are prescribed along $u=0$;  for $\ub \geq (\f{99}{100})^{N_2+1}$ we prescribe $|\chih|^2(0, \ub, \o)$ as in Step One to Step Three above. Since $N_1\ll N_2$, for $\ub \geq (\f{99}{100})^{N_1+1}$ it still holds
$$\int_0^{\ub}|\chih|^2(0,\ub', \o)d\ub'=f(\ub, \o)\ub a, \quad \quad \mbox{with} \quad \quad \f{20}{21}\leq f(\ub, \o)\leq \f{22}{21}.$$
Then the smallest MOTS is of radius $(\f{99}{100})^{N_1+1}\cdot a$. Since $N_1$ could be any large positive integer, the radius of the smallest MOTS could be arbitrary small.
\end{remark}

{\color{black}

\section{Construction of Initial Data Along $H_0^{[0,\delta]}$: Part II}\label{Appendix C}

In this part, we provide examples of initial data satisfying the following two requirements at the same time:
\begin{enumerate}
\item The initial data along $H_0^{[0,\d]}$ satisfy $$\int_0^{\d}|\chih_0|^2(\ub,\o)d\ub=f(\d,\o)\d a, \,\, \mbox{with} \,\, 1-
\f{1}{c}\leq f(\d, \o)\leq 1+\f{1}{c}, \,\, \mbox{where} \,\, 1\ll c \ll b \leq \at.$$

\item Moreover, we require that $|\chih_0|^2(\ub, \o)$ is almost a constant out of {\color{black}two} small discs {\color{black}$\{D^1_{\ub}\subset \mathbb{S}^2\}$ and $\{D^2_{\ub}\subset \mathbb{S}^2\}$}
: \footnote{Here $D^1_{\ub}$ and $D^2_{\ub}$ could change for different $\ub$.}
\begin{equation*}
\begin{split}
&0\leq |\chih_0|^2(\ub,\o)\lesssim a, \quad \mbox{for} \quad \o\in D^1_{\ub}\cup D^2_{\ub},\\
&|\chih_0|^2(\ub, \o)=[1+o(1)]\cdot a, \quad \mbox{for} \quad \o\in\mathbb{S}^2\setminus \big(D^1_{\ub}\cup D^2_{\ub}\big) \quad \mbox{with} \quad |o(1)|\leq 1/c,\\
&\iint_{D^1_{\ub}\cup D^2_{\ub}}1\cdot d\o\lesssim \f{1}{c^2}.
\end{split}
\end{equation*}
\end{enumerate}

\begin{minipage}[!t]{0.4\textwidth}
\begin{tikzpicture}[scale=0.9]
\draw [white](3,-1)-- node[midway, sloped, below,black]{$H_0(u=0)$}(4,0);

\draw [white](0.5,1.5)-- node[midway,sloped,above,black]{$H_{1-b \delta \at}$}(1.5,2.5);
\draw [white](2,2)--node [midway,sloped,above,black] {$\Hb_{\delta}(\ub=\delta)$}(4,0);
\draw [white](1,1)--node [midway,sloped, below,black] {$\Hb_{0}(\ub=0)$}(3,-1);
\draw [dashed] (0, 4)--(0, -4);
\draw [dashed] (0, -4)--(4,0)--(0,4);
\draw [dashed] (0,0)--(2,2);
\draw [dashed] (0,1)--(1.5,2.5);
\draw [dashed] (0,-4)--(2,-2);
\draw [dashed] (0,2)--(3,-1);
\draw [very thick] (1,1)--(3,-1)--(4,0)--(2,2)--(1,1);
\draw [very thick] (1,1)--(0.5,1.5)--(1.5,2.5)--(2,2)--(1,1);
\fill[yellow!70!red] (1,1)--(3,-1)--(4,0)--(2,2)--(1,1);
\fill[yellow!30!red](1,1)--(0.5,1.5)--(1.5,2.5)--(2,2)--(1,1);
\draw [white](1,1)-- node[midway,sloped,above,black]{$H_{1-\delta a}$}(2,2);
\draw [->] (3.3,-0.6)-- node[midway, sloped, above,black]{$e_4$}(3.6,-0.3);
\draw [->] (1.4,1.3)-- node[midway, sloped, below,black]{$e_4$}(1.7,1.6);
\draw [->] (3.3,0.6)-- node[midway, sloped, below,black]{$e_3$}(2.7,1.2);
\draw [->] (2.4,-0.3)-- node[midway, sloped, above,black]{$e_3$}(1.7,0.4);
\end{tikzpicture}
\end{minipage}
\begin{minipage}[!t]{0.55\textwidth}
We now prescribe characteristic initial data. We follow some basic calculations in Chapter 2 of \cite{Chr:book}. For $\ub\leq 0$, we prescribe Minkowskian initial data. For $0 \leq R \leq 1$, the surface $S_{1-R, 0}$ on the boundary $\underline{H}_0$ of Minkowskian region, is the sphere of radius $R$ in the hyperplane $t=-R$:
$$|x|=R, \quad \quad |x|=\sqrt{x_1^2+x_2^2+x_3^3}.$$
We use coordinates $(x_1, x_2, x_3)$ to define stereographic coordinates $(\theta_1, \theta_2)$ on $S_{0,0}$. We thus have two stereographic charts, the north polar chart and the south polar chart.

\end{minipage}
\hspace{0.05\textwidth}

Denote $q_2=(0,0,-1)$ to be the south pole. The domain of the north polar chart is then $U_1=S_{0,0} \backslash q_2$. And the chart is the mapping of $U_1$ onto $\mathbb{R}^2$ by $(x_1, x_2, x_3)\in U_1 \rightarrow (\theta_1, \theta_2)\in \mathbb{R}^2$:
$$\theta_1=\f{2x_1}{1+x_3}, \quad \quad \theta_2=\f{2x_2}{1+x_3}.$$
Similarly, for the south polar chart:
$$\theta_1=\f{2x_1}{1-x_3}, \quad \quad \theta_2=\f{2x_2}{1-x_3}.$$
In both charts, the standard metric on $S_{0,0}$ is 
$$\mathring{g}_{AB}(\theta_1, \theta_2)=\f{\d_{AB}}{(1+\f14(\theta_1^2+\theta_2^2))^2}.$$
The transformation from north polar coordinates to south polar coordinates is 
$$(\theta_1, \theta_2)\rightarrow f(\theta_1, \theta_2)=\f{(4\theta_1, 4\theta_2)}{\theta_1^2+\theta_2^2}.$$
Note that $f=f^{-1}$. We define
$(\theta_1', \theta_2')=f(\theta_1, \theta_2)$ and hence $(\theta_1, \theta_2)=f(\theta_1', \theta_2').$\\

We then prescribe initial data along $H_0$. The induced metric $g|_{S_{0,\ub}}$ can be expressed in the form
$$g|_{S_{0,\ub}}=(\phi|_{S_{0,\ub}})^2 \hat{g}|_{S_{0,\ub}},$$
where $\hat{g}|_{S_{0,\ub}}$ satisfies that $\Phi^*_{\ub}\hat{g}|_{S_{0,\ub}}$, a metric on $S_{0,0}$, has the same area form as $\mathring{g}|_{S_{0,0}}$:
$$d\mu_{\Phi^*_{\ub}\hat{g}|_{S_{0,\ub}}}=d\mu_{\mathring{g}|_{S_{0,0}}}.$$
In another word, $$\sqrt{\det\hat{g}(0,\ub,\theta_1, \theta_2)}=W^2(\theta_1, \theta_2), \quad  \mbox{where} \quad W(\theta_1, \theta_2)=\f{1}{1+\f14(\theta_1^2+\theta_2^2)}.$$
Thus, in both stereographic charts $\hat{g}$ is given by 
$$\hat{g}_{AB}(0,\ub,\theta_1, \theta_2)=W^2(\theta_1, \theta_2)\,m_{AB}(0,\ub, \theta_1, \theta_2),$$ 
where $m$ satisfies $\det m=1$. Set 
$$O_{CA}(\theta_1, \theta_2)=\d_{CA}-\f{2\theta_C \theta_A}{\theta_1^2+\theta_2^2} \quad \mbox{and} \quad \tilde{O} \quad \mbox{its transpose}.$$
 It is straight forward to check 
$$\tilde{O}O=I, \quad \quad \tilde{O}=O, \quad \quad \det{O}=-1.$$
From Chapter 2 in \cite{Chr:book}, we also have 
\begin{equation*}
m(0, \ub, \theta_1, \theta_2)=\tilde{O}(\theta_1, \theta_2)\,m'(0,\ub,\theta_1', \theta_2')\,O(\theta_1, \theta_2),
\end{equation*}
which we write as 
\begin{equation}\label{m m'}
m=\tilde{O}m' O.
\end{equation}
Now the matrix $m$ at a given point on $H_0$ is a 2-dimensional positive definite symmetric unimodular matrix. Such a matrix has the form 
$m=
 \begin{bmatrix}
    Z+X & Y \\
    Y & Z-X
 \end{bmatrix},
$
where $Z^2-X^2-Y^2=1$ and $Z\geq 0$. Use exponential map, {\color{black}in north polar chart} we can express: 
$$m(0, \ub, \theta_1, \theta_2)=\exp \Psi(\ub, \theta_1, \theta_2),$$
where $\Psi$ is a symmetric trace-free 2-dimensional matrix. And $\Psi(\ub, \theta_1, \theta_2)$ is the free data we can prescribe along $H_0$. {\color{black}In the below, we will prescribe $\Psi(\ub, \theta_1, \theta_2)$ in the north polar chart and  $\Psi'(\ub, \theta_1', \theta_2')$ in the south polar chart such that (\ref{m m'}) is satisfied.}\\

Consider the standard spherical coordinates $\{\theta, \phi\}$ on $\mathbb{S}^2$, where $0\leq \theta \leq \pi$ and $0\leq\phi<2\pi$. Assume $\theta=0$ and $\theta=\pi$ are corresponding to the north pole and south pole, respectively. For constant $c$, we assume  $1\ll c \ll b \leq \at$.

\begin{minipage}[!t]{0.23\textwidth}
\begin{tikzpicture}[scale=0.83]
    \draw [white](-0.5,1.05)-- node[midway,sloped,above,black]{$Q$}(-0.46,1.05);
     \draw [white](-0.9,-0.8)-- node[midway,sloped,above,black]{$\theta=\f{\pi}{2}$}(-0.5,-0.8);
    \draw [white](0.5,-1.05)-- node[midway,sloped,below,black]{$P$}(0.46,-1.05);
    \draw (-2,0) arc (180:360:2cm and 1cm);
    \draw[dashed] (-2,0) arc (180:0:2cm and 1cm);
    %\draw (0,2) arc (90:270:1cm and 2cm);
    %\draw[dashed] (0,2) arc (90:-90:1cm and 2cm);
    \draw (0,0) circle (2cm);
    \shade[ball color=black!60!white,opacity=0.2] (0,0) circle (2cm);
    \shade[ball color=gray!60, opacity=0.4] (-0.57,0.95) circle (5pt);
    \shade[ball color=gray!80, opacity=0.6] (-0.57,0.95) circle (3pt);
    %\shade[ball color=gray!90, opacity=0.6] (-0.56,1.6) circle (2pt);
    \filldraw (-0.57,0.95) circle (2pt);
    \shade[ball color=gray!60, opacity=0.4] (0.57,-0.95) circle (5pt);
    \shade[ball color=gray!80, opacity=0.6] (0.57,-0.95) circle (3pt);
    \filldraw (0.57,-0.95) circle (2pt);
    \end{tikzpicture}
\end{minipage}
\begin{minipage}[!t]{0.75\textwidth}
In the following, we will prescribe $\Psi(\ub, \theta_1, \theta_2)$ and {\color{black}$\Psi'(\ub, \theta_1', \theta_2')$} such that for fixed $\ub$, in $\{\theta, \phi\}$ coordinates, we have 
\begin{itemize}
\item $|\chih|^2_0(\ub, \theta, \phi)=0$ for $|\phi-{2\pi\ub}/{\d}|\leq \pi/2c$ and $|\theta-\pi/2|\leq \pi/2c$,
\item $|\chih|^2_0(\ub, \theta, \phi)\leq a$ for $|\phi-{2\pi\ub}/{\d}|\leq \pi/c$ and $|\theta-\pi/2|\leq \pi/c$,
\item  $|\chih|^2_0(\ub, \theta, \phi)=0$ for $|{\color{black}\pi}-\phi+{2\pi\ub}/{\d}|\leq \pi/2c$ and $|\phi|\leq \pi/2c$,
\item $|\chih|^2_0(\ub, \theta, \phi)\leq a$ for $|{\color{black}\pi}-\phi+{2\pi\ub}/{\d}|\leq \pi/c$ and $|\phi|\leq \pi/c$,
\item $|\chih|^2_0(\ub, \theta, \phi)=a$ otherwise. 
\end{itemize}
\end{minipage}
In another word we hope, for fixed $\ub$, it holds $|\chih|^2_0(\ub, \theta, \phi)=a$ for most  points on $\mathbb{S}^2$, except for $(\theta, \phi)$ lying in two small discs (with radius $\pi/c$) centered at $P$ and $Q$. Here $P, Q$ have coordinates $\theta=\pi/2, \phi=2\pi\ub/\d$ and $\theta=\pi/2, \phi={\color{black}\pi}-2\pi\ub/\d$, respectively.  Within these two small discs, it holds $|\chih|^2_0(\ub, \theta, \phi)\leq a$. \\

Fix the north pole. For any point $P\in \mathbb{S}^2$ there is a $1-1$ correspondence between its stereographic coordinates $(\theta_1(P), \theta_2(P))$ and its spherical coordinates $(\theta(P), \phi(P))$. In particular, we could write 
$$(\theta, \phi)=(\theta(\theta_1, \theta_2), \phi(\theta_1, \theta_2)).$$ 
For simplicity, in the below we omit the explicit forms of $(\theta(\theta_1, \theta_2), \phi(\theta_1, \theta_2))$. \\

To prescribe $\Psi(\ub,\theta_1, \theta_2)$ {\color{black}in the north polar chart, for $\sqrt{\theta^2_1+\theta^2_2}<20$} we set $$\Psi(\ub,\theta_1, \theta_2)=\sqrt2 x(\theta, \phi-\f{2\pi\ub}{\d})\cdot\at\cdot\Psi_0(\ub, \theta_1, \theta_2).$$ 
For fixed $\theta$, we require $x(\theta, \cdot)$ to be a periodic function with period $2\pi$ and it will be constructed later. Let
$$\Psi_0(\ub, \theta_1, \theta_2)=
\begin{bmatrix}
1&0\\
0&-1
\end{bmatrix}, \quad \mbox{for} \quad \ub>0.$$
We also define
$$m(\ub,\theta_1, \theta_2)=\exp(\Psi(\ub,\theta_1, \theta_2))=\exp\big(\sqrt2 x(\theta, \phi-\f{2\pi\ub}{\d})\cdot\at\cdot\Psi_0(\ub, \theta_1, \theta_2)\big).$$
And for $\ub>0$, we have
$$m(\ub,\theta_1, \theta_2)=
\begin{bmatrix}
\exp\big(\sqrt2 x(\theta, \phi-\f{2\pi\ub}{\d})\cdot\at\big)&0\\
0&\exp\big(-\sqrt2 x(\theta, \phi-\f{2\pi\ub}{\d})\cdot\at\big)
\end{bmatrix},$$
$$m^{-1}(\ub,\theta_1, \theta_2)=
\begin{bmatrix}
\exp\big(-\sqrt2 x(\theta, \phi-\f{2\pi\ub}{\d})\cdot\at\big)&0\\
0&\exp\big(\sqrt2 x(\theta, \phi-\f{2\pi\ub}{\d})\cdot\at\big)
\end{bmatrix}.$$
Applying 
$$\f{d}{d\ub}e^{X(\ub)}=\int_0^1 e^{\tau X(\ub)}\f{dX(\ub)}{d\ub}e^{(1-\tau)X(\ub)}d\tau,$$
we get
\begin{equation*}
\begin{split}
\f{\partial}{\partial \ub}m&(\ub,\theta_1,\theta_2)=\int_0^1 \exp \begin{bmatrix}
\sqrt2 \tau x(\theta, \phi-\f{2\pi\ub}{\d})\cdot\at&0\\
0&-\sqrt2 x(\theta, \phi-\f{2\pi\ub}{\d})\cdot\at
\end{bmatrix}\\
&\times \bigg(\sqrt2\cdot\f{\partial x}{\partial \phi}(\theta, \phi-\f{2\pi\ub}{\d})\cdot\f{-2\pi}{\d}\cdot\at\bigg) \begin{bmatrix} 1&0\\0&-1 \end{bmatrix}\\
&\times \exp \begin{bmatrix}
\sqrt2 (1-\tau)\cdot x(\theta, \phi-\f{2\pi\ub}{\d})\cdot\at&0\\
0&-\sqrt2 (1-\tau)\cdot x(\theta, \phi-\f{2\pi\ub}{\d})\cdot\at
\end{bmatrix}d\tau\\
=&\int_0^1 \exp \begin{bmatrix}
\sqrt2 x(\theta, \phi-\f{2\pi\ub}{\d})\cdot\at&0\\
0&-\sqrt2 x(\theta, \phi-\f{2\pi\ub}{\d})\cdot\at
\end{bmatrix}\\
&\times \bigg(\sqrt2\cdot \f{\partial x}{\partial \phi}(\theta, \phi-\f{2\pi\ub}{\d})\cdot\f{-2\pi}{\d}\cdot\at\bigg)\begin{bmatrix} 1&0\\0&-1 \end{bmatrix} d\tau\\
=&\sqrt2\at \begin{bmatrix}
\exp\bigg(\sqrt2 x(\theta, \phi-\f{2\pi\ub}{\d})\cdot\at\bigg)&0\\
0&-\exp\bigg(-\sqrt2 x(\theta, \phi-\f{2\pi\ub}{\d})\cdot\at\bigg)
\end{bmatrix}\\
&\times \bigg(\f{\partial x}{\partial \phi}(\theta, \phi-\f{2\pi\ub}{\d})\cdot\f{-2\pi}{\d}\bigg).
\end{split}
\end{equation*}
From Chapter 2 in \cite{Chr:book}, it holds 
$$\chih_{AB}=\f12\phi^2\f{\partial \hat{g}}{\partial \ub}=\f12\phi^2 W^2(\theta_1, \theta_2)\f{\partial}{\partial \ub}m_{AB}(\ub,\theta_1, \theta_2),$$
and
\begin{equation*}
\begin{split}
e:=\f12 |\chih|^2_{g}=&\f12 (g^{-1})^{CA}\chih_{AB}(g^{-1})^{BD}\chih_{DC}\\
=&\f18 (\hat{g}^{-1})^{CA}\f{\partial \hat{g}_{AB}}{\partial \ub}(\hat{g}^{-1})^{BD}\f{\partial \hat{g}_{DC}}{\partial \ub}\\
=&\f18 (m^{-1})^{CA}(\f{\partial m}{\partial \ub})_{AB} (m^{-1})^{BD} (\f{\partial m}{\partial \ub})_{DC}.
\end{split}
\end{equation*}
For our concrete example, we have 
\begin{equation}\label{e value}
\begin{split}
e=&\f18 (m^{-1})^{CA}(\f{\partial m}{\partial \ub})_{AB} (m^{-1})^{BD} (\f{\partial m}{\partial \ub})_{DC}\\
=&\f18\cdot \bigg(\sqrt2\cdot \f{\partial x}{\partial \phi}(\theta, \phi-\f{2\pi\ub}{\d})\cdot\f{-2\pi}{\d}\cdot\at\bigg)^2 \cdot\tr \begin{bmatrix} 1&0\\0&1 \end{bmatrix}\\
=&\f12 a\cdot  \bigg(\f{\partial x}{\partial \phi}(\theta, \phi-\f{2\pi\ub}{\d})\cdot\f{-2\pi}{\d}\bigg)^2.
\end{split}
\end{equation}
We now prescribe the function $\partial x/\partial \phi(\theta, \phi)$. Let $\partial x/\partial \phi(\theta, \phi)$ to be $2\pi$-periodic with respect to $\phi$ and require

\[ |\f{\partial x}{\partial \phi}|(\theta,\phi)\left\{ \begin{array}{ll}
=0 & \mbox{if $-\f{\pi}{2c}\leq \phi \leq \f{\pi}{2c}$} \mbox{ and }  -\f{\pi}{2c}\leq\theta-\f{\pi}{2}\leq \f{\pi}{2c};\\
%=a/2 & \mbox{if $-\pi+\pi/c\leq\o\leq -\pi/c$} \mbox{ and } -\pi/2c\leq\phi\leq \pi/2c;\\
=0 & \mbox{if ${\color{black}\pi}-\f{\pi}{2c}\leq \phi \leq {\color{black}\pi}+\f{\pi}{2c}$} \mbox{ and }  -\f{\pi}{2c}\leq\theta-\f{\pi}{2}\leq \f{\pi}{2c};\\ 
\leq {\d}/{2\pi} & \mbox{if $-\f{\pi}{c}\leq \phi \leq \f{\pi}{c}$} \mbox{ and }  -\f{\pi}{2c}\leq\theta-\f{\pi}{2}\leq \f{\pi}{2c};\\
%=a/2 & \mbox{if $-\pi+\pi/c\leq\o\leq -\pi/c$} \mbox{ and } -\pi/2c\leq\phi\leq \pi/2c;\\
\leq {\d}/{2\pi} & \mbox{if ${\color{black}\pi}-\f{\pi}{c}\leq \phi \leq {\color{black}\pi}+\f{\pi}{c}$} \mbox{ and }  -\f{\pi}{2c}\leq\theta-\f{\pi}{2}\leq \f{\pi}{2c};\\ 
        ={\d}/{2\pi} & \mbox{otherwise}.
                \end{array}  \right. \] 
Then {\color{black}in the north polar chart $\sqrt{\theta_1^2+\theta^2_2}<20$} we have 
\begin{equation}\label{C3}
\f12a\cdot \bigg(\f{\partial x}{\partial \phi}(\theta, \phi-\f{2\pi\ub}{\d})\cdot\f{-2\pi}{\d}\bigg)^2 \left\{ \begin{array}{ll}
=0 & \mbox{if $-\f{\pi}{2c}\leq \phi-\f{2\pi\ub}{\d} \leq \f{\pi}{2c}$} \mbox{ and } -\f{\pi}{2c}\leq\theta-\f{\pi}{2}\leq \f{\pi}{2c};\\
%=a/2 & \mbox{if $-\pi+\pi/c\leq\o\leq -\pi/c$} \mbox{ and } -\pi/2c\leq\phi\leq \pi/2c;\\
=0 & \mbox{if ${\color{black}\pi}-\f{\pi}{2c}\leq \phi-\f{2\pi\ub}{\d} \leq {\color{black}\pi}+\f{\pi}{2c}$} \mbox{ and }  -\f{\pi}{2c}\leq\theta-\f{\pi}{2}\leq \f{\pi}{2c};\\ 
\leq {a}/{2} & \mbox{if $-\f{\pi}{c}\leq \phi-\f{2\pi\ub}{\d} \leq \f{\pi}{c}$} \mbox{ and }  -\f{\pi}{2c}\leq\theta-\f{\pi}{2}\leq \f{\pi}{2c};\\
%=a/2 & \mbox{if $-\pi+\pi/c\leq\o\leq -\pi/c$} \mbox{ and } -\pi/2c\leq\phi\leq \pi/2c;\\ 
\leq {a}/{2} & \mbox{if ${\color{black}\pi}-\f{\pi}{c}\leq \phi-\f{2\pi\ub}{\d} \leq {\color{black}\pi}+\f{\pi}{c}$} \mbox{ and }  -\f{\pi}{2c}\leq\theta-\f{\pi}{2}\leq \f{\pi}{2c};\\ 
        ={a}/{2} & \mbox{otherwise}.
                \end{array}  \right. 
                \end{equation}

{\color{black}Next we extend $\Psi(\ub, \theta_1, \theta_2)$ to the south polar chart. In $\f15<\sqrt{\theta_1^2+\theta_2^2}<20$,\footnote{That is corresponding to $20>\sqrt{\theta_1'^2+\theta_2'^2}>\f15$ in the south polar chart.} we require $\Psi'(\ub, \theta_1', \theta_2'):=O(\theta_1', \theta_2')\Psi(\ub, \theta_1, \theta_2)O(\theta_1, \theta_2)$. Since $m(\ub, \theta_1, \theta_2)=\exp\big(\Psi(\ub, \theta_1, \theta_2)\big)$ and $m'(\ub, \theta_1', \theta_2')=\exp\big(\Psi'(\ub, \theta_1', \theta_2')\big)$, (\ref{m m'}) follows. Denote $\tilde{\lambda}:=\sqrt2 x(\theta, \phi-\f{2\pi\ub}{\d})\cdot\at.$ We calculate
\begin{equation*}
\begin{split}
e'=&\f18 (m'^{-1})^{CA}(\f{\partial m'}{\partial \ub})_{AB} (m'^{-1})^{BD} (\f{\partial m'}{\partial \ub})_{DC}\\
=&\f18 {O^{C}}_{C'}(m^{-1})^{C'A'}{O_{A'}}^{A}{O_{A}}^{A''}(\f{\partial m}{\partial \ub})_{A''B'}{O^{B'}}_{B}{O^{B}}_{B''}(m^{-1})^{B''D'}{O_{D'}}^{D}{O_{D}}^{D''}(\f{\partial m}{\partial \ub})_{D''C''}{O^{C''}}_{C}\\
=&\f18 {O^{C}}_{C'}(m^{-1})^{C'A''}(\f{\partial m}{\partial \ub})_{A''B''}(m^{-1})^{B''D''}(\f{\partial m}{\partial \ub})_{D''C''}{O^{C''}}_{C}\\
=&\f18\cdot \bigg(\sqrt2\cdot \f{\partial x}{\partial \phi}(\theta, \phi-\f{2\pi\ub}{\d})\cdot\f{-2\pi}{\d}\cdot\at\bigg)^2\\
&\times {O^{C}}_{C'}\begin{bmatrix} e^{-\tilde{\lambda}}&0\\0&e^{\tilde{\lambda}}\end{bmatrix}^{C'A''}\begin{bmatrix} e^{\tilde{\lambda}}&0\\0&-e^{-\tilde{\lambda}}\end{bmatrix}_{A''B''}\begin{bmatrix} e^{-\tilde{\lambda}}&0\\0&e^{\tilde{\lambda}}\end{bmatrix}^{B''D''}\begin{bmatrix} e^{\tilde{\lambda}}&0\\0&-e^{-\tilde{\lambda}}\end{bmatrix}_{D''C''}  {O^{C''}}_{C}\\
=&\f18\cdot \bigg(\sqrt2\cdot \f{\partial x}{\partial \phi}(\theta, \phi-\f{2\pi\ub}{\d})\cdot\f{-2\pi}{\d}\cdot\at\bigg)^2 \cdot\tr \begin{bmatrix} 1&0\\0&1 \end{bmatrix}\\
=&\f12 a\cdot  \bigg(\f{\partial x}{\partial \phi}(\theta, \phi-\f{2\pi\ub}{\d})\cdot\f{-2\pi}{\d}\bigg)^2.
\end{split}
\end{equation*}
Note that by (\ref{e value}) $e$ and $e'$ take the same value at $(\theta_1, \theta_2)$ with $\f15<\sqrt{\theta_1^2+\theta_2^2}<20$. And in view of (\ref{C3}), when the points are close to $\sqrt{\theta_1^2+\theta_2^2}=\f15$ and $\sqrt{\theta_1^2+\theta_2^2}=20$, $e$ and $e'$ are ${a}/{2}$. We then extend $m'(\ub, \theta_1', \theta_2')$ as a smooth $2$-covariant $S$ tensorfield to the south pole chart $\sqrt{\theta_1'^2+\theta_2'^2}<20$ and require  
$$e'=\f18 (m'^{-1})^{CA}(\f{\partial m'}{\partial \ub})_{AB} (m'^{-1})^{BD} (\f{\partial m'}{\partial \ub})_{DC}=\f{a}{2} \quad \mbox{ for } \quad \sqrt{\theta_1'^2+\theta_2'^2}\leq \f15.$$

}

\noindent With the above constructions we find $|\chih_0|^2(\ub, \o)$: it is a constant out of two small discs $\{D^1_{\ub}\subset \mathbb{S}^2\}$ centered at $P$ and $\{D^2_{\ub}\subset \mathbb{S}^2\}$ centered at $Q$: 
$$0\leq |\chih_0|^2(\ub,\o)\lesssim a, \quad \mbox{for} \quad \o\in D^1_{\ub}\cup D^2_{\ub},$$
$$|\chih_0|^2(\ub, \o)=a, \quad \mbox{for} \quad \o\in\mathbb{S}^2\setminus \big(D^1_{\ub}\cup D^2_{\ub}\big),$$
$$\iint_{D^1_{\ub}\cup D^2_{\ub}}1\cdot d\o\lesssim \f{1}{c^2}.$$

\begin{minipage}[!t]{0.29\textwidth}
\begin{tikzpicture}[scale=0.9]
    \draw [white](-0.5,1.05)-- node[midway,sloped,above,black]{$Q$}(-0.46,1.05);
     \draw [white](-0.9,-0.8)-- node[midway,sloped,above,black]{$\theta=\f{\pi}{2}$}(-0.5,-0.8);
    \draw [white](0.5,-1.05)-- node[midway,sloped,below,black]{$P$}(0.46,-1.05);
    \draw (-2,0) arc (180:360:2cm and 1cm);
    \draw[dashed] (-2,0) arc (180:0:2cm and 1cm);
    %\draw (0,2) arc (90:270:1cm and 2cm);
    %\draw[dashed] (0,2) arc (90:-90:1cm and 2cm);
    \draw (0,0) circle (2cm);
    \shade[ball color=black!60!white,opacity=0.2] (0,0) circle (2cm);
    \shade[ball color=gray!60, opacity=0.4] (-0.57,0.95) circle (5pt);
    \shade[ball color=gray!80, opacity=0.6] (-0.57,0.95) circle (3pt);
    %\shade[ball color=gray!90, opacity=0.6] (-0.56,1.6) circle (2pt);
    \filldraw (-0.57,0.95) circle (2pt);
    \shade[ball color=gray!60, opacity=0.4] (0.57,-0.95) circle (5pt);
    \shade[ball color=gray!80, opacity=0.6] (0.57,-0.95) circle (3pt);
    \filldraw (0.57,-0.95) circle (2pt);
    %\draw[arrows=->,line width=.4pt](.274,-.5)--(0,0)--(0,2.86);
    \draw[arrows=->](-0.7,0.7)--(-0.39, 0.7);
    \draw[arrows=->](0.7,-0.7)--(0.39, -0.7);
  \end{tikzpicture}
\end{minipage}
\begin{minipage}[!t]{0.7\textwidth}
At the same time, for fixed $\theta$ and fixed $\phi$ we let $\ub$ vary. Then these two discs $D^1_{\ub}$ and $D^2_{\ub}$ rotate (at a speed $2\pi/\d$) along $\theta=\pi/2$.
For $0<\ub\leq \d$, there are at most two intervals of length $\d/c$ such that $0\leq|\chih_0|^2\leq a$, for all the rest $\ub\in(0,\d]$, we have $|\chih_0|^2=a$. This implies that along $H_0^{[0,\d]}$, $(\chih_0)_{ab}(\ub,\o)$ also satisfies 
$$\int_0^{\d}|\chih_0|^2(\ub,\o)d\ub=f(\d,\o)\d a, \,\, \mbox{with} \,\, 1-\f{1}{c}\lesssim f(\d, \o)\lesssim 1+\f{1}{c},$$ 
where $1\ll c \ll b \leq \at$.

\end{minipage}

}

\end{document}